\pdfoutput=1
\RequirePackage{latexml}

\PassOptionsToPackage{
	pagebackref
}{hyperref}

\iflatexml
	\documentclass{article}
	\usepackage{natbib}
	\usepackage[colorlinks]{hyperref}
	\usepackage{amssymb,amsthm,tikz}
	\newtheorem{theorem}{Theorem}
	\newtheorem{lemma}{Lemma}
	\newtheorem{proposition}{Proposition}
	\newtheorem{corollary}{Corollary}
	\newtheorem{definition}{Definition}
	\newtheorem{example}{Example}
	\newtheorem{remark}{Remark}
\else
	
	\documentclass[format=acmsmall, review=false, screen, nonacm]{acmart}
\fi

\usepackage{chrkroer}

\usepackage{booktabs} %
\usepackage[ruled]{algorithm2e} %

\SetAlFnt{\small}
\SetAlCapFnt{\small}
\SetAlCapNameFnt{\small}
\SetAlCapHSkip{0pt}
\IncMargin{-\parindent}

\usepackage{amsmath}
\usepackage{amsthm}
\usepackage{microtype}
\usepackage{thmtools}
\usepackage{thm-restate}
\usepackage[export]{adjustbox}

\usepackage[nameinlink,capitalise]{cleveref}
\newcounter{program}
\crefname{program}{Program}{Programs}
\newenvironment{program}{\crefalias{equation}{program}\begin{equation}}{\end{equation}\ignorespacesafterend}
\iflatexml\else
	\usepackage{crossreftools}
\fi

\usepackage{subcaption}
\usepackage{tikz}
\usetikzlibrary{positioning,decorations.pathreplacing}
\usepackage{pgfplots}
\pgfplotsset{compat=1.18}
\usepackage{xcolor}
\usepackage{nicefrac}
\usepackage{paralist}
\DeclareMathOperator{\relint}{relint}

\usepackage{graphicx} %
\usepackage{multirow} %

\usepackage{wrapfig}
\usepackage{etoolbox}
\makeatletter
\patchcmd\WF@putfigmaybe{\lower\intextsep}{}{}{\fail}%
\AddToHook{env/wrapfigure/begin}{\setlength{\intextsep}{0pt}}
\makeatother

\renewcommand{\epsilon}{\varepsilon}

\setcitestyle{authoryear}

\title[Computing Lindahl Equilibrium for Public Goods with and without Funding Caps]{Computing Lindahl Equilibrium for Public Goods \\ with and without Funding Caps}

\iflatexml
	\author{Christian Kroer \\ Columbia University \and
		Dominik Peters \\ CNRS, LAMSADE, Universit\'e Paris Dauphine - PSL
	}
\else
	\author{Christian Kroer}
	\affiliation{%
		\institution{Columbia University}
		\country{christian.kroer@columbia.edu}
	}
	\email{ck2945@columbia.edu}
	
	\author{Dominik Peters}
	\affiliation{%
		\institution{CNRS, Université Paris Dauphine - PSL}
		\country{dominik.peters@lamsade.dauphine.fr}
	}
	\email{dominik.peters@lamsade.dauphine.fr}
	\authorsaddresses{}
\fi 

\allowdisplaybreaks

\begin{document}

\begin{abstract}
	\emph{Version v3 as accepted by the Journal of the ACM, 2026.} DOI: \href{https://doi.org/10.1145/3848506}{10.1145/3848506}
	\medskip
	
	\noindent
	Lindahl equilibrium is a solution concept for allocating a fixed budget across several divisible public goods. It always lies in the weak core, meaning that the equilibrium allocation satisfies desirable stability and proportional fairness properties.
	We consider a model where agents have separable linear utility functions over the public goods, and the output assigns to each good an amount of spending, summing to at most the available budget.
	
	In the uncapped setting, each of the public goods can absorb any amount of funding. In this case, it is known that Lindahl equilibrium is equivalent to maximizing Nash social welfare, and this allocation can be computed by a public-goods variant of the proportional response dynamics. We introduce a new convex programming formulation for computing this solution and show that it is related to Nash welfare maximization through double duality and reformulation. We then show that the proportional response dynamics is equivalent to running mirror descent on our new formulation, thereby providing a new and very immediate proof of the convergence guarantee for the dynamics.
    Our new formulation has similarities to Shmyrev's convex program for Fisher market equilibrium.
	
	In the capped setting, each public good has an upper bound on the amount of funding it can receive, which is a type of constraint that appears in fractional committee selection and participatory budgeting. In this setting, existence of Lindahl equilibrium was only known via fixed-point arguments. The existence of an efficient algorithm computing one has been a long-standing open question. We prove that our new convex program continues to work when the cap constraints are added, and its optimal solutions are Lindahl equilibria. Thus, we establish that approximate Lindahl equilibrium can be efficiently computed in the capped setting to any desired accuracy. Our result also implies that approximately core-stable allocations can be efficiently computed for the class of separable piecewise-linear concave (SPLC) utilities. 
\end{abstract}

\iflatexml\else
\enlargethispage{30pt}
\addtocontents{toc}{\protect\setcounter{tocdepth}{-1}} %
\maketitle

{\tableofcontents}
\addtocontents{toc}{\protect\setcounter{tocdepth}{2}}
\vspace{-1cm}
\fi

\section{Introduction}

We consider a setting where a fixed budget $B > 0$ needs to be spent on $m$ divisible public goods.
Thus, an outcome is a vector $x = (x_1, \dots, x_m) \in \mathbb{R}^m_{\geq 0}$ summing to at most $B$.
Some of the public goods may additionally have \emph{caps}, i.e., upper bounds on the amount of funding they can receive.
How to distribute the spending across the goods is decided based on the preferences of $n$ agents.
We will consider agents with separable linear utility functions over the goods. Agents may have heterogeneous weights (which can be interpreted as endowments).
We will study the solution concept of \emph{Lindahl equilibrium}, which is based on a virtual market with personalized prices \citep{foley1970lindahl}. This equilibrium notion is known to lead to allocations that are fair to voters, formalized via the concept of the core from cooperative game theory \citep{fain2016core}.

The classic economics literature on public goods, starting with \citet{samuelson1954pure}, focuses on how to arrive at the socially efficient amount of spending in the face of free-riding incentives.
In contrast, we consider a fixed budget and are mostly concerned with how to divide it between different public goods.
This approach, sometimes called \emph{portioning} or \emph{fair mixing}, has received increasing attention in computer science over recent years \citep[see, e.g.,][]{fain2016core,aziz2020fairmixing,brandl2021distribution,airiau2023portioning}, due to its many concrete applications. These include \emph{participatory budgeting}, a method used by many cities to let residents vote over how the government will spend a fixed part of its budget \citep{rey2023comsocPBsurvey,aziz2021pbsurvey}, and \emph{donation platforms}, where donors can influence the distribution of a fixed matching fund \citep{brandl2022fundingpublicprojects,brandt2024coordinatingcharitabledonations}.
The model also captures \emph{committee elections} (i.e., multiwinner voting) in its fractional version \citep{aziz2023bestofbothworlds,suzuki2024maxflow}, as well as the \emph{cake sharing} problem \citep{bei2024truthful}.
Voting methods for the public goods model can also be used to settle small-stakes issues such as a lecturer letting students vote over the distribution of class time across topics, or a team voting over the frequencies with which they will go to different lunch venues.
Companies and non-profit organizations can use the principles derived in this model to decide how to fairly and efficiently divide resources among units or grantees.

In many of these applications, it is desirable to select an outcome that is \emph{representative} of the voters, in that every agent has an equal influence on the overall spending (or an influence that is proportional to the weight assigned to the agent).
This can be formalized as a group fairness guarantee.
In particular, we can require that the spending distribution lie in the (weak) \emph{core}, which means that no subset of voters can construct an alternative way of spending their endowments in a way that they all prefer.
We know that an outcome in the (weak) core always exists thanks to \citet{foley1970lindahl}, who gave a definition of what he called \emph{Lindahl equilibrium} (because he was inspired by ideas of \citet{lindahl1919just}), proved its existence, and showed that it always lies in the core.
This result was introduced to the computer science literature by \citet{fain2016core}.

In the setting where each public good has no cap on funding (we call this the \emph{uncapped setting}), \citet{fain2016core} showed that Lindahl equilibrium is equivalent to the rule that maximizes Nash social welfare (i.e., the endowment-weighted product of agent utilities). The Nash rule has its roots in the \citet{nash1950bargaining} bargaining solution, and its objective function has attractive mathematical properties such as scale-freeness \citep{moulin2004fair}. The Nash rule as applied to the public goods model had already been discussed earlier and independently of Lindahl equilibrium due to its attractive group fairness properties \citep{bms2005,GuNe14a}.
The connection between Lindahl equilibrium and the Nash rule is convenient since the latter can be efficiently computed via a convex program reminiscent of the classic Eisenberg--Gale program~\citep{eisenberg1959consensus,eisenberg1961aggregation} for computing a market equilibrium for private goods.
In addition, \citet{brandl2022fundingpublicprojects} showed that the Nash rule can be computed by running a simple proportional response dynamics which converges to the Nash outcome. 
They pointed out that the same convex program had been considered in several unrelated contexts such as in the portfolio selection literature, where this dynamics had also been discovered and shown to converge \citep{cover1984algorithm}.
While the dynamics converges rapidly in practice, a formal bound on the speed of convergence had not been established by 2022, with \citet[page 11]{li2018generalconvergencemirror} noting that the ``algorithm [of Cover] possesses a guarantee of convergence but [no] convergence rate.''

In the \emph{capped setting}, where each public good may have a cap on how much funding it can receive, the Nash rule loses its fairness properties and is not equivalent to Lindahl equilibrium. In contrast, Lindahl equilibrium retains its fairness properties, and its existence is known via fixed-point theorems \citep{foley1970lindahl}. However, this existence result only applies to strictly monotonic utility functions and thus does not allow agents to have valuations equal to 0 for some goods, and it does not allow for caps except through approximating them through appropriate `saturating' utility functions \citep{fain2016core,munagala2022coremultilinear}. Most importantly, the existence result is not algorithmic, and how to compute a Lindahl equilibrium was an open question.
\citet{fain2016core} asked: ``Is computing the Lindahl equilibrium for public goods computationally hard or is there a polynomial time algorithm even [when the public goods are capped]?'' 

Since then there has been no progress on this question. Indeed, \citet{jiang2020approximatelystable} again noted that ``we do not know how to compute the Lindahl equilibrium efficiently''.
It was even open how to compute any allocation that lies in the core, not necessarily a Lindahl equilibrium allocation.
\citet{cheng2020groupfairness} noted that ``it is not known how to compute such a core outcome efficiently even for [...] approval set utilities'', and \citet{suzuki2024maxflow} concluded that ``there is no known polynomial time algorithm for computing fractional core''. 
The need for a practical algorithm was particularly pressing in the work of \citet{munagala2022coremultilinear} who studied \emph{indivisible} public goods and were aiming for allocations that are approximately in the core. Their best result is based on rounding a Lindahl equilibrium allocation and ``yields a 9.27-core, though we do not know how to implement the resulting algorithm in polynomial time''. To obtain a polynomial-time result, \citet{munagala2022coremultilinear} needed to avoid Lindahl equilibrium and in that case only achieved a 67.37-approximation.

\subsection{Contributions}

In the uncapped setting, we prove that the proportional response dynamics converges to a Lindahl equilibrium at a rate of $\log(m)/t$. We show this by developing a new convex program, distinct from the standard Eisenberg--Gale-style program for Nash welfare maximization, and then showing that applying mirror descent to this program is equivalent to the proportional response dynamics, thereby allowing us to obtain the convergence rate from known results about mirror descent.%
\footnote{While writing the paper, we became aware that \citet{zhao2023convergence} has recently obtained the same convergence rate bound of $\log(m)/t$. He obtained the convergence rate via a direct first-order analysis of the multiplicative gradient (MG) method. Zhao notes that ``the extraordinary numerical performance of the MG method is rather surprising and somewhat mysterious [because it] is extremely simple''. Our results demystify the performance of the dynamics, by showing that it is equivalent to mirror descent, but on a different convex program.}
We also show that the dynamics converges to an allocation that approximates the proportional fairness condition as well as the core.

\begin{figure}[t]
	\centering
	\iflatexml
	\includegraphics{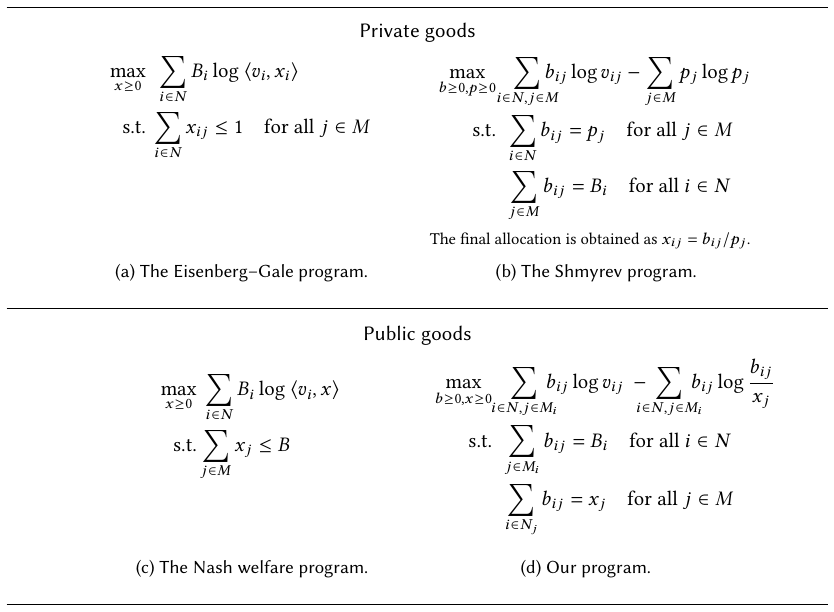}
	\else
	\rule{\linewidth}{0.1ex}\smallskip\\
	\textsf{Private goods} \\
	\begin{subfigure}{0.4\linewidth}
		\begin{align*}
			\max_{x \geq 0} \:&\ \sum_{i\in N} B_i \log\: \langle v_i, x_i \rangle \\
			\text{s.t.} \:& \sum_{i \in N} x_{ij} \leq 1 \quad \text{for all $j \in M$}
		\end{align*}
		\vspace{38pt}
		\caption{The Eisenberg--Gale program.}
	\end{subfigure}
	\qquad
	\begin{subfigure}{0.4\linewidth}
		\begin{align*}
			\max_{b \geq 0, p \geq 0}\:\:&\ \sum_{\mathclap{i\in N, j \in M}} b_{ij} \log v_{ij} - \sum_{j \in M} p_j \log p_j \\
			\text{s.t.}\:\:	& \sum_{j\in M} b_{ij} = B_i \quad \text{for all $i \in N$} \\
			& \sum_{i\in N} b_{ij} = p_j \quad \text{for all $j \in M$}
		\end{align*}
		 {\footnotesize The final allocation is obtained as $x_{ij} = b_{ij} / p_j$.}
		\caption{The Shmyrev program.}
	\end{subfigure} \\
	\rule{\linewidth}{0.1ex}\smallskip\\[1pt]
	\textsf{Public goods} \\
	\begin{subfigure}{0.4\linewidth}
		\begin{align*}
			\max_{x \geq 0}&\ \sum_{i\in N} B_i \log\: \langle v_i, x \rangle \\
			\text{s.t.}& \sum_{j\in M} x_j \leq B \\
			\phantom{\sum_{i \in N}}
		\end{align*}
		\caption{The Nash welfare program.}
	\end{subfigure}
	\qquad
	\begin{subfigure}{0.4\linewidth}
		\begin{align*}
			\max_{b\ge 0, x \ge 0} \:\: & \mathrlap{\displaystyle
				\:\: \sum_{\mathclap{i\in N, j\in M_i}} \: b_{ij} \log v_{ij} 
				\: - \: \sum_{\mathclap{i\in N, j\in M_i}}\: b_{ij} \log \frac{b_{ij}}{x_j}
			} \\
			\text{s.t.}\:\:
			& \sum_{j\in M_i} b_{ij} = B_i \quad \text{for all $i \in N$} \\
			& \sum_{i\in N_j} b_{ij} = x_j \quad \text{for all $j \in M$}
		\end{align*}
		\caption{Our program.}
	\end{subfigure}
	\rule{\linewidth}{0.1ex}
		\fi
	\caption{Comparison of convex programs for computing Fisher market equilibrium (private goods) and Lindahl equilibrium (public goods) for instances with linear utility functions.}
	\label{fig:programs}
\end{figure}

Our new convex program is related to the Eisenberg--Gale-style program through double duality: we show that it can be obtained by taking the dual, introducing new redundant variables, making a change of variable, and performing another dual derivation on this reformulated dual.
The duality and mirror descent relationship that we discover for public goods mirrors existing relationships known in the literature on \emph{private goods} allocation using Fisher market equilibrium. For the private-good setting, equilibrium is also equivalent to maximizing Nash welfare. An alternative convex program for this equilibrium was developed by \citet{shmyrev1983approach,shmyrev2009algorithm}. A proportional response dynamics exists for the private goods case as well~\citep{wu2007proportional,zhang2011proportional}, and \citet{birnbaum2011distributed} showed that it is equivalent to mirror descent on the Shmyrev program. Our new program for the uncapped public goods setting is ``Shmyrev-like'' in its structure. A comparison of these convex programs is shown in \Cref{fig:programs}.
We think it is surprising that it is possible to establish these analogous results, since there are important structural differences in the convex programs. For example, for private goods, the convex programs always have rational optimal solutions, while for public goods they may be irrational. In addition, for private goods, the Eisenberg--Gale program has the same number of variables ($nm$) as the Shmyrev program (after eliminating redundant price variables), while for public goods the Eisenberg--Gale-style program has $m$ variables whereas our new convex program has $nm$ variables. This dimensional difference means that new insights are required for transforming one program into the other through duality.
Finally, the nonlinear term in the objective is a somewhat unusual normalized entropy function in our program, whereas it is the usual entropy on the prices in the private goods case.

In the capped setting, we answer the open problem raised by \citet{fain2016core} positively: $\epsilon$-Lindahl equilibrium can be computed efficiently in the capped setting. 
Indeed, the caps can be naturally added as constraints to our new convex program, and the resulting program correctly computes a Lindahl equilibrium respecting the caps. (In contrast, as is well-known, adding caps as constraints to the Nash welfare program does not lead to a Lindahl equilibrium.) This also establishes the existence of an exact Lindahl equilibrium with caps, without approximations. By analyzing approximate solutions to our convex program, we show that an $\epsilon$-approximate Lindahl equilibrium can be computed in time polynomial in the input size and in $\log(\frac1\epsilon)$ using the ellipsoid method, and it can be computed efficiently in practice by standard solvers with support for exponential cones, such as MOSEK. We present numerical experiments on real-life data from participatory budgeting, showing that solving our program is feasible even for large instances.
Finally, we show how to apply our result to compute approximately core-stable allocations for a broader class of utility functions, namely separable piecewise-linear concave utilities (SPLC). 

\subsection{Related Work}
\label{sec:related-work}

\paragraph{Lindahl equilibrium}
Lindahl equilibrium was introduced by \citet{foley1970lindahl}, who named this equilibrium concept after \citet{lindahl1919just} who put forward related ideas of personalized taxation. However, note that there are other distinct ways of formalizing Lindahl's ideas \citep[see][]{vandennouweland2015}, including ratio and cost share equilibrium \citep{kaneko1977ratio,mas1989costshare}. In this work, we use the Foley definition.

\paragraph{Uncapped setting}
Our interest in Lindahl equilibria is motivated mainly by their proportional fairness properties (notably the core). Such fairness properties have been studied in many related models, notably the ``fair mixing'' or ``portioning'' models  \citep{bms2005,fain2016core,aziz2020fairmixing,brandl2021distribution,airiau2023portioning,gul2020lindahlequilibriumcollectivechoice} that correspond to what we call the uncapped setting. In this setting, Lindahl equilibrium coincides with the maximum Nash welfare solution, which has been axiomatically characterized \citep{GuNe14a} and noted for its strong participation incentives \citep{brandl2022fundingpublicprojects} as well as its lowest-possible price of fairness \citep{michorzewski2020price}. The Nash solution is also well-known to provide fair outcomes in other models, such as for private goods \citep{caragiannis2019unreasonable}.

\paragraph{Capped setting}
What we call the capped setting has also been studied in various special cases under various names, such as cake sharing \citep{bei2024truthful}, fractional committee elections \citep{pierczynski2022core,suzuki2024maxflow}, or divisible participatory budgeting \citep{fain2016core,aziz2021pbsurvey}. These works have mostly not considered Lindahl equilibrium, since there was no known way of computing one.

\paragraph{Discrete models}
In discrete models, the public goods can either be fully funded or not at all. This model captures the way many cities run their participatory budgets, and has thus been well-studied including via core-like fairness notions such as EJR \citep{rey2023comsocPBsurvey,peters2021mes}, that were developed in the large literature on approval-based committee elections \citep{lackner2023abc,aziz2017ejr,peters2025core}. There also exist proposals for definitions of Lindahl equilibrium for discrete models \citep{peters2021market,munagala2022auditing}.

\paragraph{Computation}
In the uncapped setting, the maximum Nash welfare solution (and thus Lindahl equilibrium) can be efficiently computed via an Eisenberg--Gale-style convex program. This program has a simple structure (maximizing a natural objective function over the standard simplex), and \citet{zhao2023convergence} has cataloged its appearance in many unrelated areas, including
portfolio selection for maximizing log investment returns \citep{cover1984algorithm},
information theory \citep{csiszar1974computation} and statistics \citep{vardi1993image},
and in medical imaging for positron emission tomography \citep{vardi1985statistical}.
\citet{cover1984algorithm} proposed a dynamics converging to the optimal solution of this program. Convergence proofs were also given by \citet{csiszar1984information} and \citet{brandl2022fundingpublicprojects}. Later, \citet{zhao2023convergence} obtained a convergence rate of $\log(m)/t$ for this dynamics. This is the same rate that we establish, though his approach does not connect the dynamics to mirror descent.
In the capped setting, very little was known about computation, except for a heuristic algorithm proposed by \citet{fain2016core} that worked well in their experiments.

\paragraph{Donor coordination}
An important application of the public goods allocation problem we study is \emph{donor coordination}, where a collection of donors wish to coordinate their charitable spending, for example by pooling their donations and voting over the division of the pool between different causes. \citet{brandl2022fundingpublicprojects} have argued (using an uncapped model with linear utilities) that the Nash rule and thus Lindahl equilibrium is a good solution concept for this use case \citep[see also][]{greaves2023bargaining}. However, a key reason for coordinating donations is the potential presence of caps: some charities may have a limited ``room for more funding''. This issue is frequently discussed within Effective Altruism, citing cases similar to \Cref{ex:non-unique} \citep{peters2019effectivealtruism}. As our work shows, the Nash solution does not extend well to settings with limited room for more funding, but Lindahl equilibrium as computed by our new convex program does. \citet{brandt2024coordinatingcharitabledonations} also discuss Lindahl equilibrium applied to donor coordination, using Leontief utility functions.

\paragraph{Other applications of the core}
The core has been employed as a fairness property in many other models.
For example, \citet{chaudhury2022federatedlearning} apply it to federated learning.%
\footnote{\citet[Theorem 2]{chaudhury2022federatedlearning} show that the Nash rule satisfies core stability for arbitrary concave utilities. On first sight, this is in contradiction to our claim that Nash fails the core in the capped setting (\Cref{ex:capped-nash-fails-core}). The difference is that \citet{chaudhury2022federatedlearning} use a much weaker notion of the core which involves scaling utilities by the coalition size, rather than scaling the endowment. \citet[Section 1.6]{fain2018indivisiblepublic} discuss the distinction between these concepts.}
It has also been used for clustering \citep{chen2019proportionally,caragiannis2024proportional,kellerhals2024proportional}, peer reviewer assignments \citep{aziz2023peerreview}, and sortition \citep{ebadian2025sortition}.

\section{Setup}
Let $M$ be a set of $m$ \emph{public goods}, which we sometimes refer to as \emph{projects}.  We have an overall \emph{budget} $B > 0$ that we can spend on the public goods.
Let $N = \{1, \dots, n\}$ be a set of $n$ \emph{agents}. 
Each agent $i \in N$ has an individual budget $B_i > 0$ representing $i$'s weight or endowment. These sum to the overall budget, $\sum_{i\in N} B_i = B$. In many applications, the entitlements are equal: $B_i = B/n$.
Each agent $i$ has a valuation $v_{ij} \geq 0$ for each public good $j \in M$. We write $v_i = (v_{ij})_{j \in M}$ for the vector of $i$'s valuations.
The \emph{utility} of an agent $i \in N$ for an outcome $x \in \mathbb{R}^m_{\ge 0}$ is $u_i(x) = \langle v_i, x \rangle = \sum_{j \in M} v_{ij} x_j$. Thus, we use separable linear utilities.
We write $M_i = \{ j \in M : v_{ij} > 0 \}$ for the projects that agent $i \in N$ likes, and we write $N_j = \{ i \in N : v_{ij} > 0 \}$ for the agents that support project $j \in M$. We will assume throughout that $M_i \neq \emptyset$ for all $i \in N$ and that $N_j \neq \emptyset$ for all $j \in M$, since otherwise we may as well remove such agents or public goods from consideration.

In the \emph{uncapped public goods} setting, an \emph{allocation} is a vector $x = (x_j)_{j \in M}$ with $x_j \geq 0$ for all $j \in M$ and $\sum_{j \in M} x_j \le B$. Here, $x_j$ denotes the total spending on project $j$. In this definition, the public goods have no upper bound on how much funding can be spent on them, so in principle the entire budget $B$ could be spent on a single good.

In the \emph{capped public goods} setting, we add the additional constraint that each good $j \in M$ has a maximum amount $\text{cap}_j > 0$ that can be spent on it.
Thus, in this setting, an \emph{allocation} is a vector $x = (x_j)_{j \in M}$ with $0 \le x_j \le \text{cap}_j$ for all $j \in M$ and $\sum_{j \in M} x_j \le B$. 
We assume that $\sum_{j \in M} \text{cap}_j \ge B$ (since that is the interesting case). The uncapped setting is the special case where $\text{cap}_j = \infty$ for each $j \in M$.

\subsection{Lindahl Equilibrium}
\label{sec:lindahl definition}
Our goal is to find a \emph{Lindahl equilibrium}, which is known to yield a fair and efficient allocation of public goods, in the sense that it yields an allocation that is (weakly) Pareto efficient and lies in the (weak) core~\citep{foley1970lindahl,fain2016core}.
Let $p = (p_{ij})_{i \in N, j \in M}$ be a collection of non-negative \emph{personalized prices}, with $p_{ij} \ge 0$ denoting the price that agent $i$ needs to pay per unit of project $j$, and $p_i = (p_{ij})_{j\in M}$ denoting the vector of prices facing $i$.
\begin{definition}[Lindahl Equilibrium]
	\label{def:lindahl-equilibrium}
	Let $x$ be an allocation and let $p$ be a collection of non-negative personalized prices. Then $(x, p)$ is a \emph{Lindahl equilibrium} if 
	\begin{itemize}
		\item $x$ is \emph{affordable}: we have $\langle p_i, x \rangle \le B_i$ for every $i \in N$,
		\item $x$ is \emph{utility-maximizing}: for every $i \in N$ and every $y \in \mathbb{R}_{\ge 0}^m$ such that $0 \le y_j \le \textup{cap}_j$ for all $j \in M$ and such that $\langle p_i, y \rangle\le B_i$, we have $u_i(x) \ge u_i(y)$, 
		\item $x$ is \emph{profit-maximizing}: for every $j \in M$, we have $\sum_{i \in N} p_{ij} \le 1$, and whenever $x_j > 0$ then $\sum_{i \in N} p_{ij} = 1$.
	\end{itemize}
\end{definition}
We say that an allocation $x$ is a \emph{Lindahl equilibrium allocation} if there exist prices $p$ such that $(x,p)$ is a Lindahl equilibrium.

The distinctive property of a Lindahl equilibrium is that prices are personalized, but every agent demands the exact same bundle $x$ of public goods. That is the content of the utility maximization condition: it says that every agent $i$ can afford $x$ given the prices $(p_{ij})_{j \in M}$ and $i$'s budget $B_i$, and prefers $x$ among all affordable allocations satisfying the cap constraint. 

The interpretation of the profit maximization condition is less clear. Its most important effect is that it imposes some amount of efficiency: an equilibrium can only spend a positive amount of budget on projects that have the maximum total price (and generally prices are higher if agent valuations for the project are higher). The condition can be seen as ``profit maximization'' if we imagine that there is a central producer of the public goods who takes in money from the agents and produces the public goods (at a cost of 1 unit of money for 1 unit of public good). This interpretation is made formal in the following simple observation.
\begin{lemma}
	[Equivalent definitions of profit maximization]
	Let $x$ be an allocation and let $p$ be a collection of non-negative personalized prices. Then the following are equivalent:
	\begin{itemize}
		\item[(a)] for every $j \in M$, we have $\sum_{i \in N} p_{ij} \le 1$, and whenever $x_j > 0$ then $\sum_{i \in N} p_{ij} = 1$,
		\item[(b)] for every $y \in \mathbb{R}^m_{\ge 0}$, we have 
		\[\sum_{j\in M}\sum_{i \in N} p_{ij}x_j - \sum_{j\in M} x_j \ge \sum_{j \in M}\sum_{i\in N} p_{ij}y_j - \sum_{j \in M} y_j.\]
	\end{itemize}
\end{lemma}
\begin{proof}
	If the prices satisfy (a), then $\sum_{j\in M}\sum_{i \in N} p_{ij}x_j - \sum_{j\in M} x_j = \sum_{j\in M} x_j - \sum_{j\in M} x_j = 0$ and for every  $y \in \mathbb{R}^m_{\ge 0}$, we have $\sum_{j \in M}\sum_{i\in N} p_{ij}y_j - \sum_{j \in M} y_j \le \sum_{j \in M} y_j - \sum_{j \in M} y_j = 0$, establishing (b).
	
	If the prices satisfy (b), but there is some $j \in M$ with $\sum_{i \in N} p_{ij} > 1$, then the profit attained by $y \in \mathbb{R}^m_{\ge 0}$ is unbounded as $y_j \to \infty$, a contradiction. If the prices satisfy (b), but there is some $j \in M$ with $x_j > 0$ but $\sum_{i \in N} p_{ij} < 1$, then taking $y$ to be identical to $x$ but with $y_j = 0$ gives larger profit, a contradiction. These two contradictions establish (a).
\end{proof}
Note that in condition (b), the producer compares $x$ to every other possible vector $y \in \mathbb{R}^m_{\ge 0}$, even if $y$ violates the cap-constraints or the overall budget constraint. Condition (b) is usually used as part of the definition of Lindahl equilibrium, but we have used condition (a) in our definition because it is simpler and more useful in proofs.

\begin{example}[Personal projects]
	\label{ex:personal-projects}
	Consider the uncapped setting, and suppose that each agent likes exactly one project that nobody else likes, so we have $N = M$, with $v_{ii} = 1$ for each $i \in N$ and $v_{ij} = 0$ for all $i \neq j$.
	Let $(x, p)$ be a Lindahl equilibrium. Note first that $p_{ii} > 0$, since if $p_{ii} = 0$ then there would be no utility-maximizing bundle for $i$ who would demand unbounded spending on project $i$.
	Now, for each $i \in N$, the utility maximizing bundles are those that place the highest possible amount on project $i$, given the prices and subject to affordability, which is $B_i/p_{ii}$. Hence, in Lindahl equilibrium $x_i = B_i/p_{ii}$. In particular $x_i > 0$ for each $i$, and the entire endowment of $i$ is spent on project $i$ since $p_{ii}x_i = B_i$. Thus, in order for $x$ to be affordable for $i$, it must be the case that $p_{ij} = 0$ whenever $j \neq i$ (because $x_j > 0$). By profit maximization, since $x_i > 0$, we get that $p_{ii} = 1$. Thus, $B_i/x_i = p_{ii} = 1$ and so $x_i = B_i$. This shows that there is a unique Lindahl equilibrium allocation $x$ with $x_i = B_i$ for each $i \in N$.
 \end{example}
 
Every Lindahl equilibrium $(x,p)$ can be \emph{decomposed}: For each $i \in N$ and $j \in M$, write
\[
b_{ij} = p_{ij}x_j
\]
for the \emph{contribution} of $i$ towards $j$.
This is a decomposition of $x$ (similar to a notion considered by \citet[Definition 2]{brandl2022fundingpublicprojects}) because the values $(b_{ij})_{ij}$ satisfy the following conditions:
\begin{itemize}
	\item For each $j \in M$, we have $x_j = \sum_{i \in N} b_{ij}$. (This is trivial if $x_j = 0$ and if $x_j > 0$ it follows because $\sum_{i \in N} p_{ij} = 1$.)
	\item For each $i \in N$, we have $\sum_{j \in M} b_{ij} \le B_i$. (This is simply a restatement of the affordability condition of the definition of Lindahl equilibrium.)
\end{itemize}
With this interpretation, we can see that $p_{ij}$ equals the fraction of spending on project $j$ that is contributed by agent $i$. This interpretation also appears in the definition of \emph{ratio equilibrium} \citep{kaneko1977ratio} which is equivalent to Foley's Lindahl equilibrium in the simple model we consider: We take the spending on a public good to be the same as the amount of the public good that is provided, which implies constant returns to scale, where several public goods equilibrium notions coincide \citep[see also][]{moore2006generalequilibrium,vandennouweland2015,mas1989costshare}.

\citet{foley1970lindahl} proved the existence of Lindahl equilibrium using a fixed-point theorem, in a model that is more general than ours. However, his result only applies to strictly monotonic preferences, and thus only establishes existence when $v_{ij} > 0$ for all $i \in N$ and $j \in M$. We will allow $v_{ij} = 0$. In the presence of zeros, it makes sense to consider Lindahl equilibria $(x,p)$ that are what we call zero-respecting.
\begin{definition}[Zero-respecting]
	A Lindahl equilibrium $(x,p)$ is \emph{zero-respecting} if for all $i \in N$ and $j \in M$, whenever  $v_{ij} = 0$ and $x_j > 0$ then $p_{ij} = 0$.
\end{definition}
This is a natural condition in view of the decomposition we considered above, because in a zero-respecting Lindahl equilibrium, an agent contributes only to projects with positive utility: if $v_{ij} = 0$ then $b_{ij} = 0$. This condition is also imposed in the decomposability condition of \citet[Definition 2]{brandl2022fundingpublicprojects}.

The following example shows that not every Lindahl equilibrium is zero-respecting, and that zero-respecting Lindahl equilibria may violate Pareto efficiency. This will motivate imposing a certain sufficient condition introduced below that will avoid this result.

\begin{example}[Lindahl equilibrium may underspend]
	\label{ex:zero-respecting Lindahl equilibrium underspends}
	Consider the following instance:
	\begin{center}
		{\upshape
			\begin{tabular}{lccc}
				\toprule
				& $B_i$ & Project 1 & Project 2 \\
				\midrule
				Agent 1 & $0.5$ & $1$ & $0$ \\
				Agent 2 & $0.5$ & $0$ & $1$ \\
				\midrule
				$\text{cap}_j$ & & $0.25$ & $\infty$ \\
				\bottomrule
		\end{tabular}}
	\end{center}
	On this instance, the unique zero-respecting Lindahl equilibrium allocation is $x = (0.25, 0.5)$. To see this, note that each agent will demand the project that the agent likes, no matter the prices. Thus $x_1, x_2 > 0$. By profit maximization and the zero-respecting condition, we have $p_{11} = p_{22} = 1$ and $p_{12} = p_{21} = 0$. Then by the affordability and utility maximization conditions of Lindahl equilibrium, we get $x = (0.25, 0.5)$. 
	Note that the total spending in this instance is $0.75$, strictly less than the available budget of $B = 0.5 + 0.5$. In particular, $x$ is Pareto-dominated by the allocation $y = (0.25, 0.75)$.
	
	If we remove the zero-respecting condition, there exist other Lindahl equilibria. In particular, $x' = (0.25, 0.75)$ forms an equilibrium with the prices $p_1 = (1, \frac13)$ and $p_2 = (0, \frac23)$.%
	\footnote{Suppose we replace zero-valuations by $\epsilon > 0$, i.e., we set $v_{12} = v_{21} = \epsilon$. Then for all $\epsilon > 0$, every Lindahl equilibrium allocation has $x_2 \ge 0.75$ by \Cref{cor:pareto} (Pareto optimality). But for $\epsilon = 0$, in a zero-respecting Lindahl equilibrium, $x_2 = 0.5$. Thus, Lindahl equilibrium does not necessarily converge to a zero-respecting Lindahl equilibrium as $\epsilon \to 0$.%
	\refstepcounter{footnote}\addtocounter{footnote}{-1}\label{fn:discontinuous}%
	}
	
\end{example}

Note that the formal model of \citet{foley1970lindahl} does not directly support caps, but these can be simulated $\epsilon$-approximately through concave utility functions \citep[Footnote 2]{munagala2022coremultilinear}. We will be able to handle caps without any approximations.

\subsection{Pareto-Optimality and the Core}
Next we discuss how the Lindahl equilibrium relates to Pareto optimality and the set of allocations that are in the core. In the uncapped setting with strictly increasing valuations, the relationship between these concepts is straightforward, and was already studied by \citet{foley1970lindahl}. However, as we shall see, there is more nuance in the capped setting and in the presence of valuations equal to 0. We begin by introducing a sufficient condition that excludes examples like \Cref{ex:zero-respecting Lindahl equilibrium underspends} where intuitively the caps of projects that receive non-zero valuations are too low. We will see that under this sufficient condition, every zero-respecting Lindahl equilibrium spends the entire budget, is Pareto efficient, and lies in the core.

For every $i \in N$, write $F_i = \{ f \in N \mid M_i \cap M_f \neq \emptyset \}$ for the set of ``friends'' of $i$ who agree that at least one common project has a positive valuation. Note that $i \in F_i$ since $M_i \neq \emptyset$ by assumption.
\begin{definition}
	\label{def:cap-sufficient}
	An instance is \emph{cap-sufficient} if we have $\sum_{j \in M_i} \textup{cap}_j \ge \sum_{f \in F_i} B_f$ for all $i \in N$.
\end{definition}

There are many interesting settings in which instances are always cap-sufficient, including:
\begin{itemize}
	\item The uncapped setting where $\text{cap}_j = +\infty$ for all $j \in M$.
	\item All valuations are positive: $v_{ij} > 0$ for all $i \in N$ and $j \in M$. (Proof: In this case, $M_i = M$ and $F_i = N$, so the cap-sufficiency condition is implied by our general assumption that $\sum_{j\in M} \text{cap}_j \ge B$.)
	\item Each agent has positive utility for goods whose total cap reaches the budget: $\sum_{j \in M_i} \text{cap}_j \ge B$.
\end{itemize}

We will show that Lindahl equilibrium has particularly desirable properties on cap-sufficient instances. A key consequence of cap-sufficiency is that every agent spends their entire budget.

\begin{proposition}
	\label{prop:cap-sufficient implications}
	On a cap-sufficient instance, if $(x, p)$ is a zero-respecting Lindahl equilibrium, then
	\begin{enumerate}
		\item[(i)] for every $i \in N$, we have $\langle p_i, x \rangle = B_i$,
		\item[(ii)] we have $\sum_{j \in M} x_j = B$,
		\item[(iii)] for every $i \in N$ and every $y \in \mathbb{R}_{\ge 0}^m$ such that $0 \le y_j \le \textup{cap}_j$ for all $j \in M$ and such that $\langle v_i, y \rangle \ge \langle v_i, x \rangle$, we have $\langle p_i, y \rangle \ge B_i$.
	\end{enumerate}
\end{proposition}
\begin{proof}
	(i) Suppose for a contradiction that $\langle p_i, x \rangle < B_i$. We claim that then for all $j \in M_i$, we have $x_j = \text{cap}_j$: Otherwise, if $x_j < \text{cap}_j$, we can increase $x_j$, thereby increasing the utility of $i$, and a sufficiently small increase is affordable since $i$ does not spend all of $B_i$ and satisfies the cap constraints. This contradicts utility maximization.

	Now, because $(x, p)$ is zero-respecting, for each $j \in M_i$, only friends of $i$ will contribute to $j$ because $N_j \subseteq F_i$. Thus, $b_{i'j} = 0$ if $i' \in N \setminus F_i$. But then
	\[ \sum_{j \in M_i} \text{cap}_j  = \sum_{j \in M_i} x_j = \sum_{j \in M_i} \sum_{f \in F_i} b_{fj} \le \sum_{f \in F_i} \langle p_f, x \rangle < \sum_{f \in F_i} B_f ,  \]
	since $\langle p_i, x \rangle < B_i$ by assumption. This contradicts that the instance is cap-sufficient.
	
	(ii) Using (i), we deduce that
	\[
	\sum_{j \in M} x_j = \langle 1, x \rangle = \sum_{j \in M} \sum_{i \in N} p_{ij} x_j = \sum_{i \in N} \sum_{j \in M} p_{ij} x_j = \sum_{i \in N} B_i = B.
	\]
	The second equality follows by distinguishing the cases $x_j > 0$ (use profit maximization) and $x_j = 0$ (does not contribute to the sum).
	
	(iii) For a contradiction, suppose there is $i \in N$ and a cap-respecting allocation $y$ such that $\langle v_i, y \rangle \ge \langle v_i, x \rangle$ but with $\langle p_i, y \rangle < B_i$. Due to (i), we have $\langle p_i, y \rangle < \langle p_i, x \rangle$. Thus, there exists some $j \in M$ such that $p_{ij} > 0$ and $y_j < x_j$. Since $(x,p)$ is zero-respecting, we have $v_{ij} > 0$. Now consider an allocation obtained from $y$ but with the $j$-coordinate increased by a small amount. For a small enough increase, the resulting allocation respects the cap-constraints (because $y_j < x_j \le \text{cap}_j$) and is affordable for~$i$ (because $\langle p_i, y \rangle < B_i$), but gives $i$ strictly higher utility, contradicting the utility-maximization condition of Lindahl equilibrium. (Note that the constructed allocation may not respect the overall budget constraint, but that is not required by the definition of Lindahl equilibrium.)
\end{proof}

A major reason to be interested in Lindahl equilibrium is that it always lies in the weak core, which is a fairness or stability property formalizing proportional representation.

\begin{definition}[Core]
	An allocation $x$ is in the \emph{core} if there is no non-empty ``blocking coalition'' $S \subseteq N$ and no objection $z = (z_j)_{j \in M} \in \mathbb{R}_{\ge 0}^m$ with $0 \le z_j \le \textup{cap}_j$ for all $j \in M$, such that $\sum_{j \in M} z_j \le \sum_{i \in S} B_i$ (it can be afforded by the blocking coalition) and for all $i \in S$, we have $\langle v_i, z \rangle \ge \langle v_i, x \rangle$ (every coalition member weakly prefers the objection) and the inequality is strict for at least one $i \in S$. It is in the \emph{weak core} if there are no such $S$ and $z$ such that $\langle v_i, z \rangle > \langle v_i, x \rangle$ for all $i \in S$.
\end{definition}

\citet[Section 6]{foley1970lindahl} proved that Lindahl equilibrium allocations are in the weak core, though his model implicitly assumed cap-sufficiency. We can more generally show the following.

\begin{proposition}
	\label{prop:lindahl is core}
	Let $(x,p)$ be a Lindahl equilibrium. Then $x$ lies in the weak core. If the instance is cap-sufficient and $(x,p)$ is zero-respecting, then $x$ lies in the core.
\end{proposition}
\begin{proof}
	Suppose not, and suppose $S \subseteq N$ is a blocking coalition with objection $z$ satisfying $\sum_{j \in M} z_j \le \sum_{i \in S} B_i$ and the cap constraints.
	We now claim that $\sum_{i \in S} \langle p_i, z \rangle > \sum_{i \in S} B_i$.
	\begin{itemize}
		\item Under the assumption that $x$ fails the weak core, note that since for every $i \in S$ we have $\langle v_i, z \rangle > \langle v_i, x \rangle$, the utility maximization condition of Lindahl equilibrium implies that $\langle p_i, z \rangle > B_i$. Summing over $i \in S$ establishes the claim.
		\item Under the assumption that $x$ fails the core, that $(x, p)$ is zero-respecting, and that the instance is cap-sufficient, \Cref{prop:cap-sufficient implications}(iii) implies that since $\langle v_i, z \rangle \ge \langle v_i, x \rangle$ for all $i \in S$, we have $\langle p_i, z \rangle \ge B_i$. Since we have $\langle v_i, z \rangle > \langle v_i, x \rangle$ for at least one $i \in S$, we have $\langle p_i, z \rangle > B_i$ for that $i$ due to utility maximization. Again, summing over $i \in S$ establishes the claim.
	\end{itemize}
	
	Combining the claim with the non-negativity of prices and profit maximization, we have
	\begin{align*}
		\sum_{i \in S} B_i 
		< \sum_{i \in S} \langle p_i, z \rangle 
		\le \sum_{i \in N} \langle p_i, z  \rangle 
		\le \langle 1, z \rangle
		 \le \sum_{i \in S} B_i,
	\end{align*}
	a contradiction.
\end{proof}

As a special case, taking $S = N$ in \Cref{prop:lindahl is core}, we see that Lindahl equilibrium allocations are (weakly) Pareto efficient, establishing a version of the First Welfare Theorem.

\begin{definition}[Pareto-optimality]
	An allocation $x$ is \emph{Pareto-optimal} if there is no allocation $x'$ such that $u_i(x') \ge u_i(x)$ for all $i \in N$ and $u_i(x') > u_i(x)$ for some $i \in N$. It is \emph{weakly} Pareto-optimal if there is no $x'$ with $u_i(x') > u_i(x)$ for all $i \in N$.
\end{definition}

\begin{corollary}
	\label{cor:pareto}
	Let $(x,p)$ be a Lindahl equilibrium. Then $x$ is weakly Pareto optimal. If the instance is cap-sufficient and $(x,p)$ is zero-respecting, then $x$ is Pareto optimal.
\end{corollary}

As another special case of the core result, it is worth noting that Lindahl equilibria also give guarantees for individual agents (by considering $S = \{i\}$), leading to an axiom generalizing the \emph{individual fair share} property of \citet{aziz2020fairmixing}.

\begin{corollary}
	Let $x$ be a Lindahl equilibrium allocation, and let $i \in N$. Then for every allocation $z$ with $\sum_{j \in M} z_j \le B_i$ and $z_j \le \textup{cap}_j$ for all $j \in M$, we have $u_i(x) \ge u_i(z)$.
\end{corollary}

In the uncapped setting, this implies that in any Lindahl equilibrium allocation, we have $u_i(x) \ge B_i \cdot \max_j v_{ij}$ for all $i \in N$.

\section{Convex Optimization Background}
This section gives background on convex optimization and the mirror descent algorithm.

\paragraph{Basic definitions.}
Let $f: \mathbb{R}^n \to (-\infty, \infty]$ be a function. 
It is called \emph{proper} if there exists $x \in \mathbb{R}^n$ with $f(x) < \infty$.
Its \emph{subdifferential} at $x \in \mathbb{R}^n$ is $\partial f(x) = \{ g \in \mathbb{R}^n : f(y) \ge f(x) + \langle g, y - x \rangle \text{ for all } y \in \mathbb{R}^n \}$. Elements of $\partial f(x)$ are called \emph{subgradients}.
The \emph{convex conjugate} of $f$ is the function $f^* : \mathbb{R}^n \to [-\infty, \infty]$ defined by $f^*(y) = \sup_{x \in \mathbb{R}^n} (\langle y, x \rangle - f(x))$. We use the conventions $0\log 0 = 0$ and $0 \log \frac00 = 0$, and in this paper $\log$ always denotes the natural logarithm. We write $B \cdot \Delta^m = \{ x \in \mathbb{R}_{\ge 0}^m :  \sum_j x_j = B \}$ for the scaled simplex.

\paragraph{KKT optimality conditions.}
We will use the following version of the Karush--Kuhn--Tucker theorem, which also works for non-differentiable objective functions.
\begin{theorem}[\citealp{ruszczynski2011nonlinear}, Thm 3.34]
	\label{thm:subdiff-kkt}
	Let $x^*$ be an optimal solution to the program
	\[ \textup{minimize } f(x) \textup{ subject to } h_i(x) \le 0 \textup{ for $i = 1, \dots, m$} \]
	where $f : \mathbb{R}^n \to (-\infty, \infty]$ and $h_i : \mathbb{R}^n \to \mathbb{R}$, $i = 1, \dots, m$, are proper convex functions. 
	Assume that $f$ is continuous at some feasible point, and that Slater's constraint qualification is satisfied, so that there is a feasible point $x$ with $h_i(x) < 0$ for $i = 1, \dots, m$.
	Then there exist $\lambda_1, \dots, \lambda_m \ge 0$ such that
	\[
	\textstyle
	0 \in \partial f(x^*) + \sum_{i = 1}^m \lambda_i \partial h_i(x^*)
	\quad
	\text{and}
	\quad
	\lambda_i h_i(x^*) = 0 \text{ for $i = 1, \dots, m$.}
	\]
	Conversely, if $x^*$ satisfies the constraints $h_i(x^*) \le 0$ for $i = 1, \dots, m$ and there exist $\lambda_1, \dots, \lambda_m \ge 0$ satisfying the above conditions, then $x^*$ is a global minimum.	
\end{theorem}
The KKT theorem can be used to characterize the optimal solutions of convex programs. Two such optimization problems will appear repeatedly in our derivations, and so we state them here.
\begin{lemma}
	\label{lem:simple simplex optimizations}
	\begin{enumerate}
		\item[(a)] Let $y \in \mathbb R^n$. Suppose $x^*$ minimizes $\sum_{i = 1}^n x_i \log x_i - \langle y, x \rangle$ subject to $x \in B \cdot \Delta^n$. Then $x^*_i =  B \cdot e^{y_i}/(\sum_{j=1}^n e^{y_j})$ for $i = 1, \dots, n$.
		\item[(b)] Let $y \in \mathbb R^n_{\ge 0} \setminus \{0\}$. Suppose $x^*$ maximizes $\sum_{i = 1}^n y_i \log x_i$ subject to $x \in \Delta^n$. Then $x^*_i =  y_i/(\sum_{j=1}^n y_j)$ for $i = 1, \dots, n$.
	\end{enumerate}
\end{lemma}
\begin{proof}
	For (a), see the book by \citet[Example 3.71]{beck2017firstorder}.
	For (b), let us assume that $y_i > 0$ for all $i$, since for $i$ with $y_i = 0$, it is optimal to set $x_i = 0$, and we can ignore these indexes when optimizing the others. Under this assumption, note that any optimal solution must have $x_i > 0$ for all $i$. Applying the KKT theorem, this means that the multipliers of the non-negativity constraints are 0. Thus, the stationarity condition implies that $0 = -y_i /x_i + \lambda$, where $\lambda \in \mathbb{R}$ is the multiplier for the constraint $\sum_{i} x_i = 1$. Thus, $x_i = y_i/\lambda$. Summing over all $i$, we see that $\lambda = \sum_{i} y_i$, which gives the result.
\end{proof}

\paragraph{Convex programming duality.}
We will use the following recipe for deriving the dual of convex programs with linear constraints. The recipe was explicitly given in \citet{cole2017convex}, though it is also a direct consequence of Fenchel duality \citep[Theorem 31.1]{rockafellar1970}, as we show for completeness.
\begin{theorem}
	[Fenchel duality for linearly constrained convex programs]
	\label{thm:fenchel dual}
	Let $f$ be a proper convex function.
	The following programs are dual:
	\begin{itemize}
		\item $\min_{x} f(x) - \langle c, x \rangle$ subject to $Ax \ge b$,
		\item $\max_{y, z} \langle b, z \rangle - f^*(y)$ subject to $A^Tz = y - c$ and $z \ge 0$.
	\end{itemize}
	In particular, the two programs have the same objective value, provided that there is a point $x$ with $f(x) < \infty$ and $Ax > b$.
\end{theorem}
\begin{proof}
	Consider the concave function
	\[
	h(x) =
	\begin{cases}
		\langle c, x \rangle & \text{if $Ax \ge b$,} \\
		-\infty & \text{otherwise.}
	\end{cases}
	\]
	Its concave conjugate is $h_*(y) =_{\text{def}} \inf_x \langle y, x \rangle - h(x) = \inf_{x : Ax \ge b} \: \langle y, x \rangle - \langle c, x \rangle = \inf_{x : Ax \ge b} \: \langle y-c, x \rangle$.
	Now, applying Fenchel duality and LP duality, we have
	\begin{align*}
		\min_{x: Ax \ge b} f(x) - \langle c, x \rangle
		&= \inf_x \: f(x) - h(x) \\
		&= \sup_y \: h_*(y) - f^*(y) \tag{Fenchel's duality theorem} \\
		&= \sup_y \: \big( \: \inf_{\mathclap{x\, :\, Ax \ge b}}\quad \langle y-c, x \rangle  - f^*(y) \:\big) \\
		&= \sup_y \: \big( \: \sup_{\mathclap{z \ge 0 : A^Tz = y - c}}\quad \langle b, z \rangle - f^*(y) \:\big) \tag{LP duality} \\
		&= \sup_{\mathclap{y, z\,:\, z\ge 0,\, A^Tz = y - c}} \quad \langle b, z \rangle - f^*(y),
	\end{align*}
	showing that the two programs are dual. Fenchel's duality theorem applies provided that $f$ and $h$ have domains whose relative interior intersects \citep[Theorem 31.1]{rockafellar1970}, which follows from the existence of some $x$ with $f(x) < \infty$ and $Ax > b$.
\end{proof}

\paragraph{Mirror descent.}
The mirror descent (MD) algorithm is a first-order method for convex minimization which generalizes projected gradient descent to allow for more general notions of distance. 
Given a convex set $X$ and a convex function $f$, the goal is to minimize $f$ over $X$ via first-order updates. MD relies on a Bregman divergence $D_h(x\|y)$, which is a function that measures the difference between $x$ and $y$. 
The function $D_h$ is constructed from some strongly convex \emph{reference function} $h$ as $D_h(x\|y) = h(x) - h(y) - \langle \nabla h(y), x-y \rangle$.
For example, taking the negative entropy reference function $h(x) = \sum_i x_i \log x_i$, the Bregman divergence becomes the KL divergence, $D_h(x\|y) = \sum_i x_i \log(x_i/y_i)$, for $x,y \in \Delta^n$.
The update rule for MD is 
\begin{equation}
    x^{t+1} = \arg\min_{x\in X} \: \langle \nabla f(x^t), x \rangle + \frac{1}{\eta}D_h(x\|x^t),
    \label{eq:mirror descent}
\end{equation}
where $\eta > 0$ is a stepsize parameter.
There are a variety of convergence results for MD. We will specifically be interested in the case where a special relationship holds between the objective $f$ and the reference function $h$, known as \emph{relative smoothness}.
The function $f$ is said to be 1-smooth relative to the reference function $h$ when it holds for all $x,y\in \relint X$ that
    \[
    f(x) \leq f(y) + \langle \nabla f(y), x-y \rangle + D_h(x\| y).
    \]
The following theorem from \citet{birnbaum2011distributed} shows that when the reference function $h$ is chosen such that relative smoothness holds, the sequence of iterates generated by mirror descent converges at a rate of $O(1/t)$:
\begin{theorem}[\citealp{birnbaum2011distributed}, Theorem 3]
\label{thm:md rel lip}
    Suppose that $f$ is 1-smooth relative to the reference function $h$, and we run mirror descent using $h$ as the distance-generating function with $\eta = 1$. Let $x^*$ be an optimal solution. Then the sequence of iterates generated by mirror descent satisfies:
    \[
    f(x^t) - f(x^*) \leq \frac{D_h(x^*\| x^0)}{t}
    \]
\end{theorem}

\section{Uncapped Public Goods}

We begin by analyzing the uncapped setting. We will first characterize the Lindahl equilibrium prices, which will be helpful for understanding the convex programs we discuss. Let $(x, p)$ be a Lindahl equilibrium. Then $\langle v_i, x\rangle > 0$ for each $i \in N$ by utility maximization (since there exists some affordable bundle that gives $i$ positive utility, as prices are finite). In the uncapped setting, one can show that every Lindahl equilibrium is zero-respecting: if $v_{ij} = 0$ and $x_j > 0$, the utility maximizing bundle for agent $i$ contains a positive amount on project $j$, which is only possible if $p_{ij} = 0$ (if the price was positive, the agent could increase utility by redirecting the spending to a good with positive valuation). Now, each agent will only demand public goods that maximize the ``bang-per-buck'' ratio $v_{ij}/p_{ij}$ among projects with positive valuation. Thus, the quantity $v_{ij}/p_{ij}$ must be equal for all projects $j$ with $x_j > 0$ and $v_{ij} > 0$. Therefore, $p_{ij} \propto v_{ij}$, say with factor of proportionality $\alpha_i$. Because each agent spends their entire budget (\Cref{prop:cap-sufficient implications} (i), which applies since in the uncapped setting Lindahl equilibrium is always zero-respecting), we have $B_i = \sum_{j \in M_i} p_{ij} x_j = \sum_{j \in M_i} \alpha_i v_{ij} x_j = \alpha_i \langle v_i, x\rangle$ and thus $\alpha_i = B_i/\langle v_i, x\rangle$. Hence, in the uncapped setting,
\begin{equation}
	\label{eq:uncapped prices}
	p_{ij} = B_i \cdot \frac{v_{ij}}{\langle v_i, x \rangle} \quad \text{for all $j \in M$ with $x_j > 0$}.
\end{equation}
(Note that when $v_{ij} = 0$ and $x_j > 0$, \eqref{eq:uncapped prices} just says $p_{ij} = 0$ which follows because $(x,p)$ is zero-respecting, so \eqref{eq:uncapped prices} also holds for $v_{ij} = 0$.)

Now consider a project $j \in M$ with $x_j = 0$. Because $i$ does not demand it, its bang-per-buck must be weakly below the bang-per-buck of funded projects. From \eqref{eq:uncapped prices}, it then follows that
\begin{equation}
	\label{eq:uncapped prices for unfunded projects}
	p_{ij} \ge B_i \cdot \frac{v_{ij}}{\langle v_i, x \rangle} \quad \text{for all $j \in M$ with $x_j = 0$}.
\end{equation}

As we explained in \Cref{sec:lindahl definition}, any Lindahl equilibrium can be decomposed into individual contributions $b_{ij} = p_{ij} x_j$. From \eqref{eq:uncapped prices}, it follows that in Lindahl equilibrium, 
\begin{equation}
	\label{eq:bij condition}
	b_{ij} = B_i \cdot \frac{v_{ij} x_j}{\langle v_i, x \rangle} \quad \text{for all $j \in M$},
\end{equation}
or more simply that $b_{ij} \propto v_{ij} x_j$, so contributions are proportional to the utility $i$ obtains in $x$ from $j$. One can view \eqref{eq:bij condition} as a kind of fixed-point property implied by Lindahl equilibrium \citep{GuNe14a}, and it suggests the proportional response dynamics that we will study later.

\subsection{Nash Welfare and the Eisenberg--Gale Program}
In the uncapped setting (i.e., $\text{cap}_j = +\infty$ for all $j \in M$), Lindahl equilibrium allocations can be nicely characterized as those maximizing the Nash social welfare $\prod_i u_i(x)^{B_i}$ \citep{fain2016core}. Such an allocation can be computed by solving the following convex program:
\begin{program}
\label{eq:EG}
    \begin{aligned}
        \max_{x \geq 0}&\ \sum_{i\in N} B_i \log\: \langle v_i, x \rangle \\
        \text{s.t.}& \sum_{j\in M} x_j \leq B
    \end{aligned}
\end{program}
This program is the public-goods analog of the Eisenberg--Gale convex program for computing a Fisher market equilibrium with private goods~\citep{eisenberg1959consensus,eisenberg1961aggregation}.
By analyzing the KKT conditions of \Cref{eq:EG}, one can show that it exactly computes Lindahl equilibrium. We include a proof for completeness.

\begin{theorem}[\citealp{fain2016core}, Corollary 2.3]
	\label{thm:nash-is-lindahl}
	In the uncapped setting, an allocation $x$ is a Lindahl equilibrium allocation if and only if it is an optimal solution to \Cref{eq:EG}.
\end{theorem}
\begin{proof}
	Let $x$ be an optimal solution to \Cref{eq:EG}. Note first that optimality implies $\sum_j x_j = B$ since otherwise we could increase the objective by increasing $x_j$ for some project liked by some agent (except in the trivial case where $v_{ij} = 0$ for all $i$ and $j$ for which there is nothing to prove). Further, every agent has strictly positive utility at $x$ since otherwise the objective value would be $-\infty$. Thus, the objective function is differentiable at $x$. By \Cref{thm:subdiff-kkt}, there exists $\lambda \ge 0$ (corresponding to the budget constraint) and $(\mu_j)_{j \in M} \ge 0$ (corresponding to the non-negativity constraints) such that for every $j \in M$, we have $0 = -\sum_{i \in N} B_i  \frac{v_{ij}}{\langle v_i, x \rangle} + \lambda - \mu_j$. Thus $\lambda \ge  \sum_{i \in N} B_i  \frac{v_{ij}}{\langle v_i, x \rangle}$, with equality if $x_j > 0$. Multiplying by $x_j$, summing over $j$, and rearranging, we get $\sum_{j \in M} \lambda x_j = \sum_{j \in M}\sum_{i \in N} B_i  \frac{v_{ij}x_j}{\langle v_i, x \rangle}$. This simplifies to $\lambda B = \sum_{i \in N} B_i = B$, so $\lambda = 1$. Hence for each $j \in M$, we have
	\[
	\sum_{i \in N} B_i  \frac{v_{ij}}{\langle v_i, x \rangle} \le 1 \text{, with equality if $x_j > 0$.}
	\]
	Set $p_{ij} = B_i  \frac{v_{ij}}{\langle v_i, x \rangle}$. Then $(x,p)$ is a Lindahl equilibrium: The above inequality immediately establishes the profit maximization condition. For utility maximization, note that the ``bang-per-buck'' of project $j$ to agent $i$ is $v_{ij} / p_{ij} = \langle v_i, x \rangle / B_i$, which is constant, so that all allocations that use up the agent's entire budget are utility maximizing. Since $\langle p_i, x \rangle = B_i$, it follows that $x$ is utility maximizing for $i$.
	
	Conversely, suppose $(x, p)$ is a Lindahl equilibrium. Set $\lambda = 1$ and $\mu_j = 1 - \sum_{i\in N} B_i  \frac{v_{ij}}{\langle v_i, x \rangle}$ for each $j \in M$. For complementary slackness, note that if $x_j > 0$, then from \eqref{eq:uncapped prices} we have $\mu_j = 1 - \sum_{i \in N_j} p_{ij}$, and thus by the profit maximization condition of Lindahl equilibrium, we get $\mu_j = 0$. Complementary slackness also holds for $\lambda$ since $\sum_{j \in M} x_j = B$ by \Cref{prop:cap-sufficient implications}(ii). We also have $\mu_j \ge 0$, combining \eqref{eq:uncapped prices}, \eqref{eq:uncapped prices for unfunded projects} and the profit maximization condition of Lindahl equilibrium. Finally, it is easy to check stationarity; for every $j \in M$ we have
	\[
	-\sum_{i \in N} B_i  \frac{v_{ij}}{\langle v_i, x \rangle} + \lambda - \mu_j = -\sum_{i \in N} B_i  \frac{v_{ij}}{\langle v_i, x \rangle} + \sum_{i\in N} B_i  \frac{v_{ij}}{\langle v_i, x \rangle} = 0.
	\]
	Thus, by \Cref{thm:subdiff-kkt}, $x$ is an optimal solution to \Cref{eq:EG}.
\end{proof}
\citet[Theorem 2.2]{fain2016core} also present Eisenberg--Gale-style programs for computing Lindahl equilibria for certain non-linear utility functions called ``scalar separable non-satiating'' including CES and Cobb-Douglas utilities.

Interestingly, for Fisher market equilibrium, the Eisenberg--Gale program always admits a rational solution \citep{devanur2008market,vazirani2012rationalconvexprogram}. However, this is not the case in our public goods setting,%
\footnote{For Fisher markets, the proof sets up a system of linear inequalities whose variables correspond to the reciprocals of equilibrium prices, $1/p_j$. However, profit maximization in Lindahl equilibrium (which has no analog in Fisher markets) involves the sum $\sum_i p_{ij}$ which is not linear in the reciprocals of prices. Thus, the Fisher market argument does not generalize.}
as the following example shows \citep[see also][Theorem 5]{airiau2023portioning}.

 \begin{example}[Irrational Lindahl equilibrium allocation]
 	\label{ex:irrational}
	Consider the uncapped setting with 4 agents with equal budgets $B_i = \frac14$ and with three projects. The agents have the following valuations:
    \begin{center}
        {\upshape
            \begin{tabular}{lcccc}
                \toprule
                & $B_i$ & Project 1 & Project 2 & Project 3 \\
                \midrule
                Agent 1 & $0.25$ & $1$ & $0$ & $0$ \\
                Agent 2 & $0.25$ & $1$ & $0$ & $1$ \\
                Agent 3 & $0.25$ & $1$ & $1$ & $0$ \\
                Agent 4 & $0.25$ & $0$ & $1$ & $1$ \\
                \midrule
                $\text{cap}_j$ & & $\infty$ & $\infty$ & $\infty$ \\
                \bottomrule
        \end{tabular}}
    \end{center}
	By \Cref{thm:nash-is-lindahl}, a Lindahl equilibrium allocation $x$ forms an optimal solution to \Cref{eq:EG}. Since projects $2$ and $3$ are symmetric and the objective function of \Cref{eq:EG} is strictly concave, we have $x_2 = x_3$. Since $B = 1$, we deduce that $x_1 = 1 - 2x_2$. Thus, the objective function of \Cref{eq:EG} can be written as $\frac{1}{4}(\log(1 - 2x_2) + 2\log(1-x_2) + \log(2x_2))$. Exponentiating, this is equivalent to maximizing $(1 - 2x_2)(1-x_2)^2(2x_2)$. Setting its derivative to $0$ and doing some algebra, we find that $x_2$ needs to solve $8(x_2)^2-7x_2+1=0$, whose only solution in $[0,\frac12]$ is at $x_2 = \frac{1}{16}(7 - \sqrt{17}) \approx 0.1798$. Thus, $x$ is irrational and the unique Lindahl equilibrium allocation.
\end{example}

\subsection{A New Convex Program}

We will present a new convex program which also captures the Lindahl equilibrium concept in the uncapped setting. As we will see, this convex program will yield several useful results. First, we will use it to show that the proportional response dynamics for uncapped public goods can indeed be interpreted as mirror descent with the entropy distance, just as in the Fisher market setting. Secondly, extending this convex program will allow us to give the first computational results for the capped public goods setting.
Our new convex program is in the spirit of the Shmyrev convex program for Fisher markets for private goods~\citep{shmyrev1983approach,shmyrev2009algorithm}, though there are important differences.
The convex program is as follows, written in a form that maximizes a concave function:
\begin{program}
    \begin{aligned}
        \max_{b\ge 0, x \ge 0} \quad & \mathrlap{\displaystyle
        	\:\: \sum_{\mathclap{i\in N, j\in M_i}} \: b_{ij} \log v_{ij} 
        	\: - \: \sum_{\mathclap{i\in N, j\in M_i}}\: b_{ij} \log\left(b_{ij} / x_j\right)
        } \\
        \text{s.t.}\quad
        & \sum_{j\in M_i} b_{ij} = B_i, \qquad &&\forall i\in N \\
        & \sum_{i\in N_j} b_{ij} = x_j, &&\forall j\in M
        &\hspace{0.5cm}
    \end{aligned}
    \label{eq:shmyrev cp}
\end{program}
The objective function of \Cref{eq:shmyrev cp} is concave because the \emph{relative entropy} function $g(x,t) = -t \log(x/t)$, $x, t \ge 0$, is known to be convex \citep[Example 3.19]{boyd2004convex},%
\footnote{\citet{boyd2004convex} do not consider the case where $x = t = 0$, but setting $g(0,0) = 0$ (following our convention $0 \log\frac00 = 0$) preserves convexity, because $g((1-\lambda)0 + \lambda x,(1 - \lambda)0 + \lambda t) = g(\lambda x, \lambda t) = -\lambda t \log(x/t) = (1 - \lambda)0 + \lambda g(x,t)$.}
which follows from the general fact that if $f(x)$ is convex then its so-called \emph{perspective function} $(x,t) \mapsto tf(x/t)$ is also convex \citep[Section 2.2]{hiriart2013convex}, and $g$ is the perspective function of $-\log(x)$. Since $g$ is convex, $-g$ is concave. Note that $-g(x,t) = -t \log(t/x)$. The objective function of \Cref{eq:shmyrev cp} is the sum of a linear function (the first term) and a sum of concave functions $-g(x_j, b_{ij})$ (the second term), hence concave. The same result can also be deduced from the \emph{log sum inequality} \citep[Section 2.7]{coverjoy2005elements}.

The program has two sets of variables, though one is implied by the other. The $x_j$ variables have the same interpretation as in \Cref{eq:EG}: $x_j$ is the amount of budget allocated to project $j$.
The $b_{ij}$ variables can be interpreted as the amount of agent $i$'s budget $B_i$ that they allocate towards project $j$. 
Note that each $x_j$ variable is directly implied by the choice of the $b_{ij}$ variables across agents $i$. It is only there as a convenience variable.

To gain some intuition for \Cref{eq:shmyrev cp}, suppose that we already knew the optimum value of the $x_j$ variables, and thus can treat them as constants and use \Cref{eq:shmyrev cp} to merely compute the values of the $b_{ij}$ variables. If we drop the bottom constraint, the program decomposes into a separate optimization for each $i \in N$. From \Cref{lem:simple simplex optimizations}(a), these optimum values satisfy $b_{ij}^* \propto v_{ij} x_j$, which exactly matches the condition \eqref{eq:bij condition} that we derived earlier from the definition of Lindahl equilibrium.

Since the values of the $x_j$ variables are determined by the values of the $b_{ij}$ variables, we can replace each occurrence of $x_j$ in \Cref{eq:shmyrev cp} by $\sum_{i\in N_j} b_{ij}$. Indeed, let us use the shorthand notation 
\[
	x_j(b) = \sum_{i\in N_j} b_{ij} \quad \text{for each $j \in M$.}
\]	
Then \cref{eq:shmyrev cp} can be rewritten as the following program, where we have also negated the objective function so we have a minimization problem:
\begin{program}
	\begin{aligned}
		\min_{b \ge 0}\quad &f(b) \defeq - \sum_{\mathclap{i\in N, j\in M_i}} \: b_{ij} \left( \log v_{ij} - \log \left(b_{ij} / x_j(b)\right) \right) \\[3pt]
		\text{s.t.}\quad & \sum_{j \in M_i} b_{ij} = B_i,\ \forall i\in N
	\end{aligned}
	\label{eq:shmyrev cp min}
\end{program}
In our proofs, we will mostly work with \cref{eq:shmyrev cp min} and only optimize over the $b_{ij}$ variables. \cref{eq:shmyrev cp min} is a convex program because its objective function is obtained from the negative objective of \cref{eq:shmyrev cp} via a pre-composition with an affine map, namely $b \mapsto ((x_j(b))_{j\in M}, b)$, which preserves convexity \citep[Proposition 2.1.5]{hiriart2013convex}.

While our program has some similarity to the Shmyrev program for private goods~\citep{shmyrev2009algorithm,birnbaum2011distributed}, it has the following important differences.
First, the Shmyrev program contains variables corresponding to prices, which do not appear in our program.
Second, the original primal variables $x_j$ appear directly in our program, whereas in Shmyrev's program these are a non-linear function of the corresponding $b$ variables.
Third, we have a somewhat unusual term that looks like a partially-normalized entropy in our objective, whereas Shmyrev's program only requires using a typical negative entropy term over prices.

\subsection{Connecting the Nash Social Welfare Program to Our New Program via Duality}
\cref{eq:EG} and \cref{eq:shmyrev cp min} can be related to each other through ``double duality''. We need the following lemma, which derives the convex conjugate of a convex function that will appear in the dual of \Cref{eq:EG}. The proof is in \Cref{app:proofs}.
\begin{restatable}{lemma}{conjugatemaxexp}
	For each $j \in M$, let $q_j \in \mathbb R^{N_j}$, and write $q = (q_j)_{j \in M}$.
	For $j \in M$, let $g_j(q_j) = \sum_{i\in N_j} e^{q_{ij}}$ and let $g(q) = B \cdot \max_{j\in M} g_j(q_j)$.
    Then the convex conjugate of $g$ satisfies
    \[
            g^*(\beta) = \sum_{j \in M}\sum_{i \in N_j} \beta_{ij} \log \tfrac{\beta_{ij}}{\tilde x_j(\beta)} - \beta_{ij}
            \quad \text{for all $\beta \ge 0$,}
    \]
    where $\tilde x_j(\beta) = B \cdot (\sum_{i \in N_j} \beta_{ij}) / (\sum_{j' \in M}\sum_{i \in N_{j'}} \beta_{ij'})$.
    \label{lem:conjugate max exp}
\end{restatable}

Now, we can state the result formalizing the relationship between the two programs: their dual programs are equivalent.
This result is an analog of a result for private goods, where the Shmyrev and the Eisenberg--Gale program also share a dual after reformulation \citep{cole2017convex}.

\begin{theorem}
	[Double Duality]
	\label{thm:dual of dual}
    \cref{eq:shmyrev cp min} is the dual of the dual of \cref{eq:EG}, after reformulation.
\end{theorem}
\begin{proof}
	For convenience, we rewrite \cref{eq:EG} by adding variables encoding agent utilities:
	\begin{program}
		\label{eq:EG with u variables}
		\begin{aligned}
			\min_{x \geq 0, u}&\ -\sum_{i\in N} B_i \log u_i \\
			\text{s.t. }& u_i \leq \langle v_i, x \rangle, \forall i \in N \\
			& \sum_{j\in M} x_j \leq B
		\end{aligned}
	\end{program}
    Using the Fenchel duality in \cref{thm:fenchel dual}, we can compute the dual of \cref{eq:EG with u variables} by noting that the convex conjugate of $f(u_i) = -B_i \log u_i$ is $f^*(\mu_i) = B_i \log B_i - B_i - B_i \log(-\mu_i)$, and introducing dual variables $\beta_i$ for the constraints involving $u_i$, dual variable $\lambda$ for the budget constraint, and dual variables $\eta_{j}$ for the non-negativity constraints on the $x_j$.
    \begin{program}
    	\begin{aligned}
    		\max_{\lambda\geq 0, \beta\geq 0, \eta \ge 0, \mu}&\ -B\lambda - \sum_{i\in N} (B_i \log B_i - B_i - B_i \log(-\mu_i)) \\
    		\text{s.t. } & -\sum_{i\in N} \beta_i v_{ij} + \lambda - \eta_j = 0, \forall j\in M \\
    		& - \beta_i = \mu_i, \forall i \in N
    	\end{aligned}
    	\label{eq:raw dual EG}
    \end{program}
    We can clean up this program by eliminating the $\mu_i$ variables (replacing them by $-\beta_i$), eliminating the $\eta_j$ variables (thereby turning equality into inequality constraints), and dropping constant terms from the objective. Thereby we obtain the following program as the ``official'' dual of \cref{eq:EG}.
    \begin{program}
        \begin{aligned}
            \min_{\lambda\geq 0, \beta\geq 0}&\ B\lambda - \sum_{i\in N} B_i \log \beta_i \\
            \text{s.t. } & \sum_{i\in N} \beta_i v_{ij} \leq \lambda, \forall j\in M
        \end{aligned}
        \label{eq:dual EG}
    \end{program}
    Now, we rewrite \cref{eq:dual EG} by introducing a redundant set of variables $\lambda = (\lambda_{ij})_{j \in M, i \in N_j}$, which represent the ``bid'' $\beta_i v_{ij}$ that agent $i$ makes on project $j$. This gives the following program:
    \begin{program}
        \begin{aligned}
            \min_{\lambda \geq 0, \beta \geq 0}&\ B\lambda - \sum_{i\in N} B_i \log \beta_i \\
            \text{s.t. } & \sum_{i\in N_j} \lambda_{ij} \leq \lambda, \forall j\in M \\
            & \beta_i v_{ij} \leq \lambda_{ij}, \forall i\in N, j\in M_i
        \end{aligned}
        \label{eq:dual EG redundant}
    \end{program}
    (One can show that in optimal solutions to this program, we have $\lambda = 1$, and $\lambda_{ij}$ corresponds to the Lindahl price of $j$ for $i$, and $\beta_i$ corresponds to the reciprocal of the maximum bang-per-buck ratio facing agent $i$. The first constraint corresponds to profit maximization of Lindahl equilibrium.)
    
    Having removed the summation from the second constraint in \cref{eq:dual EG redundant}, we can now apply a logarithm to that constraint and get a separation into individual terms.
    We note that in optimum, the value of $\lambda$ is the maximum over $j \in M$ of the value of $\sum_{i\in N_j} \lambda_{ij}$, so we can replace $\lambda$ in the objective function by a maximum.
    In addition, we perform two changes of variable: we write $\gamma_i = -\log\beta_i$ and $q_{ij} = \log \lambda_{ij}$. Thereby we arrive at the following program:
	    \begin{program}
        \begin{aligned}
            \min_{q, \gamma}&\ B \cdot \left( {\textstyle \max_{j \in M} \sum_{i\in N_j} e^{q_{ij}}} \right) + \sum_{i\in N} B_i \gamma_i \\
            \text{s.t. } & \log v_{ij} \leq q_{ij} + \gamma_i, \forall i\in N, j\in M_i
        \end{aligned}
        \label{eq:dual EG max rewrite}
    \end{program}
    
    Next, we derive the dual of \eqref{eq:dual EG max rewrite}, again using the Fenchel dual from \cref{thm:fenchel dual}. 
    For $j \in M$, let $g_j(q_j) = \sum_{i\in N_j} e^{q_{ij}}$ and let $g(q) = B \cdot \max_{j\in M} g_j(q_j)$.
    Write $b_{ij}$ for the dual variables for the constraints in \cref{eq:dual EG max rewrite}. Then the dual is
    \begin{program}
        \begin{aligned}
            \max_{b \ge 0} &\ \sum_{j \in M}\sum_{i \in N_j} b_{ij} \log v_{ij} - g^*(b) \\
            \text{s.t. } & \sum_{j \in M_i} b_{ij} = B_i,\ \forall i\in N
        \end{aligned}
        \label{eq:dual with conjugate}
    \end{program}
    By \cref{lem:conjugate max exp}, the term $g^*(b)$ equals
	$\sum_{j \in M}\sum_{i \in N_j} b_{ij} \log \tfrac{b_{ij}}{\tilde x_j(b)} - b_{ij}$,
    where we write
    $
      \tilde x_j(b)
      = B \cdot \sum_{i \in N_j} b_{ij}/
                        \sum_{j' \in M}\sum_{i \in N_{j'}} b_{ij'}.
    $
	From the constraints in \cref{eq:dual with conjugate}, we have
	$\sum_{j' \in M}\sum_{i \in N_{j'}} b_{ij'} = \sum_{i \in N}\sum_{j \in M_i} b_{ij} = \sum_{i \in N} B_i = B$,
	so we can simplify to $\tilde x_j(b) = \sum_{i \in N_j} b_{ij} = x_j(b)$.
    After dropping an additive constant $B$ from the objective function this yields the desired \cref{eq:shmyrev cp min}.
\end{proof}

Since the EG program and \cref{eq:shmyrev cp min} are connected via the same dual, we know that the solution to \cref{eq:shmyrev cp min} must imply a solution to the primal EG (through computing the implied dual solutions via KKT conditions, which can easily be done) and thus a Lindahl equilibrium. One can also show directly that \cref{eq:shmyrev cp min} yields Lindahl equilibria. We defer this proof to the section on the capped setting, where we show it for that more general case  (\Cref{thm:lindahl shmyrev capped}).

\subsection{A Possible Path to T\^atonnement for Public Goods}
Finally, we briefly remark that the dual program \eqref{eq:dual EG redundant} can be rewritten in a way that eliminates the variables $\lambda$ and $\beta_i$ and thereby turns it into an unconstrained minimization problem. This yields the following program:
\begin{program}
    \min_{\{\lambda_{ij}\}_{j \in M, i \in N_j}}
    \ B \cdot \left( {\textstyle \max_{j \in M} \sum_{i\in N_j} \lambda_{ij}} \right)
    - \sum_{i\in N} B_i \min_{j \in M_i} \log \tfrac{\lambda_{ij}}{v_{ij}}
    \label{eq:dual EG variables eliminated}
\end{program}
One interesting property of this program is that it has a t\^atonnement-like interpretation. The $\lambda_{ij}$ variables can be viewed as personalized prices offered to each agent $i$ for project $j$. 
In this interpretation, each agent chooses their favorite projects among those minimizing $\lambda_{ij} / v_{ij}$ over $j \in M_i$, i.e., ones that maximize their bang-per-buck, and spends their entire budget $B_i$ on such projects.
Formally, let $y_i \in \Delta^{M_i}$ be such that $y_{ij} > 0$ only when project $j$ minimizes $\lambda_{ij} / v_{ij}$.
Then $B_i\cdot y_{ij}$ specifies how much of their budget agent $i$ allocates to an optimal bang-per-buck project $j$.
Similarly, let $x\in B\cdot \Delta^m$ be such that $x_j > 0$ only if $\sum_{i\in N_j} \lambda_{ij} = \max_{j\in M} \sum_{i\in N_j} \lambda_{ij}$. Then $x$ specifies a budget allocation proposed by the price-setter.
Then we have that a subgradient is any $g = (g_{ij})_{j \in M, i \in N_j}$ such that
\[
    g_{ij} = x_j -  y_{ij} \frac{B_i}{\lambda_{ij}},
\]
for any pair $(x,y)$ satisfying the above conditions.
This subgradient can be interpreted as a measure of discrepancy. The price-setter is proposing a set of per-agent prices $\lambda$ and a corresponding budget allocation $x$. In turn, agent $i$ computes their preferred allocation under $\lambda_i$, where $y_{ij}\frac{B_i}{\lambda_{ij}}$ is the quantity of project $j$ they would demand at price $\lambda_{ij}$. The subgradient $g_{ij}$ is then the discrepancy between the price-setter's proposal and the agent's preferred allocation. It is positive (and thus suggests a decrease in price) if the agent spends less than the proposed allocation on the project; it is negative (and thus suggests an increase in price) if the agent spends more.
The subgradient is zero exactly when the price-setter's proposed allocation is optimal for each agent, meaning the proposed prices support the allocation.
Deriving some form of convergence results for this program would be an interesting direction for future work.

\subsection{Proportional Response Dynamics as Mirror Descent}

It is known that the Lindahl equilibrium for the uncapped public goods setting can be computed by a simple dynamics \citep{brandl2022fundingpublicprojects},
which we call the \emph{proportional response} dynamics in analogy to a similar dynamics for private-good Fisher markets \citep{wu2007proportional,zhang2011proportional}.
At each iteration $t$, the proportional response dynamics have some current budget allocation $x^t = (x_1^t,\ldots,x_m^t)$ summing to $B$ and with $x_j^t > 0$ for all $j \in M$. Let $u_i^t = \langle v_i, x^t \rangle$ be the current utility of agent $i$ under this allocation. Note that $u_i^t > 0$ because $x_j^t > 0$ for all $j$ and we assume that $M_i \neq \emptyset$.
Then the next budget allocation in the dynamics is
\[
x_j^{t+1} = \sum_{i\in N} B_i \frac{x_j^t v_{ij}}{u_i^t}.
\]
This dynamics can be interpreted as each agent $i$ independently deciding how they wish to allocate their share of the budget $B_i$ in the next round. 
Specifically, agent $i$ allocates spending proportional to how much utility each project provided them at round $t$.
This spending allocation matches the property in \eqref{eq:bij condition} we derived earlier from the definition of Lindahl equilibrium.
We will show that the proportional response dynamics is the mirror descent algorithm applied to \cref{eq:shmyrev cp min}.

\begin{theorem}
	[Equivalence of mirror descent and PR dynamics]
    Assume that the mirror descent algorithm on \cref{eq:shmyrev cp min} is initialized at a point $b^0 = (b^0_{ij})_{i \in N, j \in M_i}$ such that $\sum_{j \in M_i} b^0_{ij} = B_i$ and $b^0_{ij} > 0$ for all $i \in N$ and $j\in M_i$, and that the proportional response dynamics is initialized at the corresponding allocation $x^0 = (x_j(b^0))_{j \in M}$.
    Then the proportional response dynamics is equivalent to applying the mirror descent algorithm with the negative entropy reference function and stepsize $\eta=1$ to \cref{eq:shmyrev cp min}.
\end{theorem}
\begin{proof}
    Note first that $f$ is differentiable at any $b$ such that $b_{ij} > 0$ for all $i \in N$ and $j \in M_i$. This holds by assumption for $b^0$, and we will see that if it holds for the initial point then it holds throughout.
    We can rewrite the objective function $f$ as follows:
    \begin{align}
    	\notag
    	f(b) &= - \sum_{\mathclap{i\in N, j\in M_i}} \: b_{ij} \log v_{ij} + \sum_{\mathclap{i\in N, j\in M_i}} \: b_{ij} \log b_{ij} - \sum_{j\in M} x_j(b) \log x_j(b).
	    \intertext{The derivative of the objective in \cref{eq:shmyrev cp min} with respect to $b_{ij}$ is}
    	\label{eq:derivative of objective function}
    	\nabla_{ij} f(b) &= - \log v_{ij} + (\log b_{ij} + 1) - (\log x_j(b) + 1) = -\log v_{ij} + \log(b_{ij} / x_j(b)).
    \end{align}
        
    Let $x_j^t = \sum_{i\in N_j} b_{ij}^t$.
    We now apply the mirror descent update rule (given by \cref{eq:mirror descent}) using the negative entropy reference function $h$ and stepsize $\eta=1$. The induced Bregman divergence $D_h(b_i \| b_i^t)$ is the KL divergence (since $b_i$ and $b_i^t$ both sum to $B_i$). Thus, we get the following update for each $i \in N$, where the arg min is taken over all $b_i = (b_{ij})_{j \in M_i}$ such that $\sum_{j \in M_i} b_{ij} = B_i$:
    \begin{align*}
        b_{i}^{t+1} &= \arg\min_{b_i} \sum_{j\in M_i} b_{ij} \left(-\log v_{ij} + \log(b_{ij}^t / x_j^t)\right) + \sum_{j\in M_i} b_{ij} \log(b_{ij} / b_{ij}^t) \\
        &= \arg\min_{b_i} \sum_{j\in M_i} b_{ij} \left( -\log v_{ij} + \log(b_{ij} / x_j^t) \right).
    \end{align*}
		By \Cref{lem:simple simplex optimizations}(a), for each $i \in N$, the solution to this optimization problem satisfies
	    \[
	        b_{ij}^{t+1} \propto \exp\left(\log (v_{ij}x_j^t) \right) \propto v_{ij}x_j^t.
	    \]
	    In order to get a feasible solution we must normalize the above such that $\sum_{j\in M_i} b_{ij}^{t+1} = B_i$. 
	    Let $u_i^t = \sum_{j\in M_i} v_{ij} x_j^t$, noting that this is positive under our running assumption that $M_i \neq \emptyset$.
	    Applying normalization, we get
	    \begin{align*}
	        b_{ij}^{t+1} = \frac{B_i}{u_i^t} v_{ij}x_j^t,%
	    \end{align*}
	    Now if we sum over $i$ we get the proportional response dynamics. Moreover, we see that  $b_{ij}^{t+1} > 0$ for all $i \in N$ and $j \in M_i$ since $v_{ij} > 0$ and $x_j^t > 0$.
\end{proof}

\subsection{Convergence Rates}

We have shown that the proportional response dynamics is equivalent to mirror descent with unit stepsize. Next we wish to apply the convergence-rate result from \cref{thm:md rel lip}. Thus, we need to show that the objective in \cref{eq:shmyrev cp min} is 1-smooth relative to the entropy reference function. 

\begin{lemma}
	[1-smoothness]
	\label{lem:relative lipschitz}
    The objective function $f$ from \cref{eq:shmyrev cp min} is 1-smooth relative to the reference function $h(b) = \sum_{i\in N, j\in M_i} b_{ij}\log b_{ij}$, that is, for all $a = (a_{ij})_{i \in N, j \in M_i}$ and $b= (b_{ij})_{i \in N, j \in M_i}$ such that $a_{ij}, b_{ij} > 0$ for all $i \in N$ and $j \in M_i$ and such that $\sum_{j \in M_i} a_{ij} = \sum_{j \in M_i} b_{ij} = B_i$ for all $i \in N$, we have
    \[
    f(a) \leq f(b) + \langle \nabla f(b), a-b \rangle + D_h(a\| b).
    \]
\end{lemma}
\begin{proof}
    Note that $f(b)$ and $h(b)$ are differentiable when $b_{ij} > 0$ for all $i \in N$ and $j \in M_i$.
    For $a$ and $b$ as in the lemma statement, using \eqref{eq:derivative of objective function}, we have 
    \begin{align*}
        &f(a) - f(b) - \langle \nabla f(b), a-b \rangle \\
        =\: & \!-\!\sum_{\mathclap{i\in N, j\in M_i}} \: a_{ij} (\log v_{ij} - \log \tfrac{a_{ij}}{x_j(a)} ) 
        + \sum_{\mathclap{i\in N, j\in M_i}} \: b_{ij} (\log v_{ij} - \log \tfrac{b_{ij}}{x_j(b)} ) 
        - \! \sum_{\mathclap{i\in N, j\in M_i}} \: (a_{ij} - b_{ij})(\log\tfrac{b_{ij}}{x_j(b)} - \log v_{ij}) \\
        =\: & \sum_{\mathclap{i\in N, j\in M_i}} \: a_{ij} \log(a_{ij}/x_j(a)) - \sum_{\mathclap{i\in N, j\in M_i}} \: b_{ij} \log(b_{ij}/x_j(b))
            - \sum_{\mathclap{i\in N, j\in M_i}} \: (a_{ij} - b_{ij}) \log(b_{ij}/x_j(b)) \\
        =\: & D_g(a\| b),
    \end{align*}
    where $D_g$ is the Bregman divergence for the entropy-like function $g(b)=\sum_{i\in N, j\in M_i} b_{ij}\log (b_{ij} / x_j(b))$. %
    It remains to show that $D_g(a\|b) \leq D_h(a\|b)$:
    \begin{align*}
        &D_h(a\| b) - D_g(a\|b) \\
        =\: & \sum_{\mathclap{i\in N, j\in M_i}} \: a_{ij}\log a_{ij} - \sum_{\mathclap{i\in N, j\in M_i}} \: b_{ij}\log b_{ij} - \sum_{\mathclap{i\in N, j\in M_i}} \: (a_{ij} - b_{ij}) \log b_{ij}\\
        & - \sum_{\mathclap{i\in N, j\in M_i}} \: a_{ij}\log (a_{ij} / x_j(a)) + \sum_{\mathclap{i\in N, j\in M_i}} \: b_{ij}\log( b_{ij} / x_j(b)) + \sum_{\mathclap{i\in N, j\in M_i}} \: (a_{ij} - b_{ij}) \log(b_{ij} / x_j(b)) \\
        =\: & \sum_{\mathclap{i\in N, j\in M_i}} \: a_{ij}\log x_j(a) - \sum_{\mathclap{i\in N, j\in M_i}} \: b_{ij}\log x_j(b) - \sum_{\mathclap{i\in N, j\in M_i}} \: (a_{ij} - b_{ij}) \log x_j(b) \\
        =\: & \sum_{\mathclap{i\in N, j\in M_i}} \: a_{ij}\log(x_j(a) / x_j(b)) = \sum_{j \in M} x_j(a)\log(x_j(a) / x_j(b))
        \geq 0.
    \end{align*}
    The last step follows by noting that the second-to-last expression is the KL divergence between $x(a)$ and $x(b)$, which is always nonnegative.
\end{proof}

Now we can combine \cref{lem:relative lipschitz} with \cref{thm:md rel lip} to get a $D_h(b^*\| b^0)/t$ rate of convergence for the proportional response dynamics.
If we start the dynamics at the uniform allocation $b_{ij}^0 = B_i / |M_i|$, we can upper bound the Bregman divergence $D_h(b^*\| b^0)$ as follows:
\begin{align*}
    D_h(b^*\| b^0) 
    &= h(b^*) - h(b^0) \\
    &= \sum_{i\in N} \sum_{j\in M_i} b^*_{ij}\log\frac{b^*_{ij}}{B_i/|M_i|} \tag{definitions of $h$ and $b^0$} \\
    &= \sum_{i\in N} B_i \sum_{j\in M_i} q_{ij}\log\frac{q_{ij}}{1/|M_i|} \tag{writing $q_{ij} := b^*_{ij}/B_i$} \\
    &= \sum_{i\in N} B_i \log |M_i| + \sum_{i\in N} B_i \sum_{j\in M_i} q_{ij}\log q_{ij} \tag{using $\sum_{j \in M_i} q_{ij} = 1$} \\
    &\le B \log m. \tag{using $|M_i| \le m$ and $\sum_{j\in M_i} q_{ij}\log q_{ij} \le 0$}
\end{align*}
Combining this with \cref{thm:md rel lip} and \cref{lem:relative lipschitz}, we obtain the desired rate of convergence for the proportional response dynamics. For all rates we state, we assume that the dynamics is initialized at the uniform allocation.
\begin{corollary}
	[Convergence rate of the PR dynamics]
	\label{cor:convergence-rate-shmyrev}
	Let $b^*$ be an optimal solution to \Cref{eq:shmyrev cp min}. Then the sequence of iterates of the proportional response dynamics satisfies
	\[
	f(b^t) - f(b^*)
	\le
	\frac{B \log(m)}{t}.
	\]
\end{corollary}
For the particular case of $B = 1$, we therefore see that proportional response dynamics converges at a rate of $\log(m)/t$.

\subsubsection{Convergence of Nash Social Welfare}
The convergence rate of \Cref{cor:convergence-rate-shmyrev} is stated in terms of the objective function $f$ of our new convex program (\Cref{eq:shmyrev cp min}). We can deduce that the same convergence rate applies with respect to the Nash welfare objective. For an allocation $x$, let us write $\varphi(x) = -\sum_{i \in N} B_i \log \langle v_i, x\rangle$ for its objective value in \Cref{eq:EG} (interpreted as a minimization problem). Then we can establish the following connection between the objective value of a solution in the two programs, which we can think of as a consequence of the double duality between the two programs, though we give a direct proof in \Cref{app:approximations-uncapped}.
\begin{restatable}
	{lemma}
	{nashboundedbyshmyrev}
	\label{lem:nash-bounded-by-shmyrev}
	Let $b$ be a feasible solution to \Cref{eq:shmyrev cp min}. Then
	\[
	f(b) \ge \varphi(x(b)) + \sum_{i \in N} B_i \log B_i,
	\]
	which holds with equality if $b$ is an optimal solution to \Cref{eq:shmyrev cp min}.
\end{restatable}

\Cref{lem:nash-bounded-by-shmyrev} implies that $\varphi(x(b^t)) - \varphi(x(b^*)) \le f(b^t) - f(b^*)$.
Plugging this into \Cref{cor:convergence-rate-shmyrev}, we obtain the following convergence bound for the (log) Nash product objective.

\begin{corollary}
	[Convergence of Nash product objective]
	\label{cor:convergence-rate-nash}
	Let $b^*$ be an optimal solution to \Cref{eq:shmyrev cp min}. Then the sequence of iterates of the proportional response dynamics satisfies
	\[
	\varphi(x(b^t)) - \varphi(x(b^*))
	\le
	\frac{B \log(m)}{t}.
	\]
\end{corollary}

The convergence rate of \Cref{cor:convergence-rate-nash} was recently independently obtained by \citet{zhao2023convergence}, after it had been an open question for almost forty years after the discovery of the dynamics by \citet{cover1984algorithm}.  \citet{zhao2023convergence} derived this rate directly. Our result gives a deeper explanation of the performance of the proportional response dynamics: it is equivalent to mirror descent with the entropy reference function applied to \cref{eq:shmyrev cp min}.

\subsubsection*{Convergence of Proportional Fairness and Core}
\label{sec:convergence-pf-core}

Finally, we will give a convergence bound in terms of fairness properties of the allocation $x^t$. We will do this by showing that an allocation that has approximately optimal Nash social welfare must also approximate these fairness properties.

The first such property is known as \emph{proportional fairness} (PF). 
This property was first introduced in the context of fairly allocating bandwidth in networks \citep{kelly1998rate}, and has also been studied in the allocation of indivisible private goods \citep[App.~D]{caragiannis2019unreasonable} as well as for public goods \citep{banerjee2023proportionally} and in the distortion of voting rules \citep{ebadian2024optimized}.

An allocation $x$ is proportionally fair if
\[
	\sum_{i \in N} B_i \frac{u_i(y)}{u_i(x)} \le B \text{ for all allocations $y$}
	\iff
	\sum_{i \in N} B_i \frac{v_{ij}}{u_i(x)} \le 1 \text{ for all $j \in M$.}
\]
That these two conditions are equivalent follows from linearity of $u_i(y)$ and the definition of allocations requiring that $\sum_{j \in M} y_j \le B$.
Intuitively, the PF condition says that the multiplicative utility gain of moving from $x$ to another allocation cannot be too high. For example, note that if some agent has utility 0 in allocation $x$, then the left-hand side is infinite (by considering some $j$ with $v_{ij} > 0$), and thus $x$ is not proportionally fair. 

It is known that Lindahl equilibrium allocations are proportionally fair, and indeed our proof of \Cref{thm:nash-is-lindahl}, which shows that allocations that maximize the Nash social welfare are Lindahl equilibrium allocations, establishes proportional fairness as an intermediate step. 

Let us define the \emph{PF value} of an allocation $x$ as
\[
\textup{PF}(x) = \max_{j \in M} \sum_{i \in N} B_i \frac{v_{ij}}{u_i(x)},
\]
so that an allocation is proportionally fair if and only if $\textup{PF}(x) \le 1$.

We will show that allocations that approximately maximize the Nash social welfare satisfy an approximate version of the PF condition, in the sense of having a PF value not much higher than 1.
Recall that we write $\varphi(x) = -\sum_{i \in N} B_i \log \langle v_i, x\rangle$ for the (negative) log Nash social welfare of an allocation $x$. Let $\varphi^*$ denote Nash social welfare of the optimum allocation (i.e., the minimum value of $\varphi$).
Write $B_{\min} = \min_{i \in N} B_i$.

\begin{restatable}
	[Approximate Nash gives approximate PF]
	{theorem}
	{approxnashpf}
	\label{thm:approx-nash-pf}
	Let $0 \le \epsilon < \frac12$ be such that $\epsilon < \frac{B_{\min}}{2}$, and consider an allocation $x$ with $\sum_{j \in M} x_j = B$ and such that $\varphi(x) \le \varphi^* + \epsilon$, i.e., that is $\epsilon$-optimal w.r.t. Nash social welfare. Then
	\[
		\textup{PF}(x) \le \frac{1 + 2n\sqrt{\epsilon / B} + n\epsilon/B}{1-\sqrt{2\epsilon / B_{\min}}}.
	\]
\end{restatable}

The proof uses $\epsilon$-KKT conditions and appears in \Cref{app:approximations-uncapped}. The assumption that $\sum_{j \in M} x_j = B$ is innocuous, since an allocation that does not use the entire budget $B$ can be filled up in polynomial time without decreasing its Nash social welfare.

For the special case where $B =1$ and $B_i = 1/n$ for all $i \in N$, the bound of \Cref{thm:approx-nash-pf} reads as
\begin{align*}
	\textup{PF}(x) = \max_{j \in M} \sum_{i \in N} \frac1n \frac{v_{ij}}{u_i(x)}
	\le 
	\frac{1 + 2n \sqrt{\epsilon} + n\epsilon}{1-\sqrt{2n\epsilon}} 
	&= 1 + \frac{1 + 2n \sqrt{\epsilon} + n\epsilon - (1-\sqrt{2n\epsilon})}{1-\sqrt{2n\epsilon}}.
\end{align*}
For $\epsilon < \frac{1}{8n}$, one can show that the right-hand side is upper bounded by $1 + 8n\sqrt{\epsilon}$ and thus
\[
\textup{PF}(x)
\le 
1 + 8n\sqrt{\epsilon}.
\]

It is known that an approximately proportionally fair allocation lies in a (multiplicatively) approximate core \citep[see, e.g.,][Prop.~2.6]{ebadian2024optimized}. We include a proof for completeness.

\begin{lemma}[Approximate PF implies approximate core]
	\label{lem:pf-implies-core}
	Let $x$ be an allocation with $\textup{PF}(x) \le \alpha$. Then $x$ lies in the $\alpha$-approximate core, in the sense that there is no non-empty blocking coalition $S$ with objection $z$ such that $\sum_{j \in M} z_j \le \sum_{i \in S} B_i$ and such that $u_i(z) > \alpha \cdot u_i(x)$ for all $i \in S$.
\end{lemma}
\begin{proof}
	We show the contrapositive.
	Suppose that $x$ admits a blocking coalition $S$ with objection $z$ such that $\sum_{j \in M} z_j \le \sum_{i \in S} B_i =: B_S$ and $u_i(z) > \alpha u_i(x)$ for all $i \in S$. Consider the allocation $y = \frac{B}{B_S} z$. Then $u_i(y) = \frac{B}{B_S} u_i(z)$ by linearity of $u_i$. Thus $u_i(y) > \alpha \frac{B}{B_S} u_i(x)$ for all $i \in S$ and hence
	\[
	\sum_{i \in N} B_i \frac{u_i(y)}{u_i(x)} \ge \sum_{i \in S} B_i \frac{u_i(y)}{u_i(x)} > \sum_{i \in S} B_i\, \alpha \frac{B}{B_S} = \alpha B.
	\]
	Dividing by $B$, this implies that $\textup{PF}(x) > \alpha$.
\end{proof}

We can combine all these results into convergence results with respect to PF value and core approximation.

\begin{corollary}
	[Convergence of PF value and Core]
	For each $t \ge 1$, let 
	\[
	\epsilon_t = \frac{B \log(m)}{t}
	\quad \text{and} \quad
	\alpha_t = \frac{1 + 2n\sqrt{\epsilon_t / B} + n\epsilon_t / B}{1-\sqrt{2\epsilon_t / B_{\min}}}.
	\]
	Note that $\epsilon_t \to 0$ and $\alpha_t \to 1$ as $t \to \infty$.
	The sequence $(b^t)_{t \ge 1}$ of iterates of the proportional response dynamics induces a sequence of allocations $x^t = x(b^t)$ such that for each $t \ge 1$ that is large enough such that $\epsilon_t < \min\{\frac{1}{2}, \frac{B_{\min}}{2}\}$, we have that $\textup{PF}(x^t) \le \alpha_t$, and $x^t$ lies in the $\alpha_t$-approximate core.
\end{corollary}

This follows by chaining together \Cref{cor:convergence-rate-nash} (convergence rate in terms of Nash welfare), \Cref{thm:approx-nash-pf} (PF value guaranteed by Nash approximation), and \Cref{lem:pf-implies-core} (PF value implies core approximation).

\begin{remark}[Computation of approximate core solutions]
	The results in this section also imply that in the uncapped public goods setting, an approximate core allocation can be computed in polynomial time by finding an approximately Nash welfare optimal allocation, for example using the ellipsoid algorithm \citep[see, e.g.,][Theorem 13.1]{vishnoi2021algorithms}. Another way to compute an approximate core solution is to find an allocation with low PF value \citep[Appendix B]{fain2018indivisiblepublic}. An entirely different method of computing core allocations is proposed by \citet{martinez2025cooperation}, who give a cutting plane algorithm for computing core outcomes in a class of NTU games that includes the uncapped public goods model.
\end{remark}

It is also possible to show that an allocation $x$ with a low PF value (such as $\text{PF}(x) \le 1 + \epsilon$) forms an approximate Lindahl equilibrium using the prices $p_{ij} = B_i v_{ij} / u_i(x)$ which is well-defined since $u_i(x) > 0$ when the PF value is finite. Note that an approximate version of profit maximization is then satisfied by definition of PF value since $\sum_{i \in N} p_{ij} \le \text{PF}(x) \le 1 + \epsilon$ for each $j \in M$, and affordability and utility maximization hold exactly (observing that for each agent, the bang-per-buck ratio is constant across projects). Note that this is a different notion of approximate Lindahl equilibrium than the one we will consider later (\Cref{def:epsilon-lindahl-equilibrium}).

\section{Capped Public Goods}
Next, we study the capped public goods setting, where we have a constraint $x_j \leq \text{cap}_j$ for each good $j \in M$. 
One may na\"ively attempt to add this constraint to \cref{eq:EG} maximizing Nash welfare, but this will not lead to a Lindahl equilibrium and not even to a weak core solution.

\begin{example}[Max Nash welfare is not Lindahl]
	\label{ex:capped-nash-fails-core}
	Consider the following instance:%
	\footnote{This example is similar to well-known instances in (indivisible) approval-based multi-winner voting where the PAV rule fails the core \citep{aziz2017ejr,peters2020welfarism,peters2025core}.}
	\begin{center}
		{\upshape
			\begin{tabular}{lccccc}
				\toprule
				& $B_i$ & Project 1 & Project 2 & Project 3 & Project 4 \\
				\midrule
				Agent 1 & $2$ & $1$ & $1$ & $0$ & $0$ \\
				Agent 2 & $2$ & $1$ & $0$ & $1$ & $0$ \\
				Agent 3 & $2$ & $0$ & $0$ & $0$ & $1$ \\
				\midrule
				$\text{cap}_j$ & & $3$ & $\infty$ & $\infty$& $\infty$ \\
				\bottomrule
		\end{tabular}}
	\end{center}
	The allocation that maximizes Nash welfare subject to the cap constraints is $x = (3, 0, 0, 3)$. This allocation violates the weak core: consider the blocking coalition $S = \{1, 2\}$ and the objection $z = (3, 0.5, 0.5, 0)$ which gives each $i \in S$ a utility of $u_i(z) = 3.5$ which is strictly higher than $u_i(x) = 3$. Thus, by \Cref{prop:lindahl is core}, $x$ is not a Lindahl equilibrium. This is not an artifact of having zero-valuations; replacing 1s by 10 and 0s by 1 leads to the same situation.
\end{example}

This failure of the Nash rule to extend to capped settings has been noted several times. \citet[Proposition 4.1]{suzuki2024maxflow} provide an example similar to the one above. \citet[Comment A.1]{garg2021markets} write that Lindahl equilibrium ``does not transform into a Fisher market''.
While the Nash optimum fails the core, it can be shown that it satisfies a 2-approximation to it \citep[Corollary 3.5]{munagala2022coremultilinear}.

\subsection{Adapting the Convex Program}
We will show in this section that \cref{eq:shmyrev cp} can be used to compute a Lindahl equilibrium in the capped public goods setting through a simple modification: we simply add a constraint $x_j\leq \text{cap}_j$ for all $j\in M$. Surprisingly, we will show that this works, even though the exact same constraint does not work for the original EG program (\cref{eq:EG}) for maximizing Nash welfare. Thus, we obtain the first efficient algorithm for capped public goods, thereby resolving an open problem first posed by \citet{fain2016core}.

Our modified program for capped public goods is as follows, where as before we write $x_j(b) = \sum_{i \in N_j} b_{ij}$ as a shorthand:
\begin{program}
\begin{aligned}
	\max_{b \ge 0} \quad & \sum_{\mathclap{i \in N, j \in M_i}} \: b_{ij} \, \left(\log v_{ij} - \log (b_{ij} / x_j(b)) \right) \\
	\text{s.t.} \quad &\sum_{\mathclap{j \in M_i}} \: b_{ij} \le B_i \text{ for all $i \in N$ } \\
	&x_{j}(b) \le \text{cap}_j \text{ for all $j \in M$}
\end{aligned}
\label{eq:shmyrev cp capped}
\end{program}
Note that the agents' budget constraints are written as inequalities in this version of the problem, since the caps may preclude feasible solutions from exhausting all agents' budgets (except when the instance is cap-sufficient, see \Cref{def:cap-sufficient} and \Cref{prop:cap-sufficient implications}).

We will require that all valuations have been rescaled such that $v_{ij} > 1$ for all $v_{ij} \ne 0$,%
\footnote{A similar normalization is used by \citet{brandl2022fundingpublicprojects}.}
to ensure that the terms ``$b_{ij} \log v_{ij}$'' in the objective function have a positive coefficient. Without the rescaling, optimal solutions may not spend the entire budget of an agent even when it is possible to do so. Rescaling is harmless since the Lindahl equilibrium is invariant to scaling valuations by a positive constant (even different constants for different agents). On instances that are cap-sufficient, rescaling also does not change the optimum solutions because multiplying an agent's valuation by a factor $\alpha > 0$ changes the objective function by adding an additive term $\log \alpha \sum_{j \in M_i} b_{ij}$, which acts as a constant since cap-sufficiency implies that agents spend their full budget at optimal solutions, so $\sum_{j \in M_i} b_{ij} = B_i$. (Without cap sufficiency, rescaling can change the optimum solutions of the program.) We did not have to do a rescaling in the uncapped version of the program, because we used an equality constraint for agent budgets in that version.

\subsection{The Convex Program Computes a Lindahl Equilibrium}

In this section, we will prove that \cref{eq:shmyrev cp capped} computes a zero-respecting Lindahl equilibrium. This provides the first existence proof of Lindahl equilibrium in the capped public goods setting that does not use a fixed-point theorem. It also proves the existence of an equilibrium that is zero-respecting, which does not quite follow from the existence result of \citet{foley1970lindahl}, since his model does not allow for caps and does not allow for valuations equal to 0 (since it assumes strictly monotonic valuations). In addition, as explained in \Cref{fn:discontinuous}, zero-respecting equilibria cannot be obtained by a limit argument that replaces zero valuations by $\epsilon$ and letting $\epsilon \to 0$. The zero-respecting property is intuitively attractive, since it prevents agents from having to contribute their budget to goods that they do not have a positive valuation for.

Our proof proceeds by analyzing the KKT conditions (\Cref{thm:subdiff-kkt}) applied to \cref{eq:shmyrev cp capped}. In the notation of \Cref{thm:subdiff-kkt}, the objective function is obtained by multiplying by $-1$, giving $f(b) = -\sum_{i\in N, j\in M_i} b_{ij} \log v_{ij} + \sum_{j \in M} h(b_j)$ where $b_j$ is the vector $(b_{ij})_{i \in N_j}$ and $h$ is the function defined as $h(x) = \sum_i x_i \log\frac{x_i}{\sum_k x_k}$ when all $x_i$ are non-negative (and $h(0) = 0$), and $+\infty$ otherwise.

To write down the KKT conditions, we need to compute the subdifferential of $f$. 
\begin{restatable}
	[Subdifferential of objective function]
	{lemma}{shmyrevsubdiff}
 	\label{lem:shmyrev subdiff}
 	Let $f$ be the negative of the objective function of \cref{eq:shmyrev cp capped}, and let $b = (b_{ij})_{i\in N, j \in M_i}$ be a feasible point. If there is some $i \in N$ and $j \in M_i$ such that $b_{ij} = 0$ but $x_j(b) > 0$, then $\partial f(b) = \emptyset$. Otherwise, for all vectors $g = (g_{ij})_{i \in N, j \in M_i}$, we have
	\[
		g \in \partial f(b) 
		\iff
		\left\{
			\begin{array}{ll}
				\text{for all } j\in M \text{ s.t. } x_j(b) > 0 \text{, and all $i \in N_j$}, & g_{ij} = - \log v_{ij} + \log \frac{b_{ij}}{x_j(b)}, \\
				\text{for all } j\in M \text{ s.t. } x_j(b) = 0, & \sum_{i\in N_j} e^{g_{ij} + \log v_{ij}} \leq 1.
			\end{array}
		\right.
	\]
\end{restatable}
The detailed proof is in \Cref{app:proofs}.

Based on this computation of the subdifferential of the objective function of \cref{eq:shmyrev cp capped}, we can now prove that an optimal solution to the program will form a Lindahl equilibrium.

\begin{theorem}
	[Optimal Solutions are Lindahl]
    \label{thm:lindahl shmyrev capped}
	Assume that valuations are rescaled such that $v_{ij} > 1$ for all $v_{ij} \ne 0$.
	Let $b^*$ be an optimal solution to \cref{eq:shmyrev cp capped} with induced allocation $x^* = (x_j(b^*))_{j \in M}$. Then there exist zero-respecting prices $p$ such that $(x^*,p)$ forms a Lindahl equilibrium for the capped public goods setting.
\end{theorem}
\begin{proof}
We apply the KKT conditions of \Cref{thm:subdiff-kkt} to \cref{eq:shmyrev cp capped}, noting that Slater's constraint qualification is satisfied by a solution setting all $b_{ij}$ to a sufficiently small but positive value.
For convenience, let us label the constraints of the program using the notation of the KKT conditions:
\begin{alignat*}{3}
	h_i^{(1)}(b) &= \textstyle\sum_{j \in M_i} b_{ij} - B_i &\qquad& \text{for all $i \in N$,} \\
	h_j^{(2)}(b) &= \textstyle\sum_{i \in N_j} b_{ij} - \text{cap}_j && \text{for all $j \in M$,} \\
	h_{ij}^{(3)}(b) &= -b_{ij} && \text{for all $i \in N$ and $j \in M_i$.}
\end{alignat*}
All these functions are affine and thus differentiable, with singleton subdifferentials.

Let $b^*$ be an optimal solution to \cref{eq:shmyrev cp capped}. By \Cref{thm:subdiff-kkt}, we know that there exists a subgradient $g^* \in \partial f(b^*)$ such that
\begin{equation}
g^* + \sum_{i\in N} \lambda_i \nabla h_i^{(1)}(b^*) + \sum_{j\in M} \mu_j \nabla h_j^{(2)}(b^*) \: + \: \sum_{\mathclap{i\in N, j\in M_i}} \: \eta_{ij} \nabla h_{ij}^{(3)}(b^*) = 0,
\label{eq:zero-in-subdiff}
\end{equation}
with $\lambda_i, \mu_j, \eta_{ij} \ge 0$ and such that complementary slackness holds. 
\Cref{eq:zero-in-subdiff} can be equivalently stated as
\begin{equation}
	g^*_{ij} + \lambda_i  + \mu_j -  \eta_{ij}  = 0 \quad \text{for all $i \in N$ and $j \in M_i$.}
\end{equation}

Let us now understand the implications of \eqref{eq:zero-in-subdiff}. We go through each project $j \in M$, and distinguish the cases where $x_j^* = \text{cap}_j$, where $0 < x_j^* < \text{cap}_j$, and where $x_j^* = 0$.

\begin{itemize}
	\item Consider a project $j \in M$ with $x_j^* = \text{cap}_j$. Let $i \in N_j$. 
	Note that $b_{ij}^* > 0$, since otherwise the subdifferential is empty, contradicting $g^* \in \partial f(b^*)$. Thus, by complementary slackness, $\eta_{ij} = 0$. 
	Then the $ij$'th component of \eqref{eq:zero-in-subdiff} implies
	\[
	0 = -\log v_{ij} + \log \tfrac{b_{ij}^*}{x_j^*} + \lambda_i + \mu_j
	\]
	and thus, because $\mu_j \ge 0$, that
	\[
	0 \ge -\log v_{ij} + \log \tfrac{b_{ij}^*}{x_j^*} + \lambda_i.
	\]
	Rearranging and exponentiating both sides, we conclude that
	\[
		v_{ij} \ge e^{\lambda_i} \frac{b_{ij}^*}{x_j^*}.
	\]
	
	\item Consider a project $j \in M$ with $0 < x_j^* < \text{cap}_j$. Let $i \in N_j$. 
	Again note that $b_{ij}^* > 0$ since otherwise the subdifferential $\partial f(b^*)$ is empty.
	Then, by complementary slackness, $\mu_j = \eta_{ij} = 0$. Thus, the $ij$'th component of \eqref{eq:zero-in-subdiff} implies
	\[
	0 = g^*_{ij} + \lambda_i = - \log v_{ij} + \log \tfrac{b^*_{ij}}{x_j^*} + \lambda_i.
	\]
	Rearranging and	exponentiating both sides, we conclude that
	\[
		v_{ij} = e^{\lambda_i} \frac{b_{ij}^*}{x_j^*}.
	\]
	
	\item Finally, consider a project $j \in M$ with $x_j^* = 0$. Thus, by complementary slackness, $\mu_j = 0$. Let $i \in N_j$. 
	By definition of $x_j^*$, we have $b^*_{ij} = 0$. 
	The $ij$'th component of \eqref{eq:zero-in-subdiff} implies
	\[
	0 = g^*_{ij} + \lambda_i - \eta_{ij}.
	\]
	Writing $w_{ij} = g^*_{ij} + \log v_{ij}$ and noting that $\eta_{ij} \ge 0$, we get
	\[
	0 \le w_{ij} - \log v_{ij} + \lambda_i.
	\]
	Rearranging and exponentiating both sides, we conclude that
	\[
		v_{ij} \le e^{\lambda_i} e^{w_{ij}}.
	\]
\end{itemize}
Collecting all our conclusions, we have found that for $i \in N$ and $j \in M_i$,
\begin{equation}
	v_{ij} \begin{cases}
		\ge e^{\lambda_i} \frac{b_{ij}^*}{x_j^*} & \text{if $x_j^* = \text{cap}_j$,} \\
		=   e^{\lambda_i} \frac{b_{ij}^*}{x_j^*} & \text{if $0 < x_j^* < \text{cap}_j$, } \\
		\le e^{\lambda_i} e^{w_{ij}}         & \text{if $x_j^* = 0$.} \\
	\end{cases}
	\label{eq:kkt-conclusions}
\end{equation}

We now form a Lindahl equilibrium using the following prices for $i \in N$ and $j \in M_i$:
\begin{equation}
p_{ij} = 
\begin{cases}
	\frac{b_{ij}^*}{x_j^*} & \text{if $x_j^* > 0$,} \\
	e^{w_{ij}} & \text{if $x_j^* = 0$.} \\
\end{cases}
\label{eq:price-definition}
\end{equation}
For $i,j$ such that $v_{ij} = 0$, we set $p_{ij} = 0$ and $b_{ij}^* = 0$.
It follows from these definitions that the prices are zero-respecting.
Note that with these prices, the identity $p_{ij}x_j^* = b_{ij}^*$ holds for all $i \in N$ and $j \in M_i$ (by case analysis on whether $x_j^* = 0$).

We claim that $(x^*, p)$ forms a Lindahl equilibrium.

For profit maximization, note that if $x_j^* > 0$, then $\sum_{i \in N} p_{ij} = \sum_{i \in N_j} b_{ij}^*/x_j^* = 1$ by definition of $x_j^*$ and using $p_{ij} = 0$ when $v_{ij} = 0$, and if $x_j^* = 0$, then $\sum_{i \in N} p_{ij} = \sum_{i \in N_j} e^{w_{ij}} \le 1$ from the subdifferential characterization in \Cref{lem:shmyrev subdiff}.

For the affordability condition, for each $i \in N$ we have
\[
\langle p_i, x^* \rangle = \sum_{j \in M} p_{ij}x_j^* = \sum_{j \in M_i} b_{ij}^* \le B_i,
\]
using the identity $p_{ij}x_j^* = b_{ij}^*$ and the feasibility of $b^*$ in \cref{eq:shmyrev cp capped}.

It remains to prove utility maximization. Fix an agent $i \in N$. We will show that $x^*$ is utility maximizing subject to the budget constraint $\langle p_i, x \rangle \le B_i$.
We divide the proof of this into two parts, based on whether agent $i$ spends their entire budget under $p_i$ or not.

First suppose that $\sum_{j \in M_i} b_{ij}^* < B_i$. 
We want to show that this only occurs when all projects $j\in M$ with $v_{ij} > 0$ have $x_j^* = \text{cap}_j$.
Suppose for a contradiction that $x_j^* < \text{cap}_j$ for some $j \in M_i$.
By complementary slackness, we have $\lambda_i = \mu_j = 0$. Thus, \eqref{eq:zero-in-subdiff} implies $g^*_{ij} \geq 0$. 
If $x_{j}^* > 0$, then we have $g^*_{ij} = -\log v_{ij} + \log(b^*_{ij}/x^*_j)$ and thus $g^*_{ij} < 0$ because $v_{ij} > 1$ by assumption and $b^*_{ij}/x^*_j \le 1$ by definition of $x_j^*$, a contradiction.
Otherwise, if $x_{j}^* = 0$, then \smash{\raisebox{-0.25ex}{$\raisebox{0.25ex}{$e$}\,^{g^*_{ij} + \log v_{ij}}$}} $\leq 1$ by the subdifferential characterization in \Cref{lem:shmyrev subdiff}, which implies $g^*_{ij} + \log v_{ij} \leq 0$. Hence $g^*_{ij} \leq -\log v_{ij} < 0$ since $v_{ij} > 1$, again a contradiction.

Thus, we have shown that if agent $i$ does not spend their whole budget, then they are already achieving the maximal possible utility under \emph{any} feasible allocation (because $x^*_j = \text{cap}_j$ for all $j\in M$ such that $v_{ij} > 0$), and thus $x^*$ is utility maximizing for $i$.

Next consider the case where $\sum_{j \in M_i} b_{ij}^* = B_i$.
Combining \eqref{eq:kkt-conclusions} and \eqref{eq:price-definition}, we have that the ``bang per buck'' of project $j \in M_i$ satisfies
\begin{equation}
	\frac{v_{ij}}{p_{ij}} \begin{cases}
		\ge e^{\lambda_i} & \text{if $x_j^* = \text{cap}_j$,} \\
		=   e^{\lambda_i} & \text{if $0 < x_j^* < \text{cap}_j$, } \\
		\le e^{\lambda_i} & \text{if $x_j^* = 0$.} \\
	\end{cases}
	\label{eq:bang-per-buck}
\end{equation}

Now, an affordable bundle $y$ is utility maximizing for $i$ (among bundles satisfying the cap constraints) if and only if (i) for every project with $\smash{\frac{v_{ij}}{p_{ij}}} > e^{\lambda_i}$, we have $y_j = \text{cap}_j$, and (ii) for every project with $\smash{\frac{v_{ij}}{p_{ij}}} < e^{\lambda_i}$, we have $y_j = 0$, and (iii) the whole budget is spent. Because $x^*$ is such a bundle, it is utility maximizing for $i$.

More formally, let $y = (y_j)_{j\in M}$ be an allocation such that $0 \le y_j \le \text{cap}_j$ for all $j \in M$ and $\langle p_i, y \rangle \le B_i$. Then for every $j \in M$ with $x_j^* = 0$ we have $y_j - x_j^* \ge 0$, and for every $j \in M$ with $x_j^* = \text{cap}_j$ we have $y_j - x_j^* \le 0$. Thus
\begin{align*}
	u_i(y) - u_i(x^*) \:
	&\begin{alignedat}[t]{4}
		&= 
		&& \sum_{j \in M} v_{ij} (y_j - x_j^*)\\
		&= 
		&& \sum_{\substack{j \in M_i \\ x_j^* = 0}} v_{ij} (y_j - x_j^*) 
		&&+\: \sum_{\mathclap{\substack{j \in M_i \\ 0 < x_j^* < \text{cap}_j}}} v_{ij} (y_j - x_j^*) 
		&&+\: \sum_{\mathclap{\substack{j \in M_i \\ x_j^* = \text{cap}_j}}} v_{ij} (y_j - x_j^*) \\
		&\overset{\eqref{eq:bang-per-buck}}{\le}
		&&\sum_{\substack{j \in M_i \\ x_j^* = 0}} e^{\lambda_i} p_{ij}(y_j - x_j^*) 
		&&+\: \sum_{\mathclap{\substack{j \in M_i \\ 0 < x_j^* < \text{cap}_j}}} e^{\lambda_i} p_{ij} (y_j - x_j^*) 
		&&+\: \sum_{\mathclap{\substack{j \in M_i \\ x_j^* = \text{cap}_j}}} e^{\lambda_i} p_{ij} (y_j - x_j^*)
	\end{alignedat}
	 \\
	&= \: e^{\lambda_i} \sum_{j \in M_i} p_{ij} (y_j - x_j^*)
	= e^{\lambda_i} (\langle p_i, y \rangle - \langle p_i, x^* \rangle)
	\le 0.
\end{align*}
In the last line we used that $(x^*,p)$ is zero-respecting for the second equality, and we used that $\langle p_i, y \rangle \le B_i = \langle p_i, x^* \rangle$ for the last inequality. It follows that $u_i(y) \le u_i(x^*)$, establishing the utility maximization condition.
\end{proof}

\subsection{Approximately Optimal Solutions}
\label{sec:capped-approximate}

We have now seen in \Cref{thm:lindahl shmyrev capped} that an optimal solution to our convex program forms a Lindahl equilibrium in the capped public goods setting. This makes for a direct existence proof (since the feasible set of the program is compact and the objective function is continuous), but it does not immediately give us a way to compute a Lindahl equilibrium allocation. Indeed, even what it means to compute a Lindahl equilibrium is ill-defined since, as we saw in \Cref{ex:irrational}, there are simple rational instances where all Lindahl equilibria are irrational, so that no algorithm could write down the Lindahl equilibrium exactly. In addition, algorithms for convex programming do not produce exactly optimal but only $\epsilon$-optimal solutions.

To get around this issue, we will show that an approximate solution to our convex program produces an approximate Lindahl equilibrium. Our proof needs to relax the conditions of budget feasibility, utility maximization, and profit maximization in the definition of Lindahl equilibrium, leading to the following approximate equilibrium concept.

\begin{restatable}
	[$\epsilon$-Lindahl equilibrium]
	{definition}
	{defepsilonlindahl}
	\label{def:epsilon-lindahl-equilibrium}
	Let $\epsilon \ge 0$. A pair $(x,p)$ of a vector $x = (x_j)_{j \in M}$ such that $0 \le x_j \le \textup{cap}_j$ for all $j \in M$ and a set of personalized prices $p\ge 0$ forms an $\epsilon$-Lindahl equilibrium if
	\begin{itemize}
		\item $x$ is \emph{approximately budget-feasible}: $\sum_{j \in M} x_j \le B + \epsilon$ (i.e., $x$ may overshoot the total budget),
		\item $x$ is \emph{affordable}: we have $\langle p_i, x \rangle \le B_i$ for every $i \in N$,
		\item $x$ is \emph{approximately utility-maximizing}: for every $i \in N$ and every $y \in \mathbb{R}_{\ge 0}^m$ such that $0 \le y_j \le \textup{cap}_j$ for all $j \in M$ and such that $\langle p_i, y \rangle\le B_i$, we have $u_i(x) \ge u_i(y) - \epsilon$, 
		\item $x$ is \emph{approximately profit-maximizing}: for every $j \in M$ we have $\sum_{i \in N} p_{ij} \le 1$, and whenever $x_j > 0$ then $\sum_{i \in N} p_{ij} \ge 1 - \epsilon/x_j$.
	\end{itemize}
\end{restatable}

We are not claiming that this is the most natural definition of $\epsilon$-Lindahl equilibrium (especially because the way we relax profit maximization is not natural), but it does imply that the associated allocation approximately satisfies the core, where the core definition is relaxed in an additive sense by $\epsilon$ and the allocation itself overshoots the budget by $\epsilon$. 

\begin{lemma}
	[$\epsilon$-Lindahl equilibrium implies $\epsilon$-core]
	Suppose $(x,p)$ forms an $\epsilon$-Lindahl equilibrium in the sense of \Cref{def:epsilon-lindahl-equilibrium}, with $\sum_{j \in M} x_j \le B + \epsilon$. Then there does not exist a blocking coalition $S \subseteq N$ and objection $z = (z_j)_{j \in M} \in \mathbb{R}_{\ge 0}^m$ with $0 \le z_j \le \textup{cap}_j$ for all $j \in M$, such that $\sum_{j \in M} z_j \le \sum_{i \in S} B_i$ and $u_i(z) > u_i(x) + \epsilon$ for all $i \in S$.
\end{lemma}
\begin{proof}
	The proof is similar to the exact case (\Cref{prop:lindahl is core}). 
	
	Suppose $S \subseteq N$ is a blocking coalition with objection $z$ satisfying $\sum_{j \in M} z_j \le \sum_{i \in S} B_i$.
	Since for every $i \in S$ we have $u_i(x) < u_i(z) - \epsilon$, the utility maximization condition of $\epsilon$-Lindahl equilibrium implies that $z$ cannot be affordable for $i$, so that $\langle p_i, z \rangle > B_i$. Summing over $i \in S$, we get that $\sum_{i \in S} \langle p_i, z \rangle > \sum_{i \in S} B_i$.
	Thus, due to the non-negativity of prices and the first part of the approximate profit maximization condition of $\epsilon$-Lindahl equilibrium, we have
	\begin{align*}
		\sum_{i \in S} B_i 
		< \sum_{i \in S} \langle p_i, z \rangle 
		\le \sum_{i \in N} \langle p_i, z  \rangle 
		\le \langle 1, z \rangle
		\le \sum_{i \in S} B_i,
	\end{align*}
	a contradiction.
\end{proof}

We are now ready to state our approximation result.

\begin{restatable}
	[Computing an $\epsilon$-Lindahl equilibrium]
	{theorem}
	{approximatelindahl}
	\label{thm:approximate-lindahl}
	Let $0 < \epsilon \le \min\{1, \textup{cap}_{\min}\}$ where $\textup{cap}_{\min} = \min_{j \in M} \textup{cap}_j$. Assume that $v_{ij} > 1$ whenever $v_{ij} > 0$. There is an algorithm computing an $\epsilon$-Lindahl equilibrium for the capped public goods setting running in polynomial time (with a logarithmic dependence on $1/\epsilon$).
\end{restatable}

The algorithm of \Cref{thm:approximate-lindahl} works by slightly perturbing the input instance, and then solving \Cref{eq:shmyrev cp capped} to an appropriately chosen accuracy of $C \cdot \epsilon^3$, where $C$ depends single-exponentially on problem parameters (number of public goods, agent endowments, caps, largest $v_{ij}$).
Thus, the runtime claim of \Cref{thm:approximate-lindahl} is based on using an algorithm for solving the convex program that has a logarithmic or polylogarithmic dependence on the desired accuracy $\frac1\epsilon$ of the solution, such as the ellipsoid algorithm or interior-point methods. If using a solution algorithm for convex programming with a polynomial dependence on the accuracy (like most first-order methods), the runtime would be exponential in problem parameters.

The proof of \Cref{thm:approximate-lindahl} appears in \Cref{app:approximations-capped}. It makes use of $\epsilon$-KKT conditions. As a first step, we prove that approximately optimal solutions to \Cref{eq:shmyrev cp capped} approximately maximize agent utilities, under the additional assumption that every project receives a total funding of at least $t$, i.e., $x_j > t$ for all $j \in M$, for some given $t > 0$. We then remove the latter assumption by perturbing the input instance slightly so that every project is guaranteed at least $t$ funding; this perturbation causes the approximation errors in the budget feasibility and profit maximization conditions.

\subsection{Discussion of the Convex Program}

\paragraph{Comparison to Fisher markets.}
It is interesting to contrast our program with the Fisher market setting with private goods. There, the Eisenberg--Gale program also does not allow the introduction of saturating constraints on the primal variables (which correspond to a maximum amount of a good that an agent may receive). Yet it is not possible to add such constraints to the Shmyrev program for Fisher markets either, because that program does not contain the original primal variables encoding the allocation (in contrast to our public-goods program). Instead, the allocation is obtained through a nonlinear function of the optimization variables in the Shmyrev program.%
\footnote{It is known that \emph{spending constraints} on a per-buyer basis can be introduced to the Shmyrev program~\citep{birnbaum2011distributed}, but these are very different from saturating constraints on the primal variables.} 
Thus, \cref{eq:shmyrev cp} allows for a type of saturating consumption constraint that has previously never been possible for either private or public goods.

\paragraph{Scale-freeness.}
Subject to our requirement that $v_{ij} > 1$ whenever $v_{ij} \neq 0$ and assuming cap-sufficiency, the optimal solution to our program does not change if the valuations of some agent $i$ are multiplied by a positive constant $\alpha > 0$ that preserves the requirement that $v_{ij} > 1$ whenever $v_{ij} \neq 0$. This follows because this operation adds a term $\log \alpha \sum_{j \in M_i} b_{ij}$ to the objective, which by cap sufficiency equals the constant $(\log \alpha) B_i$ at optimum solutions. The program is also invariant to changing the ``currency'' by multiplying all the budgets $B_i$ and all the caps by a common positive constant $\beta > 0$. This is because the feasible solutions of the program before the change are in bijection with the feasible solutions of the program after the change (mapping $b_{ij} \mapsto \beta\cdot b_{ij}$), and this bijection preserves the objective value of each solution up to multiplying it by $\beta$.

\paragraph{Discontinuity as $v_{ij} \to 0$.}
Given that our program computes a Lindahl equilibrium that is zero-respecting, its output is not continuous as $v_{ij} \to 0$. This is unavoidable due to \Cref{ex:zero-respecting Lindahl equilibrium underspends} (see \Cref{fn:discontinuous}), and unsurprising in light of our normalization of valuations.

\paragraph{Not all Lindahl equilibria are optimal solutions.}
In the uncapped setting, \emph{every} Lindahl equilibrium forms an optimum of both \cref{eq:shmyrev cp} and the Eisenberg--Gale program. As the following example shows, this is not the case for the capped setting, where \cref{eq:shmyrev cp capped} captures only a strict subset of Lindahl equilibria. The example also shows that Lindahl equilibria are not unique in utilities.

\begin{example}[Lindahl equilibrium is not unique in utilities]
	\label{ex:non-unique}
	Consider the following instance:
	\begin{center}
		{\upshape
			\begin{tabular}{lcccc}
				\toprule
				& $B_i$ & Project 1 & Project 2 & Project 3 \\
				\midrule
				Agent 1 & $1$ & $1$ & $1$ & $0$ \\
				Agent 2 & $1$ & $1$ & $0$ & $1$ \\
				\midrule
				$\text{cap}_j$ & & $1$ & $\infty$ & $\infty$ \\
				\bottomrule
		\end{tabular}}
	\end{center}
	This instance is cap-sufficient, since each agent has a positive valuation for an uncapped project. 
	Let us determine the set of zero-respecting Lindahl equilibria $(x,p)$.
	By \Cref{cor:pareto}, $x$ is Pareto-optimal, and therefore $x_1 = 1$ and $x_2 + x_3 = 1$.
	For each $\gamma \in [0,1]$, one can check that $x = (1, 1-\gamma, \gamma)$ forms a Lindahl equilibrium together with the prices $p_1 = (\gamma, 1, 0)$ and $p_2 = (1- \gamma, 0, 1)$.
	It follows that, in the capped setting, Lindahl equilibria are not unique in utilities: in the equilibrium allocation $(1,1,0)$, agent 1 obtains utility $2$, but in the equilibrium allocation $(1,0,1)$, agent 1 obtains utility $1$.
	
		Note that the allocation $x^* = (1,\frac12,\frac12)$ is the unique allocation that is intuitively fair and respects the symmetry of the instance, but this allocation is not the only Lindahl equilibrium.
		To apply \cref{eq:shmyrev cp capped}, we use our normalization convention from above and replace each positive valuation by (say)~$2$.
		On this normalized instance, \cref{eq:shmyrev cp capped} uniquely selects $x^*$: because all coefficients $b_{ij}\log v_{ij}$ are positive, each agent spends their full budget; by Pareto optimality $x_1=1$; and given this, after eliminating constant terms, maximizing the objective reduces to maximizing $-b_{11}\log b_{11}-b_{21}\log b_{21}$ under $b_{11}+b_{21}=1$, whose unique optimum is $b_{11}=b_{21}=0.5$.
	\end{example}

An intuitive reason for why our program does not capture all Lindahl equilibria is that the KKT conditions that we analyzed in the proof of \Cref{thm:lindahl shmyrev capped} impose an additional constraint on the contributions of agents to projects that are fully funded ($x_j = \text{cap}_j$), saying that every agent's bang-per-buck ratio for that good should exceed their ``normal'' bang-per-buck ratio $e^{\lambda_i}$ by a common agent-independent factor $e^{\mu_j}$. On the above example, this leads the program to select the most natural equilibrium (and this remains the case if the caps and endowments are varied), suggesting that our convex program might define a desirable decision rule for selecting Lindahl equilibria.

\subsection{Computation and Experiments}

Let us briefly discuss how to solve \cref{eq:shmyrev cp capped}. 
Numerically, the program can be solved using any conic convex optimization solver supporting exponential cones, such as MOSEK, COPT, Clarabel, ECOS, or SCS, by formulating the program as
\begin{equation*}
	\begin{aligned}
		\max_{b \ge 0, x \ge 0, t} \quad & \textstyle\sum_{i \in N, j \in M_i} \left( b_{ij}\log v_{ij} - t_{ij} \right) \\[-1pt]
		\text{s.t.} \quad & (x_j, b_{ij}, -t_{ij}) \in \mathcal{K}_{\text{exp}} \text{ for all $i \in N$, $j \in M_i$} \\
		&\textstyle \sum_{j \in M_i} b_{ij} \le B_i \text{ for all $i \in N$ } \\
		&x_{j} = \textstyle \sum_{i \in N_j} b_{ij} \text{ for all $j \in M$} \\
		&x_{j} \le \text{cap}_j \text{ for all $j \in M$}
	\end{aligned}
\end{equation*}
where $\mathcal{K}_{\text{exp}} = \{ (x_1, x_2, x_3) : x_1 \ge x_2 e^{x_3/x_2}, x_1, x_2 \ge 0 \}$ is the (primal) exponential cone.
We built a simple online tool for solving moderate-size instances with the SCS solver~\citep{donoghue2016conic}, available at \href{https://dominik-peters.de/demos/lindahl/}{dominik-peters.de/demos/lindahl/}.
From a complexity-theoretic perspective, an $\epsilon$-optimal solution to \cref{eq:shmyrev cp capped} can be computed in polynomial time using the ellipsoid method
\citep[see, e.g.,][Theorem 13.1]{vishnoi2021algorithms}.

\begin{figure}[t]
	\centering
	\includegraphics[width=\linewidth]{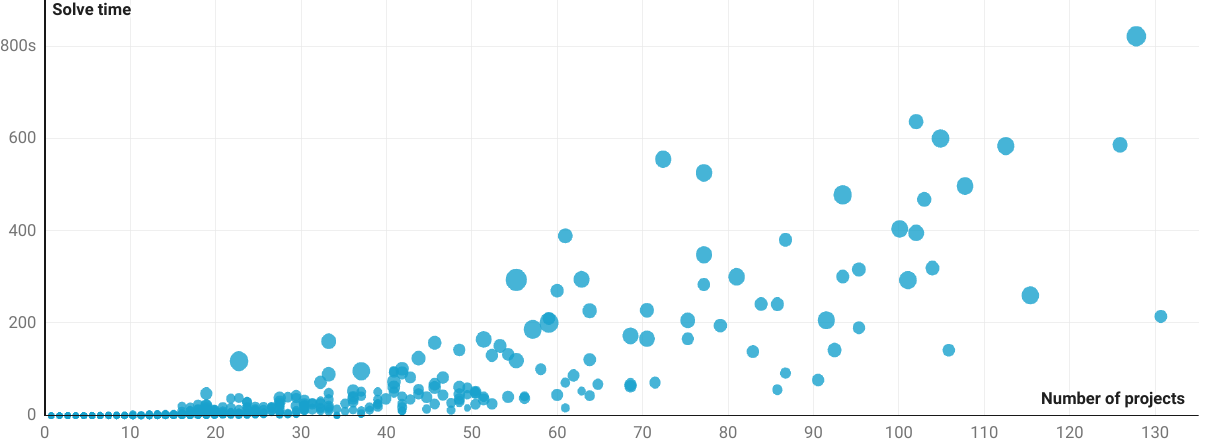}
	\caption{Results of our experiments on Pabulib instances, showing the solve time of the MOSEK solver as a function of the number of projects in the instance. The largest instances are from Warsaw and Amsterdam.}
	\label{fig:computation-time-scatter-plot}
\end{figure}

To evaluate the performance of computing Lindahl equilibrium via \cref{eq:shmyrev cp capped}, we implemented it using the MOSEK solver and applied it to the participatory budgeting datasets in the Pabulib repository \citep{faliszewski2023pabulib}. We find that the program can be solved quite quickly, with solve times shown in \Cref{fig:computation-time-scatter-plot}. The longest solve time we encountered was 822s (or 1489s including the time to write down the encoding) for an instance from Warsaw with 14\,897 voters (with 11\,426 distinct approval sets) and 134 projects.

For the uncapped setting, \citet[Section 4.2]{zhao2023analysis} present some experiments on the performance of the proportional response dynamics, and find that it outperforms several alternative solution methods.

\subsection{Computing Core Allocations for Separable Piecewise-Linear Concave Utilities}

We have set up the capped setting with the caps interpreted as an exogenous constraint. An alternative interpretation is as a capped utility function $u_i(x) = \sum_{j \in M} v_{ij}\min(x_j, \text{cap}_j)$. This view suggests a variety of generalizations: for example, we might want to allow different agents to specify different caps. We can generalize further to \emph{separable piecewise-linear concave utilities} (SPLC). These are utility functions that can be written as a sum over goods (separable), with the term corresponding to a good being a (non-decreasing) piecewise-linear concave function of $x_j$. 
See \Cref{fig:plc} for an example.

This class of utility functions is well-studied for private goods, both for Fisher markets and Arrow--Debreu exchange markets.
For these markets, just as for linear utilities, equilibrium exists and is rational under mild conditions; however, computing an equilibrium becomes PPAD-complete \citep{vazirani2011separablePLC,chen2009spending,deligkas2024constant}.
A complementary pivot algorithm for computing an equilibrium has been proposed \citep{garg2015complementary}. This algorithm is based on linear complementarity \citep{eaves1971basic,eaves1976finite}, which interestingly can also be used to show existence of Lindahl equilibria \citep[Appendix A]{munagala2022coremultilinear}.

We leave the problem of computing Lindahl equilibria for SPLC utilities open, but we show how our result for the capped setting can be used to at least compute an allocation in the core (up to any desired approximation factor). We do this by reducing an SPLC instance to a capped instance (with each piece of the piecewise-linear utilities becoming its own separate good), and show that a Lindahl equilibrium for this instance is core-stable with respect to the SPLC utilities.%
\footnote{However, it is not clear that it is also a Lindahl equilibrium with respect to SPLC utilities, since we assign different prices to different segments of a public good, while Lindahl equilibrium requires a single set of individualized prices for each good.}

We begin with formal definitions.
A function $f : [0, B] \to \mathbb R_{\ge 0}$ with $f(0) = 0$ is \emph{piecewise-linear concave} if it can be decomposed into a finite number of linear \emph{segments}, specified by their lengths $\ell^1, \dots, \ell^k$ (with $\ell^1 + \dots + \ell^k = B$) and slopes $s^1, \dots, s^k$ with $s^1 \ge \dots \ge s^k \ge 0$. Explicitly, for each $i \in [k]$, writing $a^i = \ell^1 + \dots + \ell^{i-1}$ for the left endpoint of the $i$th segment, we have $f(a^i + x) = f(a^i) + s^i x$ for $x \in [0, \ell^i]$. \Cref{fig:plc} shows an example.
\begin{figure}[t]
	\begin{tikzpicture}[yscale=0.8, font=\small]
		\draw[->] (0,0) -- (6.5,0) node[right]{$x$};
		\draw[->] (0,0) -- (0,3) node[above]{$f(x)$};

		\coordinate (A) at (0,0);
		\coordinate (B) at (1.8,1.5);
		
		\draw[dashed] (A) node[below]{$a^1$};
		
		\draw[thick,blue] (A) -- (B) node[midway, above, sloped]
		{slope = \(s^1\)};
		\draw[dashed] (B) -- ++(0,-1.5) node[below]{$a^2$};

		\coordinate (C) at (3.5,2.2);
		\draw[thick,blue] (B) -- (C) node[midway, above, sloped]
		{slope = \(s^2\)};
		\draw[dashed] (C) -- ++(0,-2.2) node[below]{$a^3$};

		\coordinate (D) at (6,2.5);
		\draw[thick,blue] (C) -- (D) node[midway, above, sloped]
		{slope = \(s^3\)};

		\node[left] at (0,0) {$0$};
		
		\draw [decorate, decoration={brace, amplitude=5pt}, yshift=2pt] (1.9,0) -- (3.4,0) 
		node[midway, above=5pt] {$\ell^2$};
		
	\end{tikzpicture}
	\caption{An example of a piecewise-linear concave function.}
	\label{fig:plc}
\end{figure}
A utility function $u_i$ is called \emph{separable piecewise-linear concave} (SPLC) if there exist piecewise-linear concave functions $f_{ij}$ for all $j \in M$ such that $u_i(x) = \sum_{j \in M} f_{ij}(x_j)$ for all allocations $x$.
Note that by subdividing segments if necessary, we may assume that for each project $j$, all agents $i$ agree on the total number $k_j$ of segments in $f_{ij}$ as well as their lengths. At the same time, we can choose this common subdivision in a minimal way, so that from one segment to the next, there is always at least one agent whose slope \emph{strictly} decreases. Let us write $s^t_{ij}$ for the slope of the $t$th segment of $f_{ij}$, and let us write $\ell_j^t$ for the length of the $t$th segment for project $j$.

Any instance of the public goods problem with SPLC utility functions can be translated into an instance with linear utility functions and with caps, using the following construction.
\begin{definition}
	[SPLC Instances]
	\label{def:splc-reduction}
	Suppose we are given an SPLC instance $I$ specified by the slopes $(s^t_{ij})_{i\in N, j \in M, t \in [k_j]}$ and lengths $(\ell^t_{j})_{j \in M, t \in [k_j]}$ of the segments.
	We construct a public goods instance $I'$ with linear utility functions on the same set of agents $N$ and with a new set of projects given by the disjoint union $M' = \bigcup_{j\in M} \{j^1, \dots, j^{k_j}\}$, with $\textup{cap}_{j^t} = \ell_j^t$. The project $j^t$ will describe how much of project $j$ will be funded in the area of its $t$th segment. Finally, we take valuations $v_{ij^t} = s^t_{ij}$.
\end{definition}

Let us say that an SPLC instance is \emph{well-behaved} if the derived instance is cap-sufficient in the sense of \Cref{def:cap-sufficient}. This is guaranteed to be the case, for example, if for every agent $i$, the total length of segments with positive slope across all projects is at least $B$. Similar sufficient conditions are used for private goods equilibria \citep[e.g.,][Section 2]{vazirani2011separablePLC}.

\begin{proposition}
	[Reduction from SPLC to caps]
	\label{thm:splc-reduction}
	Let $I$ be a well-behaved instance of the public goods problem with SPLC utilities. Then any zero-respecting Lindahl equilibrium for the instance $I'$ as constructed in \Cref{def:splc-reduction} can be transformed into a core allocation for the SPLC instance $I$.
\end{proposition}
\begin{proof}
	Suppose that $x'$ is a zero-respecting Lindahl equilibrium allocation for $I'$.
	By \Cref{cor:pareto}, $x'$ is Pareto-optimal. Note that we always have $s_{ij}^t \le s_{ij}^{t-1}$ by concavity, and for at least one agent the inequality is strict (by minimality of the chosen common subdivision).
	Thus, Pareto-optimality implies that if $x'_{j^t} > 0$ then $x'_{j^{t-1}} = \text{cap}_{j^{t-1}} = \ell_j^{t-1}$. This allows us to define an allocation $x$ for the instance $I$ by setting $x_j = \smash{\sum_{t = 1}^{k_j}} x'_{j^t}$. We now argue that $x$ is in the core.
	
	Suppose not, and there is some blocking coalition $S \subseteq N$ and objection $z = (z_j)_{j \in M}$ such that $\sum_{j \in M} z_j \le \sum_{i \in S} B_i$ and for all $i \in S$, we have $u_i(z) \ge u_i(x)$, with strict inequality for at least one $i \in S$. We construct a core objection $z'$ for $x'$, contradicting its core stability (\Cref{prop:lindahl is core}). For each $j \in M$, let $t$ be chosen minimal such that $\ell_j^1 + \cdots + \ell_j^{t-1} \le z_j < \ell_j^1 + \cdots + \ell_j^{t}$. Then set $z'_{j^r} = \text{cap}_{j^r} = \ell_j^r$ for $r = 1, \dots, t-1$ and $z'_{j^t} = z_j - (\ell_j^1 + \cdots + \ell_j^{t-1})$, as well as $z'_{j^r} = 0$ for $r = t+1, \dots, k_j$. Then it is easy to see that $u'_i(z') = u_i(z)$ for all $i$, and similarly $u'_i(x') = u_i(x)$ for all $i$, and thus $z'$ is a core deviation to $x'$, a contradiction.
\end{proof}

Thus, by (approximately) solving our convex program on the derived instance $I'$, we can find an (approximate) core allocation for the SPLC instance $I$.

\section{Conclusion}
\label{sec:conclusion}

We have developed a new class of convex programs that can be used to efficiently compute $\epsilon$-approximate Lindahl equilibria both in the uncapped and the capped setting. These new programs open up many opportunities for future research.

In the uncapped setting, our new program might lead to new proofs of known results for the well-studied maximum Nash welfare rule. This might include the result about participation incentives of \citet{brandl2022fundingpublicprojects} or the axiomatic characterization of  \citet{GuNe14a}. Perhaps our program could also shed light on the other uses of the Eisenberg--Gale program across statistics, information theory, and medical imaging, as discussed in \Cref{sec:related-work}. For the capped setting, our computability result has implications for the discrete public goods model, because it allows the efficient implementation of the 9.27-approximation to the core obtained by \citet{munagala2022coremultilinear}, rather than having to rely on their 67.37-approximation. \citet{munagala2022coremultilinear} used Lindahl equilibrium as a black box to obtain their approximation result; reasoning about the structure of our convex program might lead to even better bounds.

Since our focus has been on computational questions, we have not considered strategic aspects.  Lindahl equilibrium is well-known to have high informational requirements, and in particular we need to know the truthful valuations of the agents to compute it. Interpreted as a decision rule \citep{gul2020lindahlequilibriumcollectivechoice}, Lindahl equilibrium is not strategyproof and can be manipulated both in a free-riding sense \citep[Section 5.3]{brandl2021distribution}, and in some paradoxical ways \citep[Theorem 3(ii)]{aziz2020fairmixing}, even in the uncapped setting. Manipulability is unavoidable if one desires a Pareto-efficient and core-stable solution, both in the uncapped setting \citep[Theorems 2 and~3]{brandl2021distribution} and in the capped setting \citep[Theorem 6.2]{bei2024truthful}. These impossibilities apply even for approval (0/1) preferences. For more general linear utilities, strategyproofness is only attainable by dictatorial-type rules \citep{hylland1980lotteries}, even in the uncapped setting.

We leave several interesting technical questions open. Is the optimum of our program unique in utilities? This is known to be true for the uncapped setting, by strict convexity (in utilities) of the Eisenberg--Gale program. Can we develop first-order methods for the capped setting, or derive a natural dynamics converging to an equilibrium? Applying mirror descent to our program does not appear to lead to a nice closed-form update like in the uncapped setting. Are there versions of our convex program that work for non-linear utilities? Finally, can the cap constraint be generalized? For example, one could apply cap constraints on the total spending of \emph{sets} of public goods. This would allow us to model multi-issue and multi-round decision making settings \citep[see, e.g.,][Section 5]{banerjee2023proportionally}. It would also allow us to embed private goods in the model \citep[as in][]{conitzer2017fair}, and potentially connect the notions of Fisher market equilibrium and Lindahl equilibrium.

\begin{acks}
	We thank the anonymous reviewers at EC 2025 and JACM for useful feedback.
	We thank Matthias Greger for suggesting the perturbation strategy used in \Cref{thm:approximate-lindahl}.
	Christian Kroer was supported by the Office of Naval Research awards N00014-22-1-2530 and N00014-23-1-2374, and the National Science Foundation awards IIS-2147361 and IIS-2238960.
	\mbox{Dominik Peters} was supported by the French government under management of Agence Nationale de la Recherche as part of the ``Investissements d’avenir'' program, reference ANR-19-P3IA-0001 (PRAIRIE 3IA Institute) as well as the project ANR-22-CE26-0019 (CITIZENS).
\end{acks}


\newpage
\appendix
\counterwithin{lemma}{section} %
\iflatexml\else
    \renewcommand{\theHlemma}{\thesection.\arabic{lemma}}
\fi

\section{Convex Conjugates and Subdifferentials of Some Useful Functions}
\label{app:proofs}

\conjugatemaxexp*
\begin{proof}
		We compute the convex conjugate of $g$ using standard formulas for the conjugate of a separable function, rescalings of a function, and the exponential function \citep[see, e.g.,][Sections 4.3 and 4.4]{beck2017firstorder}. We also use Sion's minimax theorem \citep{komiya1988} which states that if $X$ is convex, $Y$ is convex and compact, and $f$ is a real-valued function on $X \times Y$ that is concave in its first argument and convex in its second argument, then $\sup_{x\in X}\min_{y\in Y} f(x,y) = \min_{y\in Y} \sup_{x\in X} f(x,y)$. Finally, given $\beta = (\beta_{ij})_{j \in M, i \in N_j}$, we write $\beta_j = (\beta_{ij})_{i \in N_j}$.
	
	Putting all of this together, we derive that
	{\allowdisplaybreaks
		\begin{align*}
			g^*(\beta) &= \sup_{q}\, \langle \beta, q \rangle - g(q) \\
			&= \sup_{q} \min_{j \in M}\, \langle \beta, q \rangle - B \cdot g_j(q_j) \tag{definition of $g$} \\
			&= \sup_{q} \min_{x \in B \cdot \Delta^{m}} \langle \beta, q \rangle - \sum_{j \in M} x_j \cdot g_j(q_j) \tag{minimum attained at a vertex} \\
			&= \min_{x \in B \cdot \Delta^{m}} \sup_{q}\, \langle \beta, q \rangle - \sum_{j \in M} x_j \cdot g_j(q_j), \tag{Sion's minimax theorem} \\
			&= \min_{x \in B \cdot \Delta^{m}} \sum_{j\in M}\sup_{q_j}\, \langle \beta_j, q_j \rangle - x_j \cdot g_j(q_j), \tag{conjugate of a separable function} \\
			&= \min_{x \in B \cdot \Delta^{m}} \sum_{j \in M} x_j \cdot g_j^*(\beta_j/x_j) \tag{$(\alpha f)^*(y) = \alpha f^*(y/\alpha)$} \\
				&= \min_{x \in B \cdot \Delta^{m}} \sum_{j \in M} x_j \cdot \left( \sum_{i\in N_j} \tfrac{\beta_{ij}}{x_j} \log \tfrac{\beta_{ij}}{x_j} - \tfrac{\beta_{ij}}{x_j} \right) \tag{$\exp^*(y) = y \log y - y$} \\
				&= \min_{x \in B \cdot \Delta^{m}} \sum_{j \in M}\sum_{i \in N_j} \beta_{ij} \log \tfrac{\beta_{ij}}{x_j} - \beta_{ij} \\
				&= \sum_{j \in M}\sum_{i \in N_j} \beta_{ij} \log \tfrac{\beta_{ij}}{\tilde x_j(\beta)} - \beta_{ij}, \tag{\Cref{lem:simple simplex optimizations}(b)}
	\end{align*}}
	as required.
\end{proof}

In order to prove \cref{lem:shmyrev subdiff}, we will need the following lemma regarding the convex conjugate of $h$.
\begin{lemma}
	\label{lem:convex-conjugate-h}
	Let $h : \mathbb{R}^n_{\ge 0} \to \mathbb R$ be the function given by $h(x) = \sum_i x_i \log\frac{x_i}{\sum_k x_k}$ for $x \neq 0$ and $h(0) = 0$ by the convention that $0 \log \frac00 = 0$. Its convex conjugate $h^*$ is the indicator of the set $\{ y : \sum_i e^{y_i} \leq 1 \}$, so $h^*(y) = 0$ if $\sum_i e^{y_i} \leq 1$ and $h^*(y) = +\infty$ otherwise.
\end{lemma}
\begin{proof}
	Let $\varphi(p) = \sum_i p_i \log p_i$ be the negative entropy function, and recall that the convex conjugate of the negative entropy function is the log-sum-exp function $\varphi^*(y) = \log \sum_i e^{y_i}$ \citep[see, e.g., ][Section 4.4.8]{beck2017firstorder}.
	Note that $h(x) = \|x\|_1 \varphi(x/\|x\|_1)$.
	Thus, writing $S = \|x\|_1$, we have
	\begin{align*}
		h^*(y) 
		&= \sup_{x \in \mathbb{R}^n_{\ge 0}} \langle y, x \rangle - h(x) \tag{definition of the convex conjugate} \\		
		&= \sup_{S \ge 0} \sup_{p \in \Delta^n} \langle y, Sp \rangle - S \varphi(p) \\
		&= \sup_{S \ge 0} S \sup_{p \in \Delta^n} \left( \langle y, p \rangle - \varphi(p) \right)  \\
		&= \sup_{S \ge 0} S \cdot \varphi^*(y) \tag{definition of the convex conjugate} \\
		&= \textstyle \sup_{S \ge 0} S \log \sum_i e^{y_i}.
	\end{align*}
	Thus, $h^*(y) = 0$ if $\sum_i e^{y_i} \leq 1$ (with the supremum attained at $S = 0$) and $h^*(y) = +\infty$ otherwise (letting $S \to +\infty$).
\end{proof}

\shmyrevsubdiff*
\begin{proof}
	We can write $f(b) = -\sum_{i\in N, j\in M_i} b_{ij} \log v_{ij} + \sum_{j \in M} h(b_j)$ where $b_j$ is the vector $(b_{ij})_{i \in N_j}$ and $h$ is the function defined as $h(x) = \sum_i x_i \log\frac{x_i}{\sum_k x_k}$ when all $x_i$ are non-negative, $h(0) = 0$ by the convention that $0 \log \frac00 = 0$, and $+\infty$ otherwise.
	
	We begin by computing the subdifferential of $h$, using that its convex conjugate $h^*(y)$ is $h^*(y) = 0$ if $\sum_i e^{y_i} \leq 1$ and $h^*(y) = +\infty$ otherwise, as shown in \Cref{lem:convex-conjugate-h}.
	Hence
	\begin{align*}
		\partial h(x)
		= \{ g : h(x) + h^*(g) = \langle x, g\rangle \}
		= \{ g : \textstyle\sum_i e^{g_i} \le 1 \text{ and } h(x) = \langle x, g\rangle \}.
	\end{align*}
	For $x = 0$, we have $h(x) = \langle x, g \rangle$ for all $g$ since $h(0) = 0$. Hence $\partial h(0) = \{ g : \sum_i e^{g_i} \le 1 \}$.
	For $x \in \mathbb{R}^n_{\ge 0}$ such that $x_i > 0$ for all $i$, note that $h$ is differentiable, so $\partial h(x) = \{ (\log \frac{x_i}{\sum_k x_k})_{i = 1}^n \}$.
	For $x \in \mathbb{R}^n_{\ge 0}$ such that $x \neq 0$ but $x_i = 0$ for some $i$, we have $\partial h(x) = \emptyset$:
	if $g \in \partial h(x)$, then for every $j$ with $x_j > 0$, we must have $g_j = \log \frac{x_j}{\sum_k x_k}$ (as in the fully interior case), and therefore $\sum_{j : x_j > 0} e^{g_j} = 1$.
	Together with $\sum_i e^{g_i} \le 1$, this forces $e^{g_i} = 0$ for each index $i$ with $x_i = 0$, i.e., $g_i = -\infty$, impossible for $g \in \mathbb{R}^n$.
	
	From this characterization of the subdifferential of $h$, the subdifferential of $f$ can be determined immediately because the subdifferential of a sum of convex functions is equal to the sum of the subdifferentials and noting that $\partial (-b_{ij} \log v_{ij}) = \{-\log v_{ij}\}$.
\end{proof}

\section{Approximately Optimal Solutions}
\label{app:approximations}

In this section, we will study properties of solutions that are $\epsilon$-optimal.
Given the problem of minimizing a function $f$ over some feasible set, we say that a feasible solution $x$  is \emph{$\epsilon$-optimal} if $f(x) \le f(y) + \epsilon$ for all feasible solutions $y$.

Our main tool for this analysis is a version of the KKT theorem for \emph{approximately} optimal solutions. It is stated in terms of the $\epsilon$-subdifferential of a convex function, which is defined as
\begin{align*}
	\partial_{\epsilon} f(x) 
	&= \{ g \in \mathbb{R}^n : f(y) \ge f(x) + \langle g, y - x \rangle - \epsilon \text{ for all } y \in \mathbb{R}^n \} \\
	&= \{ g \in \mathbb{R}^n : f(x) + f^*(g) \le \langle g, x \rangle + \epsilon \}
\end{align*}
where the second line gives a standard equivalent definition using the convex conjugate $f^*$ of $f$.
We state a version specialized to affine constraints for simplicity.

\begin{theorem}[\citealp{strodiot1983epsilon}, Thm 2.4 and Thm 3.2]
	\label{thm:epsilon-kkt}
	Let $x^*$ be an $\epsilon$-optimal solution to the program
	\[ \textup{minimize } f(x) \textup{ subject to } h_i(x) \le 0 \textup{ for $i = 1, \dots, m$} \]
	where $f : \mathbb{R}^n \to (-\infty, \infty]$ is a proper convex function and $h_i : \mathbb{R}^n \to \mathbb{R}$, $i = 1, \dots, m$, are affine functions.
	Assume the program satisfies Slater's constraint qualification.
	Then there exist multipliers $\lambda_1, \dots, \lambda_m \ge 0$ and a subgradient $g \in \partial_\epsilon f(x^*)$ such that
	\[
	\textstyle
	0 = g + \sum_{i = 1}^m \lambda_i \nabla h_i
	\quad
	\text{and}
	\quad
	-\epsilon \le \sum_{i = 1}^m \lambda_i h_i(x^*) \le 0.
	\]
	In addition, the multipliers $(\lambda_1, \dots, \lambda_m)$ form an $\epsilon$-optimal solution to the dual program
	\[
	\textup{maximize } \inf\{f(x) + \langle \lambda, h(x) \rangle : x \in \mathbb{R}^n \} \textup{ subject to } \lambda_j \ge 0 \textup{ for $j = 1, \dots, m$}
	\]
	where $h(x) = (h_1(x), \dots, h_m(x))$.
\end{theorem}

It will be useful to have an upper bound on the size of the multipliers provided by this KKT theorem. The following bound is similar to results of \citet[Section 2]{uzawa1958iterative}, as well as \citet[p. 313, Remark 2.3.3]{hiriart2013convex} and \citet[Lemma 1]{nedic2009approximate}, but ours applies to multipliers associated with \emph{approximate} solutions. It shows that the optimal KKT multiplier of a constraint ``$h_i \le 0$'' is bounded if we can identify a feasible solution $\bar x$ that has a good objective function value $f(\bar x)$ and under which the constraint has enough slack $-h_i(\bar x)$.

\begin{lemma}
	\label{lem:multiplier-bound}
	Assume that the optimization program in \Cref{thm:epsilon-kkt} has finite optimum objective value $f^*$, and that it satisfies Slater's constraint qualification. Let $\bar x$ be a feasible solution, and consider one of the constraints $h_i(x) \le 0$ of the optimization program. Let $\lambda = (\lambda_1, \dots, \lambda_m)$ be the multipliers guaranteed by \Cref{thm:epsilon-kkt} for an $\epsilon$-optimal solution. Then
	\[
	\lambda_i \le \frac{f(\bar x) - f^* + \epsilon}{-h_i(\bar x)}.
	\]
\end{lemma}
\begin{proof}
	According to \Cref{thm:epsilon-kkt}, the multipliers $\lambda \ge 0$ form an $\epsilon$-optimal solution to the dual program. Write $g(\lambda) = \inf\{f(x) + \langle \lambda, h(x) \rangle : x \in \mathbb{R}^n \}$ for its objective function and $g^*$ for its maximum value subject to $\lambda \ge 0$. Then
	\begin{align*}
		f^*&= g^* \tag{strong duality} \\
		&\le  g(\lambda) + \epsilon \tag{$\epsilon$-optimality of $\lambda$} \\
		&= \inf\{f(x) + \langle \lambda, h(x) \rangle : x \in \mathbb{R}^n \} + \epsilon \\
		&\le f(\bar x) + \langle \lambda, h(\bar x) \rangle + \epsilon \\
		&\le f(\bar x) + \lambda_i h_i(\bar x) + \epsilon \tag{feasibility of $\bar x$, so $\lambda_j h_j(\bar x) \le 0$ for all $j$} \\
		&= f(\bar x) - \lambda_i (-h_i(\bar x)) + \epsilon.
	\end{align*}
	Rearranging, the result follows.
\end{proof}

\subsection{Uncapped Public Goods}
\label{app:approximations-uncapped}

In this section, we will show that in the uncapped public goods setting, approximate solutions to our program also approximate maximum Nash social welfare and the proportional fairness criterion.

For the former result, we include here the proof of \Cref{lem:nash-bounded-by-shmyrev}, which was used in the main body of the paper to prove \Cref{cor:convergence-rate-nash} (convergence rate of the proportional response dynamics in terms of Nash social welfare).

\nashboundedbyshmyrev*
\begin{proof}
	Write $x = x(b)$ for the allocation induced by $b$, and consider the following optimization problem:
	\begin{program}
		\begin{aligned}
			\min_{b \ge 0}\quad &f^{\dagger}(b) \defeq - \sum_{\mathclap{i\in N, j\in M_i}} \: b_{ij} \left( \log v_{ij} - \log \left(b_{ij} / x_j\right) \right) \\[3pt]
			\text{s.t.}\quad & \sum_{j \in M_i} b_{ij} = B_i,\ \forall i\in N
		\end{aligned}
		\label{eq:shmyrev fixed allocation}
	\end{program}
	This is the same optimization problem as \Cref{eq:shmyrev cp min}, except that we replaced the term $x_j(b)$ by the constant $x_j$. Clearly, $b$ is a feasible solution for \Cref{eq:shmyrev fixed allocation} and $f(b) = f^{\dagger}(b)$. Let $b^*$ be an optimal solution to \Cref{eq:shmyrev fixed allocation}. Then $f^{\dagger}(b^*) \le f^{\dagger}(b)$. Noting that \Cref{eq:shmyrev fixed allocation} decomposes into a separate optimization problem for each $i \in N$, from \Cref{lem:simple simplex optimizations}(a), we can write the optimal solution $b^*$ explicitly as
	\begin{equation}
		\label{eq:characterize-optimum}
		b^*_{ij} = B_i \frac{v_{ij} x_j}{\langle v_i , x\rangle}.
	\end{equation}
	Hence
	\begin{align*}
		f^{\dagger}(b^*) 
		&= - \sum_{i \in N, j \in M_i} B_i \frac{v_{ij} x_j}{\langle v_i , x\rangle} \log \left( \frac{v_{ij}x_j}{B_i \frac{v_{ij} x_j}{\langle v_i , x\rangle}} \right) \\
		&= - \sum_{i \in N} B_i \log \left( \frac{\langle v_i , x\rangle}{B_i} \right) \sum_{j \in M_i} \frac{v_{ij} x_j}{\langle v_i , x\rangle} \\
		&= - \sum_{i \in N} B_i \log \langle v_i , x\rangle + \sum_{i \in N} B_i \log B_i
		= \varphi(x) + \sum_{i \in N} B_i \log B_i.
	\end{align*}
	Hence we have shown that $f(b) = f^{\dagger}(b) \ge f^{\dagger}(b^*) = \varphi(x(b)) + \sum_{i \in N} B_i \log B_i$, as required.
	
	Finally, if $b$ is an optimal solution to \Cref{eq:shmyrev cp min}, then the KKT conditions of that program imply that $b$ satisfies Equation \eqref{eq:characterize-optimum} (see \Cref{thm:lindahl shmyrev capped}), and thus $b^* = b$, so the claimed inequality holds with equality.
\end{proof}

For the latter result, we now prove that any outcome that approximately maximizes Nash social welfare also approximates the proportional fairness criterion.
We will do so via the $\epsilon$-KKT conditions from \Cref{thm:epsilon-kkt}. To apply these conditions, we will need the $\epsilon$-subdifferential of the negative logarithm, since this function appears in the definition of Nash social welfare. We are not aware of a formula for this $\epsilon$-subdifferential having appeared in the literature, so we derive it from scratch.
\begin{lemma}
	For $0 \le \epsilon < \frac12$, the $\epsilon$-subdifferential of the negative logarithm is
	\[
	\partial_{\epsilon}(-\log x)(x)
	=\Bigl[
	\tfrac{1}{x} W_{-1}\bigl(-e^{-(\epsilon+1)}\bigr),
	\tfrac{1}{x} W_{0}\bigl(-e^{-(\epsilon+1)}\bigr)
	\Bigr]
	\subseteq
	\Bigl[
	\tfrac{-1-2\sqrt{\epsilon}}{x},
	\tfrac{-1+\sqrt{2\epsilon}}{x}
	\Bigr],
	\]
	where $W$ refers to the \emph{Lambert $W$ function} \citep[see][or \href{https://en.wikipedia.org/wiki/Lambert_W_function}{Wikipedia}]{corless1996lambert}.
\end{lemma}
\begin{proof}
	Let $0 \le \epsilon < \frac12$.
	We begin by considering the function $f(u) = u - \log u$ for $u > 0$, depicted below.
	
	\begin{center}
		\begin{tikzpicture}
			\begin{axis}[
				xlabel={$u$},
				ylabel={$f(u) = u - \log(u)$},
				domain=0.01:2,
				samples=100,
				grid=major,
				width=7cm,
				height=4cm,
				axis lines=center,
				enlargelimits=false,
				xmin=0, xmax=2.2,
				ymin=0, ymax=3.4,
				xtick={0,0.5,1,1.5,2},
				ytick={1,2,3},
				tick label style={font=\small},
				label style={font=\small},
				ylabel style={fill=white,xshift=7pt}
				]

				\addplot[blue, thick, smooth] {x - ln(x)};

				\addplot[black, mark=*, mark size=2pt] coordinates {(1,1)};

				\node[above right] at (1,1) {$(1,1)$};
			\end{axis}
		\end{tikzpicture}
	\end{center}
	
	Its derivatives are $f'(u) = 1 - \frac1u$ and $f''(u) = 1/u^2 > 0$. Thus, the function has its minimum at $u = 1$ with $f(1) = 1$, and is strictly decreasing on $(0,1)$ and strictly increasing on $(1, +\infty)$.
	
	We will show that
	\begin{equation}
		\label{eq:log-subdiff-condition-1}
		f(u) \le 1 + \epsilon 
		\iff 
		u \in [-W_{0}(-e^{-(1+\epsilon)}), -W_{-1}(-e^{-(1+\epsilon)})],
	\end{equation}
	where $W$ refers to the Lambert $W$ function.
	By continuity of $f$ and looking at its derivatives, the range of $u \in (0, +\infty)$ for which $f(u) \le 1 + \epsilon$ is clearly an interval, whose endpoints are the solutions to the equation $f(u) = 1 + \epsilon$. Note that 
	\begin{equation}
		\label{eq:log-subdiff-condition-2}
		u - \log u = 1 + \epsilon \Leftrightarrow \log u - u = -(1 + \epsilon) \Leftrightarrow u e^{-u} = e^{-(1+\epsilon)} \Leftrightarrow -u e^{-u} = -e^{-(1+\epsilon)}.
	\end{equation}
	
	Now, for given $z \in (-1/e, 0)$, the equation $we^w = z$ has two real solutions which are $W_{-1}(z)$ and $W_0(z)$ by definition of the Lambert $W$ function, with $W_{-1}(z) < -1 < W_{0}(z) < 0$. Thus,  condition \eqref{eq:log-subdiff-condition-2} is equivalent to
	\[
	u \in \{-W_{-1}(-e^{-(1+\epsilon)}), -W_{0}(-e^{-(1+\epsilon)})\},
	\]
	thereby establishing \eqref{eq:log-subdiff-condition-1}.
	
	It is useful to also give a looser version of \eqref{eq:log-subdiff-condition-1} that is phrased in terms of elementary functions. For this, we will use the following inequalities of the logarithm:
	\begin{alignat}{3}
		\log(1-t) &< -t - \tfrac{t^2}{2} \quad && \text{for $t > 0$,} \label{eq:log-inequality-1} \\
		\log(1+y) &< y - \tfrac{y^2}{4}&& \text{for $y \in (0, \sqrt{2})$.}
		\label{eq:log-inequality-2}
	\end{alignat}
	These can be obtained from Taylor expansions or by showing that the functions $h(t) = \log(1-t) - (-t - \frac{t^2}{2})$ and $k(y) = \log(1+y) - (y - \frac{y^2}{4})$ are strictly negative on the relevant domain, by analyzing the monotonicity properties of those functions through their derivatives.
	
	Now note that
	\begin{align*}
		f(1-\sqrt{2\epsilon}) &= (1-\sqrt{2\epsilon}) - \log(1-\sqrt{2\epsilon}) \overset{\eqref{eq:log-inequality-1}}{>} (1-\sqrt{2\epsilon}) + \sqrt{2\epsilon} + \epsilon = 1 + \epsilon \\
		\intertext{and}
		f(1+2\sqrt{\epsilon}) &= 1+2\sqrt{\epsilon} - \log(1+2\sqrt{\epsilon}) \overset{\eqref{eq:log-inequality-2}}{>} 1+2\sqrt{\epsilon} - 2\sqrt{\epsilon} + \epsilon = 1 + \epsilon,
	\end{align*}
	where we could invoke \eqref{eq:log-inequality-2} because $y = 2\sqrt{\epsilon} < 2 \sqrt{1/2} = \sqrt{2}$ since $\epsilon < \frac12$ by assumption.
	
	From this, it follows that
	\begin{equation}
		\label{eq:log-subdiff-condition-3}
		f(u) \le 1 + \epsilon 
		\implies 
		u \in [1 - \sqrt{2\epsilon}, 1 + 2 \sqrt{\epsilon}].
	\end{equation}
	A similar lower bound is obtained by \citet[Theorem 1]{chatzigeorgiou2013}.
	
	We now turn to determining the $\epsilon$-subdifferential of $g(x) = -\log x$. Recall that $g^*(y) = -1 -\log(-y)$. The $\epsilon$-subdifferential can be written using the convex conjugate: for all $x > 0$,
	\begin{align*}
		\partial_{\epsilon}g(x) 
		&= \{ y : g(x) + g^*(y) \le xy + \epsilon \} \tag{definition} \\
		&= \{ y : -\log(x) - 1 - \log(-y) \le xy + \epsilon \} \\
		&= \{ y : -xy - \log(-xy) \le 1 + \epsilon \} \\
		&= \{ y : -xy \in [-W_{0}(-e^{-(1+\epsilon)}), -W_{-1}(-e^{-(1+\epsilon)})] \} \tag{by \eqref{eq:log-subdiff-condition-1}} \\
		&= [\tfrac{1}{x}W_{-1}(-e^{-(1+\epsilon)}), \tfrac{1}{x} W_{0}(-e^{-(1+\epsilon)})] \\
		&\subseteq [\tfrac{-1-2\sqrt{\epsilon}}{x}, \tfrac{-1+\sqrt{2\epsilon}}{x}], \tag{by \eqref{eq:log-subdiff-condition-3}}
	\end{align*}
	as desired.
\end{proof}

Recall from \Cref{sec:convergence-pf-core} that the PF value of an allocation $x$ is
\[
\textup{PF}(x) = \max_{j \in M} \sum_{i \in N} B_i \frac{v_{ij}}{u_i(x)}.
\]
We will show that an allocation that approximately maximizes the Nash social welfare has a PF value not much higher than 1.
Recall that we write $\varphi(x) = -\sum_{i \in N} B_i \log \langle v_i, x\rangle$ for the (negative) Nash social welfare of an allocation $x$. Let $\varphi^*$ denote Nash social welfare of the optimum allocation (i.e., the minimum).
Write $B_{\min} = \min_{i \in N} B_i$.

\approxnashpf*
\begin{proof}
	To apply the $\epsilon$-KKT conditions, we will need some elementary calculus rules for the $\epsilon$-subdifferential \citep[see also \citealp{dhara2011optimality}, Section 2.6]{hiriarturruty1982}: for $\lambda > 0$, we have $\partial_\epsilon (\lambda f)(x) = \lambda \partial_{\epsilon/\lambda} f (x)$, and for two convex functions $f_1, f_2$ defined on the same domain, we have 
	\begin{align}
	\textstyle
	\partial_\epsilon(f_1 + f_2)(x) = \bigcup_{\epsilon_1, \epsilon_2 \ge 0, \epsilon_1 + \epsilon_2 \le \epsilon} \partial_{\epsilon_1} f_1(x) + \partial_{\epsilon_2} f_2(x).
    \label{eq:eps subdiff summation rule}
	\end{align}
	Finally, for an affine function, the $\epsilon$-subdifferential is singleton: $\partial_{\epsilon}(ax + b)(x) = \{a\}$ for all $\epsilon \ge 0$.
	
	Now consider the following program for maximizing Nash social welfare.
	\begin{equation*}
		\begin{aligned}
			\max_{x, u}&\ \sum_{i\in N} B_i \log u_i \\
			\text{s.t.}& \sum_{j\in M} x_j \leq B \\
			& u_i \le \sum_{j \in M} v_{ij} x_j && \text{for all $i \in N$}\\
			& x_j \ge 0 && \text{for all $j \in M$}
		\end{aligned}
	\end{equation*}
	
	Writing this in the notation of \Cref{thm:epsilon-kkt}, we are looking at the following program:
	\begin{equation*}
		\begin{aligned}
			\min_{x, u}&- \sum_{i\in N} B_i \log u_i \\
			\text{s.t.}& \sum_{j\in M} x_j - B \leq 0 \\
			& u_i - \sum_{j \in M} v_{ij} x_j \le 0 && \text{for all $i \in N$}\\
			& -x_j \le 0 && \text{for all $j \in M$}
		\end{aligned}
	\end{equation*}
	Consider an $\epsilon$-optimal solution $(x^*, u^*)$ to this program with $\sum_{j \in M} x_j^* = B$. We may assume that $u_i^* = \sum_{j \in M} v_{ij} x_j^*$ for all $i \in N$ since if we had $u_i^* < \sum_{j \in M} v_{ij} x_j^*$ then we could increase $u_i^*$ while retaining feasibility and improving the objective value, thus preserving $\epsilon$-optimality. Let $\lambda, (\mu_i)_{i \in N}, (\eta_j)_{j \in M} \ge 0$ be the multipliers promised by \Cref{thm:epsilon-kkt} corresponding to the three types of constraints.
	
	Write $f(u) = - \sum_{i\in N} B_i \log u_i$ for the objective function. Then
	\begin{align*}
	\partial_{\epsilon} f(u^*) 
	&\subseteq 
	\bigcup_{\substack{(\bar \epsilon_i)_{i \in N} \ge 0 \\ \sum_{i \in N} \bar \epsilon_i \le \epsilon}}
	\{ (\bar u_i)_{i \in N} 
	: 
	B_i\tfrac{-1-2\sqrt{\bar \epsilon_i / B_i}}{u_i^*}
	\le \bar u_i \le
	B_i\tfrac{-1+\sqrt{2\bar \epsilon_i / B_i}}{u_i^*}
	\text{ for all $i \in N$}
	\} \\
	&\subseteq 
	\{ (\bar u_i)_{i \in N} 
	: 
	B_i\tfrac{-1-2\sqrt{\epsilon / B_i}}{u_i^*}
	\le \bar u_i \le
	B_i\tfrac{-1+\sqrt{2\epsilon / B_i}}{u_i^*}
	\text{ for all $i \in N$}
	\}.
	\tag{using $\bar \epsilon_i \le \epsilon$}
	\end{align*}
	Thus, from the stationarity condition in \Cref{thm:epsilon-kkt} applied to the $u_i^*$ variable, it follows that there exist $(\bar u_i)_{i \in N}$ 
	such that
	\begin{equation}
		\label{eq:approximate-nash-stationarity-u-1}
		0 = \bar u_i + \mu_i \quad \text{for all $i \in N$,}
	\end{equation}
	where $(\bar u_i)_{i \in N}$ is an element of $\partial_{\epsilon} f(u^*)$ and thus satisfies
	\begin{equation}
		\label{eq:approximate-nash-stationarity-u-2}
		B_i\tfrac{-1-2\sqrt{\epsilon / B_i}}{u_i^*}
		\le \bar u_i \le
		B_i\tfrac{-1+\sqrt{2\epsilon / B_i}}{u_i^*} \quad \text{for all $i \in N$}.
	\end{equation}
	By the stationarity condition applied to the variable $x_j$, we get
	\begin{equation}
		\label{eq:approximate-nash-stationarity-x-1}
		0 = \lambda - \sum_{i \in N} v_{ij} \mu_i - \eta_j \quad \text{for all $j \in M$.}
	\end{equation}
	Rearranging and multiplying by $x_j^*$ we deduce
	\begin{equation}
		\label{eq:approximate-nash-stationarity-x-2}
		\lambda x_j^* = \sum_{i \in N} \mu_i v_{ij} x_j^* + \eta_j x_j^* \quad \text{for all $j \in M$.}
	\end{equation}
	Note that the complementary slackness condition of \Cref{thm:epsilon-kkt} (which in weaker form guarantees ``$-\epsilon \le \sum_{i = 1}^m \lambda_i h_i(x^*)$'') implies $-\epsilon \le \sum_{j \in M} \eta_j(-x_j^*)$ (since we can drop terms corresponding to other constraints because they are all non-positive). Hence $\sum_{j \in M} \eta_j x_j^* \le \epsilon$.
	
	Now, summing \eqref{eq:approximate-nash-stationarity-x-2} over all $j \in M$ and using that $\sum_{j \in M} x_j^* = B$, we obtain
	\begin{align}
		\lambda B = \lambda \sum_{j \in M} x_j^*
		&= \sum_{i \in N} \mu_i u_i^* + \sum_{j \in M} \eta_j x_j^* \nonumber \\
		&\overset{\eqref{eq:approximate-nash-stationarity-u-1}}{=} \sum_{i \in N} (-\bar u_i) u_i^* + \sum_{j \in M} \eta_j x_j^* \nonumber  \\
		&\overset{\eqref{eq:approximate-nash-stationarity-u-2}}{\le} \sum_{i \in N}B_i (1+2\sqrt{\epsilon / B_i}) + \epsilon.
		\label{eq:approximate-nash-stationarity-x-3}
	\end{align}
	Thus, from \eqref{eq:approximate-nash-stationarity-x-1}, we get
	\begin{equation}
		\label{eq:approximate-nash-conclusion-1}
		\sum_{i \in N} v_{ij} \mu_i = \lambda - \eta_j \le \lambda \overset{\eqref{eq:approximate-nash-stationarity-x-3}}{\le} \frac1B \sum_{i \in N}\left( B_i (1+2\sqrt{\epsilon / B_i}) + \epsilon \right).
	\end{equation}
	Also
	\begin{equation}
		\label{eq:approximate-nash-conclusion-2}
		\sum_{i \in N} v_{ij} \mu_i 
		\overset{\eqref{eq:approximate-nash-stationarity-u-1}}{=} 
		\sum_{i \in N} v_{ij} (-\bar u_i) 
		\overset{\eqref{eq:approximate-nash-stationarity-u-2}}{\ge}
		\sum_{i \in N} v_{ij} B_i\tfrac{1-\sqrt{2\epsilon / B_i}}{u_i^*} = \sum_{i \in N} B_i \frac{v_{ij}}{u_i^*} (1-\sqrt{2\epsilon / B_i}).
	\end{equation}
	Putting \eqref{eq:approximate-nash-conclusion-1} and \eqref{eq:approximate-nash-conclusion-2} together, we obtain
	\[
	\sum_{i \in N} B_i \frac{v_{ij}}{u_i^*} (1-\sqrt{2\epsilon / B_i})
	\le 
	\frac1B \sum_{i \in N}\left( B_i (1+2\sqrt{\epsilon / B_i}) + \epsilon \right)
	\]
	which we can simplify to 
	\[
	\sum_{i \in N} B_i \frac{v_{ij}}{u_i^*} (1-\sqrt{2\epsilon / B_i})
	\le 
	1 + \sum_{i \in N} \left( 2\sqrt{\epsilon B_i}/B  + \epsilon/B \right).
	\]
	Using $B_i \le B$ and noting that the assumption $\epsilon < \frac{B_{\min}}{2}$ implies that $1-\sqrt{2\epsilon / B_i} > 1-\sqrt{2\epsilon / B_{\min}} > 0$, we obtain
	\[
	\sum_{i \in N} B_i \frac{v_{ij}}{u_i^*}
	\le 
	\frac{1 + 2n\sqrt{\epsilon / B} + n\epsilon/B}{1-\sqrt{2\epsilon / B_{\min}}}.
	\]
	Since this holds for all $j \in M$, this proves the desired bound on the PF value of $x^*$.
\end{proof}

\subsection{Capped Public Goods}
\label{app:approximations-capped}

In this section, we use $\epsilon$-KKT conditions to show that $\epsilon$-optimal solutions to our convex program form approximate Lindahl equilibria. In the proof, we will need a version of Pinsker's inequality (connecting $L_1$ distance and KL divergence) that works even if one distribution is sub-normalized.
\begin{lemma}
	\label{lem:sub-pinsker}
	Let $p = (p_1, \dots, p_n) \ge 0$ and $q = (q_1, \dots, q_n) \ge 0$ be non-negative vectors with $\sum_i p_i = 1$ (so $p$ is a probability distribution) and $\sum_i q_i \le 1$ (so $q$ is a sub-probability distribution). Then $\| p - q \|_1 \le \sqrt{2\textup{KL}(p \| q)}$.
\end{lemma}
\begin{proof}
	Write $s = \sum_i q_i \le 1$, and consider the two probability distributions $\hat p = (p_1, \dots, p_n, 0)$ and $\hat q = (q_1, \dots, q_n, 1 - s)$. Note that $\textup{KL}(\hat p \| \hat q) = \textup{KL}(p \| q) + 0 \log \frac{0}{1-s} = \textup{KL}(p \| q) $. Now, by the standard version of Pinsker's inequality, $\| \hat p - \hat q \|_1 \le \sqrt{2\textup{KL}(\hat p \| \hat q)}$. Hence
	\[
		\| p - q \|_1
		= \textstyle\sum_i |p_i - q_i|
		 = \| \hat p - \hat q \|_1 - |(1 - s) - 0|
		\le \| \hat p - \hat q \| _1
		\le \sqrt{2\textup{KL}(\hat p \| \hat q)}
		= \sqrt{2\textup{KL}(p \| q)}.
		\qedhere
	\]
\end{proof}

As the first part of the argument, we show that approximate solutions satisfy an approximate utility maximization condition, under the assumption that each project receives some minimum amount of spending, i.e., $x_j \ge t$ for all $j \in M$. Write $\text{cap}_{\text{sum}} = \sum_{j \in M} \text{cap}_j$.\footnote{For computational complexity purposes, we assume that the caps are given as part of the input and in particular that $\text{cap}_{\text{sum}} < \infty$. Uncapped projects thus need to be given an explicit cap of $B$ or higher. This can change the equilibrium set compared to a cap of $+\infty$.}

\begin{lemma}
	\label{thm:approximate-lindahl-lower-bound}
	Let $t > 0$ and $\epsilon \ge 0$. Assume that $v_{ij} > 1$ whenever $v_{ij} > 0$. Let $b$ be an $\epsilon$-optimal solution to \Cref{eq:shmyrev cp capped} such that $x_j(b) \ge t$ for all $j \in M$.
	Fix an agent $i \in N$ such that $\sum_{j \in M_i} b_{ij} = B_i$. Suppose that $\epsilon\le (tB_i^2)/(8m^2n^2B^2)$. Consider the personalized prices $p_{ij} = b_{ij}/x_j$ for all $j \in M_i$. Then for 
	\[
	\epsilon' = (B+\textup{cap}_{\textup{sum}}) e^{2F_{i}} (\sqrt{2\epsilon/t} + 2\epsilon/t), \qquad \text{where } F_i = 1+\log\frac{2mn\cdot B\cdot v_{\max}}{B_i},
	\]
	the allocation $x = x(b)$ maximizes $i$'s utility up to $\epsilon'$. More precisely, for every allocation $y$ satisfying the cap constraints and satisfying $\langle p_{i}, y\rangle \le B_{i}$, we have $u_{i}(x) \ge u_{i}(y) - \epsilon'$.
\end{lemma}

\begin{proof}
	Suppose $b$ forms an $\epsilon$-optimal solution to the convex program and is such that there is some $t > 0$ with $x_j(b) \ge t$ for all $j \in M$. Write $p_{ij} := b_{ij}/x_j(b)$ for $i \in N$, $j \in M_i$.
	
	By \Cref{thm:epsilon-kkt} ($\epsilon$-KKT), there exists $g \in \partial_{\epsilon} f(b)$ and multipliers $\lambda_i, \mu_j, \eta_{ij} \ge 0$ such that
	\[
	g_{ij} + \lambda_i  + \mu_j -  \eta_{ij}  = 0
	\]
	Define the slackness errors
	\[
	\textstyle
	\epsilon_{\lambda_i}:=\lambda_i(B_i-\sum_j b_{ij}),
	\qquad 
	\epsilon_{\mu_j}:=\mu_j(\text{cap}_j-\sum_i b_{ij}),
	\qquad 
	\epsilon_{\eta_{ij}}:=\eta_{ij}b_{ij}.
	\]
	These errors are non-negative, and $\epsilon$-complementary slackness in \Cref{thm:epsilon-kkt} gives $\sum_{i \in N} \epsilon_{\lambda_i} + \sum_{j \in M} \epsilon_{\mu_j} + \sum_{i \in N, j \in M_i} \epsilon_{\eta_{ij}} \leq \epsilon$.
	
	\paragraph{Computing the subdifferential.}
	Next, let us understand the consequences of $g \in \partial_{\epsilon} f(b)$. For this, it is convenient to first compute the $\epsilon$-subdifferential of the function $h(x) = \sum_i x_i \log\frac{x_i}{\sum_k x_k}$. Using that its convex conjugate $h^*(y)$ satisfies $h^*(y) = 0$ if $\sum_i e^{y_i} \leq 1$ and $h^*(y) = +\infty$ otherwise, as shown in \Cref{lem:convex-conjugate-h}, we find that
	\begin{align*}
		\partial_{\epsilon}h(x) 
		&= \{ w : h(x) + h^*(w) \le \langle x, w\rangle + \epsilon \} \\
		&= \left\{ w : \sum_{i} e^{w_i} \le 1 \text{ and } \langle w, x\rangle \ge \sum_i x_i \log\frac{x_i}{\sum_k x_k} - \epsilon \right\}.
	\end{align*}
	Therefore, by the summation rule for $\epsilon$-subdifferentials (see \cref{eq:eps subdiff summation rule} and surrounding text), it follows that the vector $g \in \partial_{\epsilon} f(b)$ can be written as $g_{ij} = -\log v_{ij} + w_{ij}$ for some $(w_{ij})_{i \in N, j \in M_i}$ such that $\sum_{i \in N_j} e^{w_{ij}} \le 1$ for all $j \in M$ and that satisfy the following:
	\[
	\sum_{i \in N_j} b_{ij} w_{ij} \ge \sum_{i \in N_j} b_{ij} \log\frac{b_{ij}}{x_j(b)} - \epsilon
	\quad
	\text{for all $j \in M$}.
	\]
	We can rewrite this latter condition as
	\[
	\sum_{i \in N_j} b_{ij} \left( \log p_{ij} - \log e^{w_{ij}} \right) \le \epsilon
	\quad
	\text{for all $j \in M$}
	\]
	which implies
	\[
	x_j(b) \text{KL}\left( (p_{ij})_{i \in N_j} \: \middle\| \: (e^{w_{ij}})_{i \in N_j} \right) \le \epsilon
	\quad
	\text{for all $j \in M$}.
	\]
	Applying Pinsker's inequality (\Cref{lem:sub-pinsker}, which is applicable since $\sum_{i \in N_j} e^{w_{ij}} \le 1$) and using that $x_j(b) \ge t$, we find that
	\[
	\left\| (p_{ij})_{i \in N_j} - (e^{w_{ij}})_{i \in N_j} \right\|_1^2 \le 2\epsilon/t
	\quad
	\text{for all $j \in M$},
	\]
	and since $\|\cdot\|_{\infty} \le \|\cdot\|_1$, we get
	\begin{equation}
		\label{eq:epsilon-difference-to-price}
		\left|p_{ij} - e^{w_{ij}} \right| \le \sqrt{2 \epsilon / t} \quad \text{for all $j \in M$ and $i \in N_j$}.
	\end{equation}
	
	\paragraph{Stationarity condition.}
	Now, the stationarity conditions become
	\[
	\log v_{ij} = w_{ij} + \lambda_i + \mu_j - \eta_{ij}.
	\]
	Exponentiating,
	\begin{equation}
		\label{eq:epsilon-stationarity}
		v_{ij} = e^{w_{ij}} e^{\lambda_i} e^{\mu_j} e^{-\eta_{ij}}.
	\end{equation}
	
	\paragraph{Bounding the multipliers.}
	We will now derive upper bounds on the multipliers $\lambda_i$ and $\mu_j$.
	Fix an agent $i \in N$ and assume, as in the lemma statement, that $i$ spends their entire budget ($\sum_{j \in M_i} b_{ij} = B_i$) and that we have $\epsilon\le (tB_i^2)/(8m^2n^2B^2)$.
	
	We need to give an upper bound on $\lambda_i
	=\log\frac{v_{ij}}{e^{w_{ij}}}-\mu_j+\eta_{ij}$, so let us derive a lower bound on $e^{w_{ij}}$ and an upper bound on $\eta_{ij}$.
	Since $\sum_{j\in M_i}b_{ij}=B_i$, there is some $j\in M_i$
	with $b_{ij}\ge B_i/m$. Since $x_j\le B$, we have
	$p_{ij}\ge B_i/(mB)$. 
	By \eqref{eq:epsilon-difference-to-price}
	and the upper bound on $\epsilon$ (which implies $2\epsilon/t \le B_i^2/(4m^2B^2)$),
	\[
	e^{w_{ij}}
	\ge p_{ij}-\sqrt{2\epsilon/t}
	\ge \frac{B_i}{mB} - \frac{B_i}{2mB}
	= \frac{B_i}{2mB}.
	\]
	Moreover, approximate complementary slackness gives (because $\epsilon \le (t/B) (B_i/B)(B_i/m) \le B_i/m$)
	\[
	\eta_{ij}\le\frac{\epsilon}{b_{ij}}
	\le\frac{m\epsilon}{B_i}\le1.
	\]
	Hence, using stationarity and $\mu_j\ge0$,
	\[
	\lambda_i
	=\log\frac{v_{ij}}{e^{w_{ij}}}-\mu_j+\eta_{ij}
	\le \log \frac{v_{\max}}{e^{w_{ij}}} - 0 +\eta_{ij}
	\le 1+\log\frac{2mB \cdot v_{\max}}{B_i}.
	\]
	
	Next, we give an upper bound on $\mu_j$. Fix a project $j$. There must be some $k\in N_j$ with
	$b_{kj}\ge x_j/n\ge t/n$, and therefore $p_{kj}\ge1/n$.
	Again by \eqref{eq:epsilon-difference-to-price}, the upper bound on $\epsilon$ (which implies $2\epsilon/t \le 1/(4n^2)$), and
	approximate complementary slackness,
	\[
	e^{w_{kj}}\ge p_{kj}-\sqrt{2\epsilon/t} \ge \frac{1}{n} - \frac{1}{2n} \ge\frac1{2n}
	\qquad\text{and}\qquad
	\eta_{kj}\le\frac{\epsilon}{b_{kj}}
	\le\frac{n\epsilon}{t}\le1.
	\]
	Using stationarity and $\lambda_k\ge0$, we obtain
	\[
	\mu_j
	=\log\frac{v_{kj}}{e^{w_{kj}}}-\lambda_k+\eta_{kj}
	\le1+\log(2nv_{\max}).
	\]
	Thus, for the value from the lemma statement
	\[
	F_i=1+\log\frac{2mn\cdot B\cdot v_{\max}}{B_i},
	\]
	we then have $\lambda_i\le F_i$ and $\mu_j\le F_i$ for every $j\in M$.

	\paragraph{Utility maximization}
	Continue to fix the agent $i \in N$, and like in the lemma statement assume that $\sum_{j \in M_i} b_{ij} = B_i$. We wish to show that the allocation $x = x(b)$ approximately maximizes the utility of $i$ with respect to the prices $p_{ij} = b_{ij}/x_j(b)$. Note that with these prices, $x$ is exactly affordable by our assumption, $\langle p_i, x \rangle = \sum_{j \in M_i} b_{ij}/x_j(b) \cdot x_j(b) = B_i$.
	Let $y$ be any other allocation satisfying the cap constraints, i.e., $0 \le y_j \le \text{cap}_j$, and which is affordable for agent $i$, i.e., $\langle p_i, y \rangle \le B_i$. 
	By the cap constraints, $\sum_{j \in M} y_j \le \text{cap}_{\text{sum}}$.
	
	To establish that $x$ is approximately utility-maximizing, we need to show that $u_i(y) - u_i(x)$ is small. Rewriting,
	\begin{align*}
		u_i(y) - u_i(x) &= \sum_{j \in M_i} v_{ij} (y_j - x_j) \\
		&\overset{\eqref{eq:epsilon-stationarity}}{=} \sum_{j \in M_i} e^{w_{ij}} e^{\lambda_i} e^{\mu_j} e^{-\eta_{ij}} (y_j - x_j) \\
		&= e^{\lambda_i} \sum_{j \in M_i} e^{w_{ij}} e^{\mu_j} e^{-\eta_{ij}} (y_j - x_j)
	\end{align*}
	For each $j \in M_i$, let us write $e^{w_{ij}} = p_{ij} + \delta_{ij}$ for some error value $\delta_{ij}$, where from \eqref{eq:epsilon-difference-to-price} we know $|\delta_{ij}| \le \sqrt{2\epsilon/t}$. Then
	\begin{align*}
		u_i(y) - u_i(x) &= e^{\lambda_i} \sum_{j \in M_i} (p_{ij} + \delta_{ij}) e^{\mu_j} e^{-\eta_{ij}} (y_j - x_j) \\
		&= \underbrace{e^{\lambda_i} \sum_{j \in M_i} p_{ij} e^{\mu_j} e^{-\eta_{ij}} (y_j - x_j)}_{\text{Main Term}} + \underbrace{e^{\lambda_i} \sum_{j \in M_i} \delta_{ij} e^{\mu_j} e^{-\eta_{ij}} (y_j - x_j)}_{E_1\text{: Error from subgradient approx.}}
	\end{align*}
	
	Let's bound the error term $E_1$:
	\begin{align*}
		|E_1| &\le e^{\lambda_i} \sum_{j \in M_i} |\delta_{ij}| e^{\mu_j} e^{-\eta_{ij}} |y_j - x_j| \\
		&\le e^{F_i} \left( \sqrt{2\epsilon/t} \right) e^{F_i} \cdot 1 \cdot \sum_{j \in M_i} |y_j - x_j| \tag{using $\eta_{ij} \ge 0$ so $e^{-\eta_{ij}} \le 1$} \\
		&\le e^{2F_i} \sqrt{2\epsilon/t} \left(\textstyle\sum_j y_j + \sum_j x_j\right) \\
		&\le (B+\text{cap}_{\text{sum}}) e^{2F_i} \sqrt{2\epsilon/t}.
	\end{align*}
	
	Now, let's analyze the Main Term. We can rewrite it as:
	\[
	\text{Main Term} = e^{\lambda_i} \sum_{j \in M_i} p_{ij}(y_j-x_j) + \underbrace{e^{\lambda_i} \sum_{j \in M_i} p_{ij} (e^{\mu_j}e^{-\eta_{ij}} - 1) (y_j - x_j)}_{E_2\text{: Error from slackness}}.
	\]
	The first part is non-positive because $y$ is affordable for $i$ and $x$ exhausts the budget of $i$:
	\[ e^{\lambda_i} \sum_{j \in M_i} p_{ij}(y_j-x_j) = e^{\lambda_i}(\langle p_i, y \rangle - \langle p_i, x \rangle) \le e^{\lambda_i}(B_i - B_i) = 0. \]
	So, we only need to bound the slackness error, $E_2$. Let's split the sum based on the sign of $(y_j - x_j)$.
	
	\textbf{Case 1: $y_j > x_j$.} The term $(y_j-x_j)$ is positive. We need an upper bound for $(e^{\mu_j}e^{-\eta_{ij}} - 1)$. Since $\eta_{ij} \ge 0$, we have $e^{-\eta_{ij}} \le 1$, so $e^{\mu_j}e^{-\eta_{ij}} - 1 \le e^{\mu_j} - 1$. Now,
	\begin{align*}
		p_{ij} (e^{\mu_j} - 1) (y_j - x_j) &\le p_{ij} (e^{\mu_j} - 1) (\text{cap}_j - x_j) \tag{since $y_j \le \text{cap}_j$} \\
		&\le p_{ij} e^{\mu_j} \mu_j (\text{cap}_j - x_j) \tag{using $e^z-1 \le z e^z$ for $z \ge 0$}
	\end{align*}
	From $\epsilon$-complementary slackness, we have $\mu_j(\text{cap}_j - x_j) \le \epsilon_{\mu_j}$ with $\sum_j \epsilon_{\mu_j} \le \epsilon$.
	The contribution from this case to $E_2$ is at most:
	\[ e^{\lambda_i} \sum_{j: y_j>x_j} p_{ij} e^{\mu_j} \epsilon_{\mu_j} \le e^{F_i} \sum_{j: y_j>x_j} (B/t) e^{F_i} \epsilon_{\mu_j} \le (B/t)e^{2F_i} \sum_j \epsilon_{\mu_j} \le (B/t)e^{2F_i} \epsilon. \]
	
	\textbf{Case 2: $y_j < x_j$.} The term $(y_j-x_j)$ is negative. We need a lower bound for $(e^{\mu_j}e^{-\eta_{ij}} - 1)$ to get an upper bound on the product. Using $e^{-z} \ge 1-z$ for $z \ge 0$:
	\[ e^{\mu_j}e^{-\eta_{ij}} - 1 \ge e^{\mu_j}(1-\eta_{ij}) - 1 = (e^{\mu_j}-1) - \eta_{ij}e^{\mu_j} \ge -\eta_{ij}e^{\mu_j}. \]
	The term $p_{ij} (e^{\mu_j}e^{-\eta_{ij}} - 1) (y_j - x_j)$ is therefore bounded above by $p_{ij} (-\eta_{ij}e^{\mu_j}) (y_j-x_j)$.
	From $\epsilon$-complementary slackness, $\eta_{ij}b_{ij} \le \epsilon_{\eta_{ij}}$ with $\sum_{ij} \epsilon_{\eta_{ij}} \le \epsilon$. This means $\eta_{ij} p_{ij} x_j \le \epsilon_{\eta_{ij}}$, so $\eta_{ij}p_{ij} \le \epsilon_{\eta_{ij}}/x_j \le \epsilon_{\eta_{ij}}/t$. Hence
	\begin{align*}
		p_{ij} (-\eta_{ij}e^{\mu_j}) (y_j-x_j) &= (- (y_j-x_j)) \cdot (p_{ij}\eta_{ij}) \cdot e^{\mu_j} \\
		&\le x_j \cdot (\epsilon_{\eta_{ij}}/t) \cdot e^{F_i} \le (B/t) e^{F_i} \epsilon_{\eta_{ij}}.
	\end{align*}
	The contribution from this case to $E_2$ is at most $e^{\lambda_i} \sum_{j: y_j<x_j} (B/t) e^{F_i} \epsilon_{\eta_{ij}} \le (B/t)e^{2F_i}\epsilon$.
	
	Combining both cases, we get $E_2 \le (2B/t)e^{2F_i}\epsilon$.
	
	Finally, putting everything together:
	\begin{align*}
		u_i(y) - u_i(x) \le 0 + |E_1| + E_2
		&\le (B+\text{cap}_{\text{sum}})e^{2F_i} \sqrt{2\epsilon/t} + (2B/t)e^{2F_i}\epsilon \\
		&\le (B+\text{cap}_{\text{sum}}) e^{2F_i}(\sqrt{2\epsilon/t} + 2\epsilon/t).
	\end{align*}
	This shows that $u_i(x) \ge u_i(y) - \epsilon'$ where $\epsilon' = (B+\text{cap}_{\text{sum}}) e^{2F_i}(\sqrt{2\epsilon/t} + 2\epsilon/t)$, as required.
\end{proof}

We now use \Cref{thm:approximate-lindahl-lower-bound} to show that an approximate solution to our convex program produces an approximate Lindahl equilibrium, in the following sense discussed in \Cref{sec:capped-approximate}.

\defepsilonlindahl*

Since \Cref{thm:approximate-lindahl-lower-bound} makes an extra assumption that $x_j \ge t$ for all $j \in M$, our argument perturbs the input instance by adding additional voters with $\epsilon$ weight who ensure that in any feasible solution to the program, we have $x_j \ge \epsilon/m$ for all $j \in M$. We then invoke \Cref{thm:approximate-lindahl-lower-bound} with $t = \epsilon/m$, and show that the resulting solution for the perturbed instance forms an approximate Lindahl equilibrium with respect to the original instance.

\approximatelindahl*
\begin{proof}
	Suppose we are given an instance $I$ of the capped public goods problem and a value of $\epsilon > 0$. Let $v_{\max} = \max_{ij} v_{ij} > 1$. Our algorithm will construct a perturbed instance $I'$ on agent set $N' = N \cup \{i^{+}_1, \dots, i^{+}_m\}$, obtained by adding an extra agent $i^{+}_j$ for each project $j \in M$, with budget $B_{i^{+}_j} = \epsilon/m$ and with a ``single-minded'' valuation, meaning that 
	\[
	  v_{i^{+}_j j} = v_{\max} \quad \text{and} \quad  v_{i^{+}_j k} = 0 \text{ for all $k \neq j$}.
	  \]
	  Thus, the overall budget available in $I'$ is $B' = B + m \cdot \frac{\epsilon}{m} = B + \epsilon$. Note that $B' \le B + 1$.
	
	Now let $B_{\min} = \min_{i \in N} B_i$ (where the minimum is taken only over the original agents), $\text{cap}_{\text{sum}} = \sum_{j \in M} \text{cap}_j$, $n' = n + m$, and set
	\[
	\epsilon_{\text{solver}} = 
	\frac%
	{\epsilon^3}%
	{\left(
		2 \gamma
		\right)^2 \cdot 2m}
	\qquad
	\text{for }
	\gamma := (B + 1 + \text{cap}_{\text{sum}}) e^{2F}
	\text{ and }
	F := 1+\log\frac{2mn'(B+1)v_{\max}}{B_{\min}}.
	\]
	The algorithm will compute an $\epsilon_{\text{solver}}$-optimal solution $b$ to  \Cref{eq:shmyrev cp capped} as applied to the perturbed instance $I'$. We can do this using the ellipsoid method \citep[see, e.g.,][Theorem 13.1]{vishnoi2021algorithms} whose runtime has a dependence on $\log \frac1{\epsilon_{\text{solver}}}$ which is polynomial (as $\epsilon_{\text{solver}}$ contains single-exponential terms).
	
	The algorithm then checks if there is a voter $i \in N'$ who doesn't spend all of their budget (i.e., $\sum_{j \in M_i} b_{ij} < B_i$) but approves a project that doesn't exhaust its cap (i.e., $x_j(b) < \text{cap}_j$ for some $j \in M_i$). In these cases, the algorithm increases $b_{ij}$ appropriately. This increase preserves feasibility and improves the objective function value, because:
	\begin{itemize}
		\item if $x_j(b) > 0$, then the derivative of the objective function with respect to $b_{ij}$ (or the one-sided derivative in case $b_{ij} = 0$) is $- \log v_{ij} + \log \frac{b_{ij}}{x_j(b)}$ (compare \Cref{lem:shmyrev subdiff}), which is negative because $v_{ij} > 1$ by our normalization and $0 \le b_{ij}/x_j(b) \le 1$ by definition of $x_j(b)$;
		\item if $x_j(b) = 0$, then $b_{ij} = 0$ by definition of $x_j(b)$ and increasing $b_{ij}$ by a small $\tau > 0$ gives objective change $-\tau\log v_{ij} < 0$ since $v_{ij} > 1$.
	\end{itemize}
	Thus, increasing $b_{ij}$ decreases the objective function which is to be minimized. Thus, this increase preserves $\epsilon_{\text{solver}}$-optimality. This process is repeated until no such voter $i$ exists.
	
	Now, for every $j \in M$, we either have $x_j(b) = \text{cap}_j \ge \epsilon$, or else (by the postprocessing we just did) the entire budget of the extra voter $i^{+}_j$ must be going towards project $j$ and hence $x_j(b) \ge b_{i^{+}_jj} = B_{i^{+}_j} = \epsilon/m$. In either case $x_j(b) \ge \epsilon/m$. This will later allow us to invoke \Cref{thm:approximate-lindahl-lower-bound} with $t = \epsilon/m$. 
	
	Consider the personalized prices defined as $p_{ij} = b_{ij}/x_j$ for all $i \in N$ and $j \in M_i$. We will show that  $(x, (p_{ij})_{i \in N, j \in M})$ is an $\epsilon$-approximate Lindahl equilibrium (for the original instance $I$).
	\begin{itemize}
		\item For approximate budget-feasibility, note that $\sum_{j \in M} x_j \le B' = B + \epsilon$.
		\item For affordability, note that $\langle p_i, x\rangle = \sum_{j \in M} \frac{b_{ij}}{x_j} x_j = \sum_{j \in M_i} \frac{b_{ij}}{x_j} x_j = \sum_{j \in M_i} b_{ij} \le B_i$ for all $i \in N$, since $b$ is a feasible solution to \Cref{eq:shmyrev cp capped}.
		\item For approximate profit maximization, note that for every $j \in M$ we have $\sum_{i \in N} p_{ij} = \sum_{i \in N} b_{ij}/x_j \le \sum_{i \in N'} b_{ij}/x_j = 1$ by definition of $x_j = x_j(b)$. Also note that by construction, we have $x_j > 0$ for every $j$, and we have $\sum_{i \in N} p_{ij} = \sum_{i \in N} b_{ij}/x_j = \sum_{i \in N'} b_{ij}/x_j - \sum_{k \in M} b_{i^{+}_kj}/x_j = 1 - b_{i^{+}_jj}/x_j \ge 1 - B_{i^{+}_j}/x_j = 1 - \frac{\epsilon}{mx_j} \ge 1 - \frac{\epsilon}{x_j} $.
	\end{itemize}
	It remains to verify approximate utility maximization. Consider some agent $i \in N$. By our post-processing of the solution, if we have $\sum_{j \in M_i} b_{ij} < B_i$, then all projects that agent $i$ likes have been funded up to their cap, so $i$ receives the best-possible outcome and thus utility maximization is satisfied for $i$. Otherwise, $\sum_{j \in M_i} b_{ij} = B_i$. In this case, we will invoke \Cref{thm:approximate-lindahl-lower-bound} applied to $I'$ with $\epsilon_{\text{solver}}$ and $t = \epsilon/m$.
	To apply that lemma, we need to check that $\epsilon_{\text{solver}}$ satisfies the upper bound imposed in the lemma statement, i.e., that
	\[
	\epsilon_{\text{solver}} = \frac%
	{\epsilon^3}%
	{
		8 \gamma^2m}
		 \le \frac{(\epsilon/m) B_i^2}{8m^2n'^2B'^2}
	\iff
	\frac{B_i^2\gamma^2}{\epsilon^2 m^2n'^2B'^2} \ge 1
	\iff
	\gamma \ge \frac{\epsilon mn'B'}{B_i}.
	\]
	The last equivalent form can be checked by plugging in the definition of $\gamma$ and using $\epsilon \le 1$.
	Hence we can apply \Cref{thm:approximate-lindahl-lower-bound}, which shows that $x$ is utility-maximizing for $i$ up to an error of $\epsilon_{\text{util}}$ as specified in the statement of \Cref{thm:approximate-lindahl-lower-bound}, which is in terms of a parameter $F_i$ (involving a $1/B_i$ term) that we can upper bound by $F$ (which instead involves a $1/B_{\min}$ term). In the following computation, we will use that $2\gamma = 2(B + 1 + \text{cap}_{\text{sum}}) e^{2F} \ge 1$ which follows since $2F > 0$ (which itself follows from $v_{\max} > 1$ due to the normalization of valuations). We get
	\begin{align*}
		\epsilon_{\text{util}} 
		&\le (B + 1 + \text{cap}_{\text{sum}}) e^{2F_{i}} \left(\sqrt{2\epsilon_{\text{solver}}/(\epsilon/m)} + 2\epsilon_{\text{solver}}/(\epsilon/m)\right) \tag{\Cref{thm:approximate-lindahl-lower-bound} and $B' \le B+1$} \\
		&\le (B + 1 + \text{cap}_{\text{sum}}) e^{2F} \left(\sqrt{2\epsilon_{\text{solver}}/(\epsilon/m)} + 2\epsilon_{\text{solver}}/(\epsilon/m)\right) \tag{as $F_i \le F$} \\
		&\le \gamma \left(\sqrt{2\epsilon_{\text{solver}}/(\epsilon/m)} + 2\epsilon_{\text{solver}}/(\epsilon/m)\right) \tag{def. of $\gamma$} \\
		&= \gamma \left(\sqrt{2\epsilon^3/((2\gamma)^2 2m \cdot\epsilon/m)} + 2\epsilon^3/((2\gamma)^2 2m \cdot\epsilon/m)\right) \tag{def. of $\epsilon_{\text{solver}}$} \\
		&= \gamma \left(\sqrt{\epsilon^2/(2\gamma)^2} + \epsilon^2/(2\gamma)^2\right) \tag{cancelling} \\
		&= \gamma \left(\epsilon/(2\gamma) + \epsilon^2/(2\gamma)^2\right) \tag{take square root} \\
		&\le \gamma \cdot (\epsilon/(2\gamma) + \epsilon/(2\gamma)) \tag{$\epsilon \le 1$ and $2\gamma \ge 1$} \\
		&= \epsilon,
	\end{align*}
	hence showing that $x$ is $\epsilon$-utility maximizing for $i$.
	
	Thus, the algorithm has successfully computed an $\epsilon$-approximate Lindahl equilibrium.
\end{proof}


\begin{thebibliography}{93}


\ifx \showCODEN    \undefined \def \showCODEN     #1{\unskip}     \fi
\ifx \showISBNx    \undefined \def \showISBNx     #1{\unskip}     \fi
\ifx \showISBNxiii \undefined \def \showISBNxiii  #1{\unskip}     \fi
\ifx \showISSN     \undefined \def \showISSN      #1{\unskip}     \fi
\ifx \showLCCN     \undefined \def \showLCCN      #1{\unskip}     \fi
\ifx \shownote     \undefined \def \shownote      #1{#1}          \fi
\ifx \showarticletitle \undefined \def \showarticletitle #1{#1}   \fi
\ifx \showURL      \undefined \def \showURL       {\relax}        \fi
\providecommand\bibfield[2]{#2}
\providecommand\bibinfo[2]{#2}
\providecommand\natexlab[1]{#1}
\providecommand\showeprint[2][]{arXiv:#2}

\bibitem[Airiau et~al\mbox{.}(2023)]%
        {airiau2023portioning}
\bibfield{author}{\bibinfo{person}{St{\'e}phane Airiau}, \bibinfo{person}{Haris
  Aziz}, \bibinfo{person}{Ioannis Caragiannis}, \bibinfo{person}{Justin
  Kruger}, \bibinfo{person}{J{\'e}r{\^o}me Lang}, {and}
  \bibinfo{person}{Dominik Peters}.} \bibinfo{year}{2023}\natexlab{}.
\newblock \showarticletitle{Portioning using ordinal preferences: Fairness and
  efficiency}.
\newblock \bibinfo{journal}{\emph{Artificial Intelligence}}
  \bibinfo{volume}{314} (\bibinfo{year}{2023}), \bibinfo{pages}{103809}.
\newblock
\href{https://doi.org/10.1016/j.artint.2022.103809}{doi:\nolinkurl{10.1016/j.artint.2022.103809}}


\bibitem[Aziz et~al\mbox{.}(2020)]%
        {aziz2020fairmixing}
\bibfield{author}{\bibinfo{person}{Haris Aziz}, \bibinfo{person}{Anna
  Bogomolnaia}, {and} \bibinfo{person}{Herv\'{e} Moulin}.}
  \bibinfo{year}{2020}\natexlab{}.
\newblock \showarticletitle{Fair mixing: The case of dichotomous preferences}.
\newblock \bibinfo{journal}{\emph{ACM Transactions on Economics and Computation
  (TEAC)}} \bibinfo{volume}{8}, \bibinfo{number}{4}, Article
  \bibinfo{articleno}{18} (\bibinfo{year}{2020}), \bibinfo{numpages}{27}~pages.
\newblock
\href{https://doi.org/10.1145/3417738}{doi:\nolinkurl{10.1145/3417738}}


\bibitem[Aziz et~al\mbox{.}(2017)]%
        {aziz2017ejr}
\bibfield{author}{\bibinfo{person}{Haris Aziz}, \bibinfo{person}{Markus Brill},
  \bibinfo{person}{Vincent Conitzer}, \bibinfo{person}{Edith Elkind},
  \bibinfo{person}{Rupert Freeman}, {and} \bibinfo{person}{Toby Walsh}.}
  \bibinfo{year}{2017}\natexlab{}.
\newblock \showarticletitle{Justified Representation in approval-based
  committee voting}.
\newblock \bibinfo{journal}{\emph{Social Choice and Welfare}}
  \bibinfo{volume}{48}, \bibinfo{number}{2} (\bibinfo{year}{2017}),
  \bibinfo{pages}{461--485}.
\newblock
\href{https://doi.org/10.1007/s00355-016-1019-3}{doi:\nolinkurl{10.1007/s00355-016-1019-3}}


\bibitem[Aziz et~al\mbox{.}(2023a)]%
        {aziz2023bestofbothworlds}
\bibfield{author}{\bibinfo{person}{Haris Aziz}, \bibinfo{person}{Xinhang Lu},
  \bibinfo{person}{Mashbat Suzuki}, \bibinfo{person}{Jeremy Vollen}, {and}
  \bibinfo{person}{Toby Walsh}.} \bibinfo{year}{2023}\natexlab{a}.
\newblock \showarticletitle{Best-of-both-worlds fairness in committee voting}.
  In \bibinfo{booktitle}{\emph{Proceedings of the 19th International Conference
  on Web and Internet Economics (WINE)}}. \bibinfo{pages}{676}.
\newblock
\newblock
\shownote{Full version
  \href{https://arxiv.org/abs/2303.03642}{arXiv:2303.03642}}.


\bibitem[Aziz et~al\mbox{.}(2023b)]%
        {aziz2023peerreview}
\bibfield{author}{\bibinfo{person}{Haris Aziz}, \bibinfo{person}{Evi Micha},
  {and} \bibinfo{person}{Nisarg Shah}.} \bibinfo{year}{2023}\natexlab{b}.
\newblock \showarticletitle{Group fairness in peer review}. In
  \bibinfo{booktitle}{\emph{Advances in Neural Information Processing
  Systems}}, Vol.~\bibinfo{volume}{36}. \bibinfo{pages}{64885--64895}.
\newblock
\urldef\tempurl%
\url{https://proceedings.neurips.cc/paper_files/paper/2023/file/ccba10dd4e80e7276054222bb95d467c-Paper-Conference.pdf}
\showURL{%
\tempurl}


\bibitem[Aziz and Shah(2021)]%
        {aziz2021pbsurvey}
\bibfield{author}{\bibinfo{person}{Haris Aziz} {and} \bibinfo{person}{Nisarg
  Shah}.} \bibinfo{year}{2021}\natexlab{}.
\newblock \showarticletitle{Participatory budgeting: Models and approaches}.
\newblock In \bibinfo{booktitle}{\emph{Pathways Between Social Science and
  Computational Social Science: Theories, Methods, and Interpretations}},
  \bibfield{editor}{\bibinfo{person}{Tam{\'a}s Rudas} {and}
  \bibinfo{person}{G{\'a}bor P{\'e}li}} (Eds.). \bibinfo{publisher}{Springer},
  \bibinfo{pages}{215--236}.
\newblock
\href{https://doi.org/10.1007/978-3-030-54936-7_10}{doi:\nolinkurl{10.1007/978-3-030-54936-7_10}}


\bibitem[Banerjee et~al\mbox{.}(2023)]%
        {banerjee2023proportionally}
\bibfield{author}{\bibinfo{person}{Siddhartha Banerjee},
  \bibinfo{person}{Vasilis Gkatzelis}, \bibinfo{person}{Safwan Hossain},
  \bibinfo{person}{Billy Jin}, \bibinfo{person}{Evi Micha}, {and}
  \bibinfo{person}{Nisarg Shah}.} \bibinfo{year}{2023}\natexlab{}.
\newblock \showarticletitle{Proportionally fair online allocation of public
  goods with predictions}. In \bibinfo{booktitle}{\emph{Proceedings of the 32nd
  International Joint Conference on Artificial Intelligence (IJCAI)}}.
  \bibinfo{pages}{20--28}.
\newblock
\href{https://doi.org/10.24963/ijcai.2023/3}{doi:\nolinkurl{10.24963/ijcai.2023/3}}


\bibitem[Beck(2017)]%
        {beck2017firstorder}
\bibfield{author}{\bibinfo{person}{Amir Beck}.}
  \bibinfo{year}{2017}\natexlab{}.
\newblock \bibinfo{booktitle}{\emph{First-Order Methods in Optimization}}.
\newblock \bibinfo{publisher}{Society for Industrial and Applied Mathematics
  (SIAM)}.
\newblock
\href{https://doi.org/10.1137/1.9781611974997}{doi:\nolinkurl{10.1137/1.9781611974997}}


\bibitem[Bei et~al\mbox{.}(2025)]%
        {bei2024truthful}
\bibfield{author}{\bibinfo{person}{Xiaohui Bei}, \bibinfo{person}{Xinhang Lu},
  {and} \bibinfo{person}{Warut Suksompong}.} \bibinfo{year}{2025}\natexlab{}.
\newblock \showarticletitle{Truthful cake sharing}.
\newblock \bibinfo{journal}{\emph{Social Choice and Welfare}}
  \bibinfo{volume}{64}, \bibinfo{number}{1--2} (\bibinfo{year}{2025}),
  \bibinfo{pages}{309--343}.
\newblock
\href{https://doi.org/10.1007/s00355-023-01503-0}{doi:\nolinkurl{10.1007/s00355-023-01503-0}}


\bibitem[Birnbaum et~al\mbox{.}(2011)]%
        {birnbaum2011distributed}
\bibfield{author}{\bibinfo{person}{Benjamin Birnbaum},
  \bibinfo{person}{Nikhil~R. Devanur}, {and} \bibinfo{person}{Lin Xiao}.}
  \bibinfo{year}{2011}\natexlab{}.
\newblock \showarticletitle{Distributed algorithms via gradient descent for
  {Fisher} markets}. In \bibinfo{booktitle}{\emph{Proceedings of the 12th ACM
  Conference on Electronic Commerce (EC)}}. \bibinfo{pages}{127--136}.
\newblock
\href{https://doi.org/10.1145/1993574.1993594}{doi:\nolinkurl{10.1145/1993574.1993594}}


\bibitem[Bogomolnaia et~al\mbox{.}(2005)]%
        {bms2005}
\bibfield{author}{\bibinfo{person}{Anna Bogomolnaia}, \bibinfo{person}{Herv\'e
  Moulin}, {and} \bibinfo{person}{Richard Stong}.}
  \bibinfo{year}{2005}\natexlab{}.
\newblock \showarticletitle{Collective choice under dichotomous preferences}.
\newblock \bibinfo{journal}{\emph{Journal of Economic Theory}}
  \bibinfo{volume}{122}, \bibinfo{number}{2} (\bibinfo{year}{2005}),
  \bibinfo{pages}{165--184}.
\newblock
\href{https://doi.org/10.1016/j.jet.2004.05.005}{doi:\nolinkurl{10.1016/j.jet.2004.05.005}}


\bibitem[Boyd and Vandenberghe(2004)]%
        {boyd2004convex}
\bibfield{author}{\bibinfo{person}{Stephen~P. Boyd} {and}
  \bibinfo{person}{Lieven Vandenberghe}.} \bibinfo{year}{2004}\natexlab{}.
\newblock \bibinfo{booktitle}{\emph{Convex Optimization}}.
\newblock \bibinfo{publisher}{Cambridge University Press}.
\newblock
\href{https://doi.org/10.1017/CBO9780511804441}{doi:\nolinkurl{10.1017/CBO9780511804441}}


\bibitem[Brandl et~al\mbox{.}(2022)]%
        {brandl2022fundingpublicprojects}
\bibfield{author}{\bibinfo{person}{Florian Brandl}, \bibinfo{person}{Felix
  Brandt}, \bibinfo{person}{Matthias Greger}, \bibinfo{person}{Dominik Peters},
  \bibinfo{person}{Christian Stricker}, {and} \bibinfo{person}{Warut
  Suksompong}.} \bibinfo{year}{2022}\natexlab{}.
\newblock \showarticletitle{Funding public projects: A case for the Nash
  product rule}.
\newblock \bibinfo{journal}{\emph{Journal of Mathematical Economics}}
  \bibinfo{volume}{99} (\bibinfo{year}{2022}), \bibinfo{pages}{102585}.
\newblock
\href{https://doi.org/10.1016/j.jmateco.2021.102585}{doi:\nolinkurl{10.1016/j.jmateco.2021.102585}}


\bibitem[Brandl et~al\mbox{.}(2021)]%
        {brandl2021distribution}
\bibfield{author}{\bibinfo{person}{Florian Brandl}, \bibinfo{person}{Felix
  Brandt}, \bibinfo{person}{Dominik Peters}, {and} \bibinfo{person}{Christian
  Stricker}.} \bibinfo{year}{2021}\natexlab{}.
\newblock \showarticletitle{Distribution rules under dichotomous preferences:
  two out of three ain't bad}. In \bibinfo{booktitle}{\emph{Proceedings of the
  22nd ACM Conference on Economics and Computation (EC)}}.
  \bibinfo{pages}{158--179}.
\newblock
\href{https://doi.org/10.1145/3465456.3467653}{doi:\nolinkurl{10.1145/3465456.3467653}}


\bibitem[Brandt et~al\mbox{.}(2025)]%
        {brandt2024coordinatingcharitabledonations}
\bibfield{author}{\bibinfo{person}{Felix Brandt}, \bibinfo{person}{Matthias
  Greger}, \bibinfo{person}{Erel Segal-Halevi}, {and} \bibinfo{person}{Warut
  Suksompong}.} \bibinfo{year}{2025}\natexlab{}.
\newblock \showarticletitle{Coordinating charitable donations with {Leontief}
  preferences}.
\newblock \bibinfo{journal}{\emph{Journal of Economic Theory}}
  \bibinfo{volume}{230} (\bibinfo{year}{2025}), \bibinfo{pages}{106096}.
\newblock
\href{https://doi.org/10.1016/j.jet.2025.106096}{doi:\nolinkurl{10.1016/j.jet.2025.106096}}


\bibitem[Caragiannis et~al\mbox{.}(2019)]%
        {caragiannis2019unreasonable}
\bibfield{author}{\bibinfo{person}{Ioannis Caragiannis}, \bibinfo{person}{David
  Kurokawa}, \bibinfo{person}{Herv{\'e} Moulin}, \bibinfo{person}{Ariel~D.
  Procaccia}, \bibinfo{person}{Nisarg Shah}, {and} \bibinfo{person}{Junxing
  Wang}.} \bibinfo{year}{2019}\natexlab{}.
\newblock \showarticletitle{The unreasonable fairness of maximum {Nash}
  welfare}.
\newblock \bibinfo{journal}{\emph{ACM Transactions on Economics and Computation
  (TEAC)}} \bibinfo{volume}{7}, \bibinfo{number}{3} (\bibinfo{year}{2019}),
  \bibinfo{pages}{1--32}.
\newblock
\href{https://doi.org/10.1145/3355902}{doi:\nolinkurl{10.1145/3355902}}


\bibitem[Caragiannis et~al\mbox{.}(2024)]%
        {caragiannis2024proportional}
\bibfield{author}{\bibinfo{person}{Ioannis Caragiannis}, \bibinfo{person}{Evi
  Micha}, {and} \bibinfo{person}{Nisarg Shah}.}
  \bibinfo{year}{2024}\natexlab{}.
\newblock \showarticletitle{Proportional fairness in non-centroid clustering}.
  In \bibinfo{booktitle}{\emph{Advances in Neural Information Processing
  Systems}}, Vol.~\bibinfo{volume}{37}. \bibinfo{pages}{19139--19166}.
\newblock
\urldef\tempurl%
\url{https://proceedings.neurips.cc/paper_files/paper/2024/file/220cbc7435d6a56205c87d73d15d9eda-Paper-Conference.pdf}
\showURL{%
\tempurl}


\bibitem[Chatzigeorgiou(2013)]%
        {chatzigeorgiou2013}
\bibfield{author}{\bibinfo{person}{Ioannis Chatzigeorgiou}.}
  \bibinfo{year}{2013}\natexlab{}.
\newblock \showarticletitle{Bounds on the Lambert function and their
  application to the outage analysis of user cooperation}.
\newblock \bibinfo{journal}{\emph{IEEE Communications Letters}}
  \bibinfo{volume}{17}, \bibinfo{number}{8} (\bibinfo{year}{2013}),
  \bibinfo{pages}{1505--1508}.
\newblock
\href{https://doi.org/10.1109/LCOMM.2013.070113.130972}{doi:\nolinkurl{10.1109/LCOMM.2013.070113.130972}}
\showeprint[arxiv]{1601.04895}


\bibitem[Chaudhury et~al\mbox{.}(2022)]%
        {chaudhury2022federatedlearning}
\bibfield{author}{\bibinfo{person}{Bhaskar~Ray Chaudhury},
  \bibinfo{person}{Linyi Li}, \bibinfo{person}{Mintong Kang},
  \bibinfo{person}{Bo Li}, {and} \bibinfo{person}{Ruta Mehta}.}
  \bibinfo{year}{2022}\natexlab{}.
\newblock \showarticletitle{Fairness in federated learning via core-stability}.
  In \bibinfo{booktitle}{\emph{Advances in Neural Information Processing
  Systems}}, Vol.~\bibinfo{volume}{35}. \bibinfo{pages}{5738--5750}.
\newblock
\urldef\tempurl%
\url{https://proceedings.neurips.cc/paper_files/paper/2022/file/25e92e33ac8c35fd49f394c37f21b6da-Paper-Conference.pdf}
\showURL{%
\tempurl}


\bibitem[Chen et~al\mbox{.}(2019)]%
        {chen2019proportionally}
\bibfield{author}{\bibinfo{person}{Xingyu Chen}, \bibinfo{person}{Brandon
  Fain}, \bibinfo{person}{Liang Lyu}, {and} \bibinfo{person}{Kamesh Munagala}.}
  \bibinfo{year}{2019}\natexlab{}.
\newblock \showarticletitle{Proportionally fair clustering}. In
  \bibinfo{booktitle}{\emph{Proceedings of the 36th International Conference on
  Machine Learning (ICML)}}. PMLR, \bibinfo{pages}{1032--1041}.
\newblock
\urldef\tempurl%
\url{https://proceedings.mlr.press/v97/chen19d.html}
\showURL{%
\tempurl}


\bibitem[Chen and Teng(2009)]%
        {chen2009spending}
\bibfield{author}{\bibinfo{person}{Xi Chen} {and} \bibinfo{person}{Shang-Hua
  Teng}.} \bibinfo{year}{2009}\natexlab{}.
\newblock \showarticletitle{Spending is not easier than trading: On the
  computational equivalence of {Fisher} and {Arrow}-{Debreu} equilibria}. In
  \bibinfo{booktitle}{\emph{Proceedings of the 20th International Symposium on
  Algorithms and Computation (ISAAC)}}. \bibinfo{pages}{647--656}.
\newblock
\href{https://doi.org/10.1007/978-3-642-10631-6_66}{doi:\nolinkurl{10.1007/978-3-642-10631-6_66}}


\bibitem[Cheng et~al\mbox{.}(2020)]%
        {cheng2020groupfairness}
\bibfield{author}{\bibinfo{person}{Yu Cheng}, \bibinfo{person}{Zhihao Jiang},
  \bibinfo{person}{Kamesh Munagala}, {and} \bibinfo{person}{Kangning Wang}.}
  \bibinfo{year}{2020}\natexlab{}.
\newblock \showarticletitle{Group fairness in committee selection}.
\newblock \bibinfo{journal}{\emph{ACM Transactions on Economics and Computation
  (TEAC)}} \bibinfo{volume}{8}, \bibinfo{number}{4}, Article
  \bibinfo{articleno}{23} (\bibinfo{year}{2020}), \bibinfo{numpages}{18}~pages.
\newblock
\href{https://doi.org/10.1145/3417750}{doi:\nolinkurl{10.1145/3417750}}


\bibitem[Cole et~al\mbox{.}(2017)]%
        {cole2017convex}
\bibfield{author}{\bibinfo{person}{Richard Cole}, \bibinfo{person}{Nikhil~R.
  Devanur}, \bibinfo{person}{Vasilis Gkatzelis}, \bibinfo{person}{Kamal Jain},
  \bibinfo{person}{Tung Mai}, \bibinfo{person}{Vijay~V. Vazirani}, {and}
  \bibinfo{person}{Sadra Yazdanbod}.} \bibinfo{year}{2017}\natexlab{}.
\newblock \showarticletitle{Convex program duality, {Fisher} markets, and
  {Nash} social welfare}. In \bibinfo{booktitle}{\emph{Proceedings of the 2017
  ACM Conference on Economics and Computation (EC)}}.
  \bibinfo{pages}{459--460}.
\newblock
\href{https://doi.org/10.1145/3033274.3085109}{doi:\nolinkurl{10.1145/3033274.3085109}}


\bibitem[Conitzer et~al\mbox{.}(2017)]%
        {conitzer2017fair}
\bibfield{author}{\bibinfo{person}{Vincent Conitzer}, \bibinfo{person}{Rupert
  Freeman}, {and} \bibinfo{person}{Nisarg Shah}.}
  \bibinfo{year}{2017}\natexlab{}.
\newblock \showarticletitle{Fair public decision making}. In
  \bibinfo{booktitle}{\emph{Proceedings of the 2017 ACM Conference on Economics
  and Computation (EC)}}. \bibinfo{pages}{629--646}.
\newblock
\href{https://doi.org/10.1145/3033274.3085125}{doi:\nolinkurl{10.1145/3033274.3085125}}


\bibitem[Corless et~al\mbox{.}(1996)]%
        {corless1996lambert}
\bibfield{author}{\bibinfo{person}{Robert~M. Corless},
  \bibinfo{person}{Gaston~H. Gonnet}, \bibinfo{person}{David E.~G. Hare},
  \bibinfo{person}{David~J. Jeffrey}, {and} \bibinfo{person}{Donald~E. Knuth}.}
  \bibinfo{year}{1996}\natexlab{}.
\newblock \showarticletitle{On the {Lambert} $W$ function}.
\newblock \bibinfo{journal}{\emph{Advances in Computational Mathematics}}
  \bibinfo{volume}{5}, \bibinfo{number}{1} (\bibinfo{year}{1996}), \bibinfo{pages}{329--359}.
\newblock
\href{https://doi.org/10.1007/BF02124750}{doi:\nolinkurl{10.1007/BF02124750}}


\bibitem[Cover(1984)]%
        {cover1984algorithm}
\bibfield{author}{\bibinfo{person}{Thomas Cover}.}
  \bibinfo{year}{1984}\natexlab{}.
\newblock \showarticletitle{An algorithm for maximizing expected log investment
  return}.
\newblock \bibinfo{journal}{\emph{IEEE Transactions on Information Theory}}
  \bibinfo{volume}{30}, \bibinfo{number}{2} (\bibinfo{year}{1984}),
  \bibinfo{pages}{369--373}.
\newblock
\href{https://doi.org/10.1109/TIT.1984.1056869}{doi:\nolinkurl{10.1109/TIT.1984.1056869}}


\bibitem[Cover and Thomas(2005)]%
        {coverjoy2005elements}
\bibfield{author}{\bibinfo{person}{Thomas~M. Cover} {and}
  \bibinfo{person}{Joy~A. Thomas}.} \bibinfo{year}{2005}\natexlab{}.
\newblock \bibinfo{booktitle}{\emph{Elements of Information Theory}}.
\newblock \bibinfo{publisher}{John Wiley \& Sons}.
\newblock
\href{https://doi.org/10.1002/047174882X}{doi:\nolinkurl{10.1002/047174882X}}


\bibitem[Csisz{\'a}r(1974)]%
        {csiszar1974computation}
\bibfield{author}{\bibinfo{person}{Imre Csisz{\'a}r}.}
  \bibinfo{year}{1974}\natexlab{}.
\newblock \showarticletitle{On the computation of rate-distortion functions
  (corresp.)}.
\newblock \bibinfo{journal}{\emph{IEEE Transactions on Information Theory}}
  \bibinfo{volume}{20}, \bibinfo{number}{1} (\bibinfo{year}{1974}),
  \bibinfo{pages}{122--124}.
\newblock
\href{https://doi.org/10.1109/TIT.1974.1055146}{doi:\nolinkurl{10.1109/TIT.1974.1055146}}


\bibitem[Csisz{\'a}r(1984)]%
        {csiszar1984information}
\bibfield{author}{\bibinfo{person}{Imre Csisz{\'a}r}.}
  \bibinfo{year}{1984}\natexlab{}.
\newblock \showarticletitle{Information geometry and alternating minimization
  procedures}.
\newblock \bibinfo{journal}{\emph{Statistics and Decisions}}
  \bibinfo{volume}{Supplemental Issue No. 1} (\bibinfo{year}{1984}),
  \bibinfo{pages}{205--237}.
\newblock
\urldef\tempurl%
\url{https://dominik-peters.de/archive/csiszar1984.pdf}
\showURL{%
\tempurl}


\bibitem[Deligkas et~al\mbox{.}(2024)]%
        {deligkas2024constant}
\bibfield{author}{\bibinfo{person}{Argyrios Deligkas}, \bibinfo{person}{John
  Fearnley}, \bibinfo{person}{Alexandros Hollender}, {and}
  \bibinfo{person}{Themistoklis Melissourgos}.}
  \bibinfo{year}{2024}\natexlab{}.
\newblock \showarticletitle{Constant inapproximability for {Fisher} markets}.
  In \bibinfo{booktitle}{\emph{Proceedings of the 25th ACM Conference on
  Economics and Computation (EC)}}. \bibinfo{pages}{13--39}.
\newblock
\href{https://doi.org/10.1145/3670865.3673533}{doi:\nolinkurl{10.1145/3670865.3673533}}


\bibitem[Devanur et~al\mbox{.}(2008)]%
        {devanur2008market}
\bibfield{author}{\bibinfo{person}{Nikhil~R. Devanur},
  \bibinfo{person}{Christos~H. Papadimitriou}, \bibinfo{person}{Amin Saberi},
  {and} \bibinfo{person}{Vijay~V. Vazirani}.} \bibinfo{year}{2008}\natexlab{}.
\newblock \showarticletitle{Market equilibrium via a primal--dual algorithm for
  a convex program}.
\newblock \bibinfo{journal}{\emph{Journal of the ACM (JACM)}}
  \bibinfo{volume}{55}, \bibinfo{number}{5} (\bibinfo{year}{2008}),
  \bibinfo{pages}{1--18}.
\newblock
\href{https://doi.org/10.1145/1411509.1411512}{doi:\nolinkurl{10.1145/1411509.1411512}}


\bibitem[Dhara and Dutta(2011)]%
        {dhara2011optimality}
\bibfield{author}{\bibinfo{person}{Anulekha Dhara} {and}
  \bibinfo{person}{Joydeep Dutta}.} \bibinfo{year}{2011}\natexlab{}.
\newblock \bibinfo{booktitle}{\emph{Optimality Conditions in Convex
  Optimization: A Finite-Dimensional View}}.
\newblock \bibinfo{publisher}{CRC Press}.
\newblock
\href{https://doi.org/10.1201/b11156}{doi:\nolinkurl{10.1201/b11156}}


\bibitem[Eaves(1971)]%
        {eaves1971basic}
\bibfield{author}{\bibinfo{person}{B.~Curtis Eaves}.}
  \bibinfo{year}{1971}\natexlab{}.
\newblock \showarticletitle{On the basic theorem of complementarity}.
\newblock \bibinfo{journal}{\emph{Mathematical Programming}}
  \bibinfo{volume}{1}, \bibinfo{number}{1} (\bibinfo{year}{1971}),
  \bibinfo{pages}{68--75}.
\newblock
\href{https://doi.org/10.1007/BF01584073}{doi:\nolinkurl{10.1007/BF01584073}}


\bibitem[Eaves(1976)]%
        {eaves1976finite}
\bibfield{author}{\bibinfo{person}{B.~Curtis Eaves}.}
  \bibinfo{year}{1976}\natexlab{}.
\newblock \showarticletitle{A finite algorithm for the linear exchange model}.
\newblock \bibinfo{journal}{\emph{Journal of Mathematical Economics}}
  \bibinfo{volume}{3}, \bibinfo{number}{2} (\bibinfo{year}{1976}),
  \bibinfo{pages}{197--203}.
\newblock
\href{https://doi.org/10.1016/0304-4068(76)90028-8}{doi:\nolinkurl{10.1016/0304-4068(76)90028-8}}


\bibitem[Ebadian et~al\mbox{.}(2024)]%
        {ebadian2024optimized}
\bibfield{author}{\bibinfo{person}{Soroush Ebadian}, \bibinfo{person}{Anson
  Kahng}, \bibinfo{person}{Dominik Peters}, {and} \bibinfo{person}{Nisarg
  Shah}.} \bibinfo{year}{2024}\natexlab{}.
\newblock \showarticletitle{Optimized distortion and proportional fairness in
  voting}.
\newblock \bibinfo{journal}{\emph{ACM Transactions on Economics and Computation
  (TEAC)}} \bibinfo{volume}{12}, \bibinfo{number}{1}, Article
  \bibinfo{articleno}{3} (\bibinfo{year}{2024}), \bibinfo{numpages}{39}~pages.
\newblock
\href{https://doi.org/10.1145/3640760}{doi:\nolinkurl{10.1145/3640760}}


\bibitem[Ebadian and Micha(2025)]%
        {ebadian2025sortition}
\bibfield{author}{\bibinfo{person}{Soroush Ebadian} {and} \bibinfo{person}{Evi
  Micha}.} \bibinfo{year}{2025}\natexlab{}.
\newblock \showarticletitle{Boosting sortition via proportional
  representation}. In \bibinfo{booktitle}{\emph{Proceedings of 24th
  International Conference on Autonomous Agents and Multiagent Systems
  (AAMAS)}}. \bibinfo{pages}{667--675}.
\newblock
\urldef\tempurl%
\url{https://www.ifaamas.org/Proceedings/aamas2025/pdfs/p667.pdf}
\showURL{%
\tempurl}


\bibitem[Eisenberg(1961)]%
        {eisenberg1961aggregation}
\bibfield{author}{\bibinfo{person}{Edmund Eisenberg}.}
  \bibinfo{year}{1961}\natexlab{}.
\newblock \showarticletitle{Aggregation of utility functions}.
\newblock \bibinfo{journal}{\emph{Management Science}} \bibinfo{volume}{7},
  \bibinfo{number}{4} (\bibinfo{year}{1961}), \bibinfo{pages}{337--350}.
\newblock
\href{https://doi.org/10.1287/mnsc.7.4.337}{doi:\nolinkurl{10.1287/mnsc.7.4.337}}


\bibitem[Eisenberg and Gale(1959)]%
        {eisenberg1959consensus}
\bibfield{author}{\bibinfo{person}{Edmund Eisenberg} {and}
  \bibinfo{person}{David Gale}.} \bibinfo{year}{1959}\natexlab{}.
\newblock \showarticletitle{Consensus of subjective probabilities: The
  pari-mutuel method}.
\newblock \bibinfo{journal}{\emph{The Annals of Mathematical Statistics}}
  \bibinfo{volume}{30}, \bibinfo{number}{1} (\bibinfo{year}{1959}),
  \bibinfo{pages}{165--168}.
\newblock
\href{https://doi.org/10.1214/aoms/1177706369}{doi:\nolinkurl{10.1214/aoms/1177706369}}


\bibitem[Fain et~al\mbox{.}(2016)]%
        {fain2016core}
\bibfield{author}{\bibinfo{person}{Brandon Fain}, \bibinfo{person}{Ashish
  Goel}, {and} \bibinfo{person}{Kamesh Munagala}.}
  \bibinfo{year}{2016}\natexlab{}.
\newblock \showarticletitle{The core of the participatory budgeting problem}.
  In \bibinfo{booktitle}{\emph{Proceedings of the 12th International Conference
  on Web and Internet Economics (WINE)}}. \bibinfo{pages}{384--399}.
\newblock
\href{https://doi.org/10.1007/978-3-662-54110-4_27}{doi:\nolinkurl{10.1007/978-3-662-54110-4_27}}


\bibitem[Fain et~al\mbox{.}(2018)]%
        {fain2018indivisiblepublic}
\bibfield{author}{\bibinfo{person}{Brandon Fain}, \bibinfo{person}{Kamesh
  Munagala}, {and} \bibinfo{person}{Nisarg Shah}.}
  \bibinfo{year}{2018}\natexlab{}.
\newblock \showarticletitle{Fair allocation of indivisible public goods}. In
  \bibinfo{booktitle}{\emph{Proceedings of the 2018 ACM Conference on Economics
  and Computation (EC)}}. \bibinfo{pages}{575--592}.
\newblock
\href{https://doi.org/10.1145/3219166.3219174}{doi:\nolinkurl{10.1145/3219166.3219174}}


\bibitem[Faliszewski et~al\mbox{.}(2023)]%
        {faliszewski2023pabulib}
\bibfield{author}{\bibinfo{person}{Piotr Faliszewski},
  \bibinfo{person}{Jaros{\l}aw Flis}, \bibinfo{person}{Dominik Peters},
  \bibinfo{person}{Grzegorz Pierczy{\'n}ski}, \bibinfo{person}{Piotr Skowron},
  \bibinfo{person}{Dariusz Stolicki}, \bibinfo{person}{Stanis{\l}aw Szufa},
  {and} \bibinfo{person}{Nimrod Talmon}.} \bibinfo{year}{2023}\natexlab{}.
\newblock \showarticletitle{Participatory budgeting: Data, tools, and
  analysis}. In \bibinfo{booktitle}{\emph{Proceedings of the 32nd International
  Joint Conference on Artificial Intelligence (IJCAI)}}.
  \bibinfo{pages}{2667--2674}.
\newblock
\href{https://doi.org/10.24963/ijcai.2023/297}{doi:\nolinkurl{10.24963/ijcai.2023/297}}


\bibitem[Foley(1970)]%
        {foley1970lindahl}
\bibfield{author}{\bibinfo{person}{Duncan~K. Foley}.}
  \bibinfo{year}{1970}\natexlab{}.
\newblock \showarticletitle{Lindahl's solution and the core of an economy with
  public goods}.
\newblock \bibinfo{journal}{\emph{Econometrica}}
  \bibinfo{volume}{38}, \bibinfo{number}{1} (\bibinfo{year}{1970}),
  \bibinfo{pages}{66--72}.
\newblock
\href{https://doi.org/10.2307/1909241}{doi:\nolinkurl{10.2307/1909241}}


\bibitem[Garg et~al\mbox{.}(2015)]%
        {garg2015complementary}
\bibfield{author}{\bibinfo{person}{Jugal Garg}, \bibinfo{person}{Ruta Mehta},
  \bibinfo{person}{Milind Sohoni}, {and} \bibinfo{person}{Vijay~V. Vazirani}.}
  \bibinfo{year}{2015}\natexlab{}.
\newblock \showarticletitle{A complementary pivot algorithm for market
  equilibrium under separable, piecewise-linear concave utilities}.
\newblock \bibinfo{journal}{\emph{SIAM J. Comput.}} \bibinfo{volume}{44},
  \bibinfo{number}{6} (\bibinfo{year}{2015}), \bibinfo{pages}{1820--1847}.
\newblock
\href{https://doi.org/10.1137/140971002}{doi:\nolinkurl{10.1137/140971002}}


\bibitem[Garg et~al\mbox{.}(2021)]%
        {garg2021markets}
\bibfield{author}{\bibinfo{person}{Nikhil Garg}, \bibinfo{person}{Ashish Goel},
  {and} \bibinfo{person}{Benjamin Plaut}.} \bibinfo{year}{2021}\natexlab{}.
\newblock \showarticletitle{Markets for public decision-making}.
\newblock \bibinfo{journal}{\emph{Social Choice and Welfare}}
  \bibinfo{volume}{56}, \bibinfo{number}{4} (\bibinfo{year}{2021}),
  \bibinfo{pages}{755--801}.
\newblock
\href{https://doi.org/10.1007/s00355-020-01298-4}{doi:\nolinkurl{10.1007/s00355-020-01298-4}}


\bibitem[Greaves and Cotton-Barratt(2024)]%
        {greaves2023bargaining}
\bibfield{author}{\bibinfo{person}{Hilary Greaves} {and} \bibinfo{person}{Owen
  Cotton-Barratt}.} \bibinfo{year}{2024}\natexlab{}.
\newblock \showarticletitle{A bargaining-theoretic approach to moral
  uncertainty}.
\newblock \bibinfo{journal}{\emph{Journal of Moral Philosophy}}
  \bibinfo{volume}{21}, \bibinfo{number}{Issue 1-2} (\bibinfo{year}{2024}),
  \bibinfo{pages}{127--169}.
\newblock
\href{https://doi.org/10.1163/17455243-20233810}{doi:\nolinkurl{10.1163/17455243-20233810}}


\bibitem[Guerdjikova and Nehring(2014)]%
        {GuNe14a}
\bibfield{author}{\bibinfo{person}{Ani Guerdjikova} {and}
  \bibinfo{person}{Klaus Nehring}.} \bibinfo{year}{2014}\natexlab{}.
\newblock \showarticletitle{Weighing experts, weighing sources: The diversity
  value}.
\newblock  (\bibinfo{year}{2014}).
\newblock
\urldef\tempurl%
\url{https://dominik-peters.de/archive/guerdjikova2014.pdf}
\showURL{%
\tempurl}
\newblock
\shownote{Working Paper}.


\bibitem[Gul and Pesendorfer(2026)]%
        {gul2020lindahlequilibriumcollectivechoice}
\bibfield{author}{\bibinfo{person}{Faruk Gul} {and} \bibinfo{person}{Wolfgang
  Pesendorfer}.} \bibinfo{year}{2026}\natexlab{}.
\newblock \showarticletitle{Pseudo Lindahl equilibrium as a collective choice
  rule}.
\newblock \bibinfo{journal}{\emph{The Review of Economic Studies}}
  \bibinfo{volume}{93}, \bibinfo{number}{2} (\bibinfo{year}{2026}),
  \bibinfo{pages}{938--967}.
\newblock
\href{https://doi.org/10.1093/restud/rdaf043}{doi:\nolinkurl{10.1093/restud/rdaf043}}
\showeprint[arxiv]{2008.09932}~[econ.TH]


\bibitem[Hiriart-Urruty(1982)]%
        {hiriarturruty1982}
\bibfield{author}{\bibinfo{person}{Jean-Baptiste Hiriart-Urruty}.}
  \bibinfo{year}{1982}\natexlab{}.
\newblock \showarticletitle{$\varepsilon$-subdifferential calculus}.
\newblock In \bibinfo{booktitle}{\emph{Convex Analysis and Optimization}},
  \bibfield{editor}{\bibinfo{person}{Jean-Pierre Aubin} {and}
  \bibinfo{person}{Richard~B. Vinter}} (Eds.). \bibinfo{series}{Research Notes
  in Mathematics}, Vol.~\bibinfo{volume}{57}. \bibinfo{publisher}{Pitman},
  \bibinfo{pages}{43--92}.
\newblock
\urldef\tempurl%
\url{https://dominik-peters.de/archive/hiriarturruty1982.pdf}
\showURL{%
\tempurl}


\bibitem[Hiriart-Urruty and Lemar{\'e}chal(2013)]%
        {hiriart2013convex}
\bibfield{author}{\bibinfo{person}{Jean-Baptiste Hiriart-Urruty} {and}
  \bibinfo{person}{Claude Lemar{\'e}chal}.} \bibinfo{year}{2013}\natexlab{}.
\newblock \bibinfo{booktitle}{\emph{Convex Analysis and Minimization Algorithms
  I: Fundamentals}}. \bibinfo{series}{Grundlehren der mathematischen
  Wissenschaften}, Vol.~\bibinfo{volume}{305}.
\newblock \bibinfo{publisher}{Springer}.
\newblock
\href{https://doi.org/10.1007/978-3-662-02796-7}{doi:\nolinkurl{10.1007/978-3-662-02796-7}}


\bibitem[Hylland(1980)]%
        {hylland1980lotteries}
\bibfield{author}{\bibinfo{person}{Aanund Hylland}.}
  \bibinfo{year}{1980}\natexlab{}.
\newblock \bibinfo{title}{Strategyproofness of voting procedures with lotteries
  as outcomes and infinite sets of strategies}.  (\bibinfo{year}{1980}).
\newblock
\urldef\tempurl%
\url{https://dominik-peters.de/archive/hylland1980.pdf}
\showURL{%
\tempurl}
\newblock
\shownote{Unpublished}.


\bibitem[Jiang et~al\mbox{.}(2020)]%
        {jiang2020approximatelystable}
\bibfield{author}{\bibinfo{person}{Zhihao Jiang}, \bibinfo{person}{Kamesh
  Munagala}, {and} \bibinfo{person}{Kangning Wang}.}
  \bibinfo{year}{2020}\natexlab{}.
\newblock \showarticletitle{Approximately stable committee selection}. In
  \bibinfo{booktitle}{\emph{Proceedings of the 52nd Annual ACM SIGACT Symposium
  on Theory of Computing (STOC)}}. \bibinfo{pages}{463--472}.
\newblock
\href{https://doi.org/10.1145/3357713.3384238}{doi:\nolinkurl{10.1145/3357713.3384238}}


\bibitem[Kaneko(1977)]%
        {kaneko1977ratio}
\bibfield{author}{\bibinfo{person}{Mamoru Kaneko}.}
  \bibinfo{year}{1977}\natexlab{}.
\newblock \showarticletitle{The ratio equilibrium and a voting game in a public
  goods economy}.
\newblock \bibinfo{journal}{\emph{Journal of Economic Theory}}
  \bibinfo{volume}{16}, \bibinfo{number}{2} (\bibinfo{year}{1977}),
  \bibinfo{pages}{123--136}.
\newblock
\href{https://doi.org/10.1016/0022-0531(77)90001-1}{doi:\nolinkurl{10.1016/0022-0531(77)90001-1}}


\bibitem[Kellerhals and Peters(2024)]%
        {kellerhals2024proportional}
\bibfield{author}{\bibinfo{person}{Leon Kellerhals} {and}
  \bibinfo{person}{Jannik Peters}.} \bibinfo{year}{2024}\natexlab{}.
\newblock \showarticletitle{Proportional fairness in clustering: A social
  choice perspective}. In \bibinfo{booktitle}{\emph{Advances in Neural
  Information Processing Systems}}, Vol.~\bibinfo{volume}{37}.
  \bibinfo{pages}{111299--111317}.
\newblock
\urldef\tempurl%
\url{https://proceedings.neurips.cc/paper_files/paper/2024/file/c981fd12b1d5703f19bd8289da9fc996-Paper-Conference.pdf}
\showURL{%
\tempurl}


\bibitem[Kelly et~al\mbox{.}(1998)]%
        {kelly1998rate}
\bibfield{author}{\bibinfo{person}{Frank~P. Kelly}, \bibinfo{person}{Aman~K.
  Maulloo}, {and} \bibinfo{person}{David K.~H. Tan}.}
  \bibinfo{year}{1998}\natexlab{}.
\newblock \showarticletitle{Rate control for communication networks: Shadow
  prices, proportional fairness and stability}.
\newblock \bibinfo{journal}{\emph{Journal of the Operational Research Society}}
  \bibinfo{volume}{49}, \bibinfo{number}{3} (\bibinfo{year}{1998}),
  \bibinfo{pages}{237--252}.
\newblock
\href{https://doi.org/10.1057/palgrave.jors.2600523}{doi:\nolinkurl{10.1057/palgrave.jors.2600523}}


\bibitem[Komiya(1988)]%
        {komiya1988}
\bibfield{author}{\bibinfo{person}{Hidetoshi Komiya}.}
  \bibinfo{year}{1988}\natexlab{}.
\newblock \showarticletitle{Elementary proof for {Sion's} minimax theorem}.
\newblock \bibinfo{journal}{\emph{Kodai Mathematical Journal}}
  \bibinfo{volume}{11}, \bibinfo{number}{1} (\bibinfo{year}{1988}),
  \bibinfo{pages}{5 -- 7}.
\newblock
\href{https://doi.org/10.2996/kmj/1138038812}{doi:\nolinkurl{10.2996/kmj/1138038812}}


\bibitem[Lackner and Skowron(2023)]%
        {lackner2023abc}
\bibfield{author}{\bibinfo{person}{Martin Lackner} {and} \bibinfo{person}{Piotr
  Skowron}.} \bibinfo{year}{2023}\natexlab{}.
\newblock \bibinfo{booktitle}{\emph{Multi-Winner Voting with Approval
  Preferences}}.
\newblock \bibinfo{publisher}{Springer}.
\newblock
\href{https://doi.org/10.1007/978-3-031-09016-5}{doi:\nolinkurl{10.1007/978-3-031-09016-5}}


\bibitem[Li et~al\mbox{.}(2018)]%
        {li2018generalconvergencemirror}
\bibfield{author}{\bibinfo{person}{Yen-Huan Li}, \bibinfo{person}{Carlos~A.
  Riofrio}, {and} \bibinfo{person}{Volkan Cevher}.}
  \bibinfo{year}{2018}\natexlab{}.
\newblock \bibinfo{title}{A general convergence result for mirror descent with
  {Armijo} line search}.
\newblock
\showeprint[arxiv]{1805.12232}~[math.OC]


\bibitem[Lindahl(1919)]%
        {lindahl1919just}
\bibfield{author}{\bibinfo{person}{Erik Lindahl}.}
  \bibinfo{year}{1919}\natexlab{}.
\newblock \showarticletitle{Just taxation---a positive solution. Translated
  from German (\textit{Die Gerechtigkeit der Besteuerung}, Lund 1919, Part I,
  Chap. 4, pp. 85--98: Positive Lösung) by Elizabeth Henderson}.
\newblock In \bibinfo{booktitle}{\emph{Classics in the Theory of Public
  Finance}}, \bibfield{editor}{\bibinfo{person}{Richard~A. Musgrave} {and}
  \bibinfo{person}{Alan~T. Peacock}} (Eds.). \bibinfo{publisher}{Palgrave
  Macmillan (1958)}, \bibinfo{pages}{168--176}.
\newblock
\href{https://doi.org/10.1007/978-1-349-23426-4_11}{doi:\nolinkurl{10.1007/978-1-349-23426-4_11}}


\bibitem[Martinez~Mori and Toriello(2025)]%
        {martinez2025cooperation}
\bibfield{author}{\bibinfo{person}{Juan~Carlos Martinez~Mori} {and}
  \bibinfo{person}{Alejandro Toriello}.} \bibinfo{year}{2025}\natexlab{}.
\newblock \showarticletitle{Cooperation and the design of public goods}. In
  \bibinfo{booktitle}{\emph{Proceedings of the 26th ACM Conference on Economics
  and Computation (EC)}}. \bibinfo{pages}{511}.
\newblock
\href{https://doi.org/10.1145/3736252.3742582}{doi:\nolinkurl{10.1145/3736252.3742582}}
\showeprint[arxiv]{2506.05251}~[cs.GT]


\bibitem[Mas-Colell and Silvestre(1989)]%
        {mas1989costshare}
\bibfield{author}{\bibinfo{person}{Andreu Mas-Colell} {and}
  \bibinfo{person}{Joaquim Silvestre}.} \bibinfo{year}{1989}\natexlab{}.
\newblock \showarticletitle{Cost share equilibria: A Lindahlian approach}.
\newblock \bibinfo{journal}{\emph{Journal of Economic Theory}}
  \bibinfo{volume}{47}, \bibinfo{number}{2} (\bibinfo{year}{1989}),
  \bibinfo{pages}{239--256}.
\newblock
\href{https://doi.org/10.1016/0022-0531(89)90019-7}{doi:\nolinkurl{10.1016/0022-0531(89)90019-7}}


\bibitem[Michorzewski et~al\mbox{.}(2020)]%
        {michorzewski2020price}
\bibfield{author}{\bibinfo{person}{Marcin Michorzewski},
  \bibinfo{person}{Dominik Peters}, {and} \bibinfo{person}{Piotr Skowron}.}
  \bibinfo{year}{2020}\natexlab{}.
\newblock \showarticletitle{Price of fairness in budget division and
  probabilistic social choice}. In \bibinfo{booktitle}{\emph{Proceedings of the
  34th AAAI Conference on Artificial Intelligence (AAAI)}}.
  \bibinfo{pages}{2184--2191}.
\newblock
\href{https://doi.org/10.1609/aaai.v34i02.5594}{doi:\nolinkurl{10.1609/aaai.v34i02.5594}}


\bibitem[Moore(2006)]%
        {moore2006generalequilibrium}
\bibfield{author}{\bibinfo{person}{James~C. Moore}.}
  \bibinfo{year}{2006}\natexlab{}.
\newblock \bibinfo{booktitle}{\emph{General Equilibrium and Welfare Economics:
  An Introduction}}.
\newblock \bibinfo{publisher}{Springer}.
\newblock
\href{https://doi.org/10.1007/978-3-540-32223-8}{doi:\nolinkurl{10.1007/978-3-540-32223-8}}


\bibitem[Moulin(2004)]%
        {moulin2004fair}
\bibfield{author}{\bibinfo{person}{Herv{\'e} Moulin}.}
  \bibinfo{year}{2004}\natexlab{}.
\newblock \bibinfo{booktitle}{\emph{Fair Division and Collective Welfare}}.
\newblock \bibinfo{publisher}{MIT Press}.
\newblock
\href{https://doi.org/10.7551/mitpress/2954.001.0001}{doi:\nolinkurl{10.7551/mitpress/2954.001.0001}}


\bibitem[Munagala et~al\mbox{.}(2022a)]%
        {munagala2022auditing}
\bibfield{author}{\bibinfo{person}{Kamesh Munagala}, \bibinfo{person}{Yiheng
  Shen}, {and} \bibinfo{person}{Kangning Wang}.}
  \bibinfo{year}{2022}\natexlab{a}.
\newblock \showarticletitle{Auditing for core stability in participatory
  budgeting}. In \bibinfo{booktitle}{\emph{Proceedings of the 18th
  International Conference on Web and Internet Economics (WINE)}}.
  \bibinfo{pages}{292--310}.
\newblock
\href{https://doi.org/10.1007/978-3-031-22832-2_17}{doi:\nolinkurl{10.1007/978-3-031-22832-2_17}}


\bibitem[Munagala et~al\mbox{.}(2022b)]%
        {munagala2022coremultilinear}
\bibfield{author}{\bibinfo{person}{Kamesh Munagala}, \bibinfo{person}{Yiheng
  Shen}, \bibinfo{person}{Kangning Wang}, {and} \bibinfo{person}{Zhiyi Wang}.}
  \bibinfo{year}{2022}\natexlab{b}.
\newblock \showarticletitle{Approximate core for committee selection via
  multilinear extension and market clearing}. In
  \bibinfo{booktitle}{\emph{Proceedings of the 2022 Annual ACM-SIAM Symposium
  on Discrete Algorithms (SODA)}}. \bibinfo{pages}{2229--2252}.
\newblock
\href{https://doi.org/10.1137/1.9781611977073.89}{doi:\nolinkurl{10.1137/1.9781611977073.89}}


\bibitem[Nash(1950)]%
        {nash1950bargaining}
\bibfield{author}{\bibinfo{person}{John~F. Nash}.}
  \bibinfo{year}{1950}\natexlab{}.
\newblock \showarticletitle{The bargaining problem}.
\newblock \bibinfo{journal}{\emph{Econometrica}} \bibinfo{volume}{18},
  \bibinfo{number}{2} (\bibinfo{year}{1950}), \bibinfo{pages}{155--162}.
\newblock
\href{https://doi.org/10.2307/1907266}{doi:\nolinkurl{10.2307/1907266}}


\bibitem[Nedi{\'c} and Ozdaglar(2009)]%
        {nedic2009approximate}
\bibfield{author}{\bibinfo{person}{Angelia Nedi{\'c}} {and}
  \bibinfo{person}{Asuman Ozdaglar}.} \bibinfo{year}{2009}\natexlab{}.
\newblock \showarticletitle{Approximate primal solutions and rate analysis for
  dual subgradient methods}.
\newblock \bibinfo{journal}{\emph{SIAM Journal on Optimization}}
  \bibinfo{volume}{19}, \bibinfo{number}{4} (\bibinfo{year}{2009}),
  \bibinfo{pages}{1757--1780}.
\newblock
\href{https://doi.org/10.1137/070708111}{doi:\nolinkurl{10.1137/070708111}}


\bibitem[O'Donoghue et~al\mbox{.}(2016)]%
        {donoghue2016conic}
\bibfield{author}{\bibinfo{person}{Brendan O'Donoghue}, \bibinfo{person}{Eric
  Chu}, \bibinfo{person}{Neal Parikh}, {and} \bibinfo{person}{Stephen Boyd}.}
  \bibinfo{year}{2016}\natexlab{}.
\newblock \showarticletitle{Conic optimization via operator splitting and
  homogeneous self-dual embedding}.
\newblock \bibinfo{journal}{\emph{Journal of Optimization Theory and
  Applications}} \bibinfo{volume}{169}, \bibinfo{number}{3}
  (\bibinfo{date}{June} \bibinfo{year}{2016}), \bibinfo{pages}{1042--1068}.
\newblock
\urldef\tempurl%
\url{http://stanford.edu/~boyd/papers/scs.html}
\showURL{%
\tempurl}


\bibitem[Peters(2019)]%
        {peters2019effectivealtruism}
\bibfield{author}{\bibinfo{person}{Dominik Peters}.}
  \bibinfo{year}{2019}\natexlab{}.
\newblock \showarticletitle{Economic design for effective altruism}.
\newblock In \bibinfo{booktitle}{\emph{The Future of Economic Design: The Continuing
  Development of a Field as Envisioned by Its Researchers}}.
  \bibinfo{publisher}{Springer}, \bibinfo{pages}{381--388}.
\newblock
\href{https://doi.org/10.1007/978-3-030-18050-8_53}{doi:\nolinkurl{10.1007/978-3-030-18050-8_53}}


\bibitem[Peters(2025)]%
        {peters2025core}
\bibfield{author}{\bibinfo{person}{Dominik Peters}.}
  \bibinfo{year}{2025}\natexlab{}.
\newblock \showarticletitle{The core of approval-based committee elections with
  few seats}. In \bibinfo{booktitle}{\emph{Proceedings of the 34th
  International Joint Conference on Artificial Intelligence (IJCAI)}}.
  \bibinfo{pages}{4014--4022}.
\newblock
\href{https://doi.org/10.24963/ijcai.2025/447}{doi:\nolinkurl{10.24963/ijcai.2025/447}}


\bibitem[Peters et~al\mbox{.}(2021b)]%
        {peters2021market}
\bibfield{author}{\bibinfo{person}{Dominik Peters}, \bibinfo{person}{Grzegorz
  Pierczy{\'n}ski}, \bibinfo{person}{Nisarg Shah}, {and} \bibinfo{person}{Piotr
  Skowron}.} \bibinfo{year}{2021}\natexlab{b}.
\newblock \showarticletitle{Market-based explanations of collective decisions}.
  In \bibinfo{booktitle}{\emph{Proceedings of the 35th AAAI Conference on
  Artificial Intelligence (AAAI)}}. \bibinfo{pages}{5656--5663}.
\newblock
\href{https://doi.org/10.1609/aaai.v35i6.16710}{doi:\nolinkurl{10.1609/aaai.v35i6.16710}}


\bibitem[Peters et~al\mbox{.}(2021a)]%
        {peters2021mes}
\bibfield{author}{\bibinfo{person}{Dominik Peters}, \bibinfo{person}{Grzegorz
  Pierczy\'{n}ski}, {and} \bibinfo{person}{Piotr Skowron}.}
  \bibinfo{year}{2021}\natexlab{a}.
\newblock \showarticletitle{Proportional participatory budgeting with additive
  utilities}. In \bibinfo{booktitle}{\emph{Advances in Neural Information
  Processing Systems}}, Vol.~\bibinfo{volume}{34}.
  \bibinfo{pages}{12726--12737}.
\newblock
\urldef\tempurl%
\url{https://proceedings.neurips.cc/paper_files/paper/2021/file/69f8ea31de0c00502b2ae571fbab1f95-Paper.pdf}
\showURL{%
\tempurl}


\bibitem[Peters and Skowron(2020)]%
        {peters2020welfarism}
\bibfield{author}{\bibinfo{person}{Dominik Peters} {and} \bibinfo{person}{Piotr
  Skowron}.} \bibinfo{year}{2020}\natexlab{}.
\newblock \showarticletitle{Proportionality and the limits of welfarism}. In
  \bibinfo{booktitle}{\emph{Proceedings of the 21st ACM Conference on Economics
  and Computation (EC)}}. \bibinfo{pages}{793--794}.
\newblock
\newblock
\shownote{Full version
  arXiv:\href{https://arxiv.org/abs/1911.11747}{1911.11747}}.


\bibitem[Pierczy{\'n}ski and Skowron(2022)]%
        {pierczynski2022core}
\bibfield{author}{\bibinfo{person}{Grzegorz Pierczy{\'n}ski} {and}
  \bibinfo{person}{Piotr Skowron}.} \bibinfo{year}{2022}\natexlab{}.
\newblock \showarticletitle{Core-stable committees under restricted domains}.
  In \bibinfo{booktitle}{\emph{Proceedings of the 18th International Conference
  on Web and Internet Economics (WINE)}}. \bibinfo{pages}{311--329}.
\newblock
\href{https://doi.org/10.1007/978-3-031-22832-2_18}{doi:\nolinkurl{10.1007/978-3-031-22832-2_18}}


\bibitem[Rey and Maly(2023)]%
        {rey2023comsocPBsurvey}
\bibfield{author}{\bibinfo{person}{Simon Rey} {and} \bibinfo{person}{Jan
  Maly}.} \bibinfo{year}{2023}\natexlab{}.
\newblock \bibinfo{title}{The (computational) social choice take on indivisible
  participatory budgeting}.
\newblock
\showeprint[arxiv]{2303.00621}~[cs.GT]


\bibitem[Rockafellar(1970)]%
        {rockafellar1970}
\bibfield{author}{\bibinfo{person}{Ralph~Tyrell Rockafellar}.}
  \bibinfo{year}{1970}\natexlab{}.
\newblock \bibinfo{booktitle}{\emph{Convex Analysis}}.
  \bibinfo{series}{Princeton Mathematical Series}, Vol.~\bibinfo{volume}{28}.
\newblock \bibinfo{publisher}{Princeton University Press}.
\newblock
\href{https://doi.org/10.1515/9781400873173}{doi:\nolinkurl{10.1515/9781400873173}}


\bibitem[Ruszczynski(2011)]%
        {ruszczynski2011nonlinear}
\bibfield{author}{\bibinfo{person}{Andrzej Ruszczynski}.}
  \bibinfo{year}{2011}\natexlab{}.
\newblock \bibinfo{booktitle}{\emph{Nonlinear Optimization}}.
\newblock \bibinfo{publisher}{Princeton University Press}.
\newblock
\href{https://doi.org/10.2307/j.ctvcm4hcj}{doi:\nolinkurl{10.2307/j.ctvcm4hcj}}


\bibitem[Samuelson(1954)]%
        {samuelson1954pure}
\bibfield{author}{\bibinfo{person}{Paul~A. Samuelson}.}
  \bibinfo{year}{1954}\natexlab{}.
\newblock \showarticletitle{The pure theory of public expenditure}.
\newblock \bibinfo{journal}{\emph{The Review of Economics and Statistics}}
  \bibinfo{volume}{36}, \bibinfo{number}{4} (\bibinfo{year}{1954}),
  \bibinfo{pages}{387--389}.
\newblock
\href{https://doi.org/10.2307/1925895}{doi:\nolinkurl{10.2307/1925895}}


\bibitem[Shmyrev(1983)]%
        {shmyrev1983approach}
\bibfield{author}{\bibinfo{person}{Vadim~I. Shmyrev}.}
  \bibinfo{year}{1983}\natexlab{}.
\newblock \showarticletitle{On an approach to the determination of equilibrium
  in elementary exchange models}.
\newblock \bibinfo{journal}{\emph{Doklady Akademii Nauk SSSR}}
  \bibinfo{volume}{268}, \bibinfo{number}{5} (\bibinfo{year}{1983}),
  \bibinfo{pages}{1062--1066}.
\newblock
\urldef\tempurl%
\url{https://www.mathnet.ru/eng/dan10141}
\showURL{%
\tempurl}
\newblock
\shownote{(In Russian)}.


\bibitem[Shmyrev(2009)]%
        {shmyrev2009algorithm}
\bibfield{author}{\bibinfo{person}{Vadim~I. Shmyrev}.}
  \bibinfo{year}{2009}\natexlab{}.
\newblock \showarticletitle{An algorithm for finding equilibrium in the linear
  exchange model with fixed budgets}.
\newblock \bibinfo{journal}{\emph{Journal of Applied and Industrial
  Mathematics}} \bibinfo{volume}{3}, \bibinfo{number}{4}
  (\bibinfo{year}{2009}), \bibinfo{pages}{505}.
\newblock
\href{https://doi.org/10.1134/S1990478909040097}{doi:\nolinkurl{10.1134/S1990478909040097}}


\bibitem[Strodiot et~al\mbox{.}(1983)]%
        {strodiot1983epsilon}
\bibfield{author}{\bibinfo{person}{Jean-Jacques Strodiot},
  \bibinfo{person}{V.~Hien Nguyen}, {and} \bibinfo{person}{Norbert Heukemes}.}
  \bibinfo{year}{1983}\natexlab{}.
\newblock \showarticletitle{$\varepsilon$-optimal solutions in
  nondifferentiable convex programming and some related questions}.
\newblock \bibinfo{journal}{\emph{Mathematical Programming}}
  \bibinfo{volume}{25}, \bibinfo{number}{3} (\bibinfo{year}{1983}), \bibinfo{pages}{307--328}.
\newblock
\href{https://doi.org/10.1007/BF02594782}{doi:\nolinkurl{10.1007/BF02594782}}


\bibitem[Suzuki and Vollen(2024)]%
        {suzuki2024maxflow}
\bibfield{author}{\bibinfo{person}{Mashbat Suzuki} {and}
  \bibinfo{person}{Jeremy Vollen}.} \bibinfo{year}{2024}\natexlab{}.
\newblock \showarticletitle{Maximum flow is fair: A network flow approach to
  committee voting}. In \bibinfo{booktitle}{\emph{Proceedings of the 25th ACM
  Conference on Economics and Computation (EC)}}. \bibinfo{pages}{964--983}.
\newblock
\href{https://doi.org/10.1145/3670865.3673603}{doi:\nolinkurl{10.1145/3670865.3673603}}


\bibitem[Uzawa(1958)]%
        {uzawa1958iterative}
\bibfield{author}{\bibinfo{person}{Hirofumi Uzawa}.}
  \bibinfo{year}{1958}\natexlab{}.
\newblock \showarticletitle{Iterative methods for concave programming}.
\newblock In \bibinfo{booktitle}{\emph{Studies in Linear and Non-Linear
  Programming}}, \bibfield{editor}{\bibinfo{person}{K.~Arrow},
  \bibinfo{person}{L.~Hurwicz}, {and} \bibinfo{person}{H.~Uzawa}} (Eds.).
  \bibinfo{publisher}{Stanford University Press}, \bibinfo{pages}{154--165}.
\newblock
\href{https://doi.org/10.1017/CBO9780511664496.011}{doi:\nolinkurl{10.1017/CBO9780511664496.011}}
\newblock
\shownote{(reprint) \url{https://dominik-peters.de/archive/uzawa1958.pdf}}.


\bibitem[van~den Nouweland(2015)]%
        {vandennouweland2015}
\bibfield{author}{\bibinfo{person}{Anne van~den Nouweland}.}
  \bibinfo{year}{2015}\natexlab{}.
\newblock \showarticletitle{Lindahl and equilibrium}.
\newblock In \bibinfo{booktitle}{\emph{Individual and Collective Choice and
  Social Welfare: Essays in Honor of Nick Baigent}},
  \bibfield{editor}{\bibinfo{person}{Constanze Binder}, \bibinfo{person}{Giulio
  Codognato}, \bibinfo{person}{Miriam Teschl}, {and} \bibinfo{person}{Yongsheng
  Xu}} (Eds.). \bibinfo{publisher}{Springer}, \bibinfo{pages}{335--362}.
\newblock
\href{https://doi.org/10.1007/978-3-662-46439-7_18}{doi:\nolinkurl{10.1007/978-3-662-46439-7_18}}


\bibitem[Vardi and Lee(1993)]%
        {vardi1993image}
\bibfield{author}{\bibinfo{person}{Yehuda Vardi} {and} \bibinfo{person}{D.
  Lee}.} \bibinfo{year}{1993}\natexlab{}.
\newblock \showarticletitle{From image deblurring to optimal investments:
  Maximum likelihood solutions for positive linear inverse problems}.
\newblock \bibinfo{journal}{\emph{Journal of the Royal Statistical Society
  Series B: Statistical Methodology}} \bibinfo{volume}{55}, \bibinfo{number}{3}
  (\bibinfo{year}{1993}), \bibinfo{pages}{569--598}.
\newblock
\href{https://doi.org/10.1111/j.2517-6161.1993.tb01925.x}{doi:\nolinkurl{10.1111/j.2517-6161.1993.tb01925.x}}


\bibitem[Vardi et~al\mbox{.}(1985)]%
        {vardi1985statistical}
\bibfield{author}{\bibinfo{person}{Yehuda Vardi}, \bibinfo{person}{Larry~A.
  Shepp}, {and} \bibinfo{person}{Linda Kaufman}.}
  \bibinfo{year}{1985}\natexlab{}.
\newblock \showarticletitle{A statistical model for positron emission
  tomography}.
\newblock \bibinfo{journal}{\emph{Journal of the American statistical
  Association}} \bibinfo{volume}{80}, \bibinfo{number}{389}
  (\bibinfo{year}{1985}), \bibinfo{pages}{8--20}.
\newblock
\href{https://doi.org/10.1080/01621459.1985.10477119}{doi:\nolinkurl{10.1080/01621459.1985.10477119}}


\bibitem[Vazirani(2012)]%
        {vazirani2012rationalconvexprogram}
\bibfield{author}{\bibinfo{person}{Vijay~V. Vazirani}.}
  \bibinfo{year}{2012}\natexlab{}.
\newblock \showarticletitle{The notion of a rational convex program, and an
  algorithm for the {Arrow}-{Debreu} {Nash} bargaining game}.
\newblock \bibinfo{journal}{\emph{Journal of the ACM (JACM)}}
  \bibinfo{volume}{59}, \bibinfo{number}{2} (\bibinfo{year}{2012}),
  \bibinfo{pages}{1--36}.
\newblock
\href{https://doi.org/10.1145/2160158.2160160}{doi:\nolinkurl{10.1145/2160158.2160160}}


\bibitem[Vazirani and Yannakakis(2011)]%
        {vazirani2011separablePLC}
\bibfield{author}{\bibinfo{person}{Vijay~V. Vazirani} {and}
  \bibinfo{person}{Mihalis Yannakakis}.} \bibinfo{year}{2011}\natexlab{}.
\newblock \showarticletitle{Market equilibrium under separable,
  piecewise-linear, concave utilities}.
\newblock \bibinfo{journal}{\emph{Journal of the ACM (JACM)}}
  \bibinfo{volume}{58}, \bibinfo{number}{3} (\bibinfo{year}{2011}),
  \bibinfo{pages}{1--25}.
\newblock
\href{https://doi.org/10.1145/1970392.1970394}{doi:\nolinkurl{10.1145/1970392.1970394}}


\bibitem[Vishnoi(2021)]%
        {vishnoi2021algorithms}
\bibfield{author}{\bibinfo{person}{Nisheeth~K. Vishnoi}.}
  \bibinfo{year}{2021}\natexlab{}.
\newblock \bibinfo{booktitle}{\emph{Algorithms for Convex Optimization}}.
\newblock \bibinfo{publisher}{Cambridge University Press}.
\newblock
\href{https://doi.org/10.1017/9781108699211}{doi:\nolinkurl{10.1017/9781108699211}}
\newblock
\shownote{\url{https://convex-optimization.github.io/}}.


\bibitem[Wu and Zhang(2007)]%
        {wu2007proportional}
\bibfield{author}{\bibinfo{person}{Fang Wu} {and} \bibinfo{person}{Li Zhang}.}
  \bibinfo{year}{2007}\natexlab{}.
\newblock \showarticletitle{Proportional response dynamics leads to market
  equilibrium}. In \bibinfo{booktitle}{\emph{Proceedings of the 39th Annual ACM
  Symposium on Theory of Computing (STOC)}}. \bibinfo{pages}{354--363}.
\newblock
\href{https://doi.org/10.1145/1250790.1250844}{doi:\nolinkurl{10.1145/1250790.1250844}}


\bibitem[Zhang(2011)]%
        {zhang2011proportional}
\bibfield{author}{\bibinfo{person}{Li Zhang}.} \bibinfo{year}{2011}\natexlab{}.
\newblock \showarticletitle{Proportional response dynamics in the {Fisher}
  market}.
\newblock \bibinfo{journal}{\emph{Theoretical Computer Science}}
  \bibinfo{volume}{412}, \bibinfo{number}{24} (\bibinfo{year}{2011}),
  \bibinfo{pages}{2691--2698}.
\newblock
\href{https://doi.org/10.1016/j.tcs.2010.06.021}{doi:\nolinkurl{10.1016/j.tcs.2010.06.021}}


\bibitem[Zhao(2023)]%
        {zhao2023convergence}
\bibfield{author}{\bibinfo{person}{Renbo Zhao}.}
  \bibinfo{year}{2023}\natexlab{}.
\newblock \showarticletitle{Convergence rate analysis of the multiplicative
  gradient method for {PET}-type problems}.
\newblock \bibinfo{journal}{\emph{Operations Research Letters}}
  \bibinfo{volume}{51}, \bibinfo{number}{1} (\bibinfo{year}{2023}),
  \bibinfo{pages}{26--32}.
\newblock
\href{https://doi.org/10.1016/j.orl.2022.11.010}{doi:\nolinkurl{10.1016/j.orl.2022.11.010}}


\bibitem[Zhao and Freund(2023)]%
        {zhao2023analysis}
\bibfield{author}{\bibinfo{person}{Renbo Zhao} {and} \bibinfo{person}{Robert~M.
  Freund}.} \bibinfo{year}{2023}\natexlab{}.
\newblock \showarticletitle{Analysis of the {Frank}--{Wolfe} method for convex
  composite optimization involving a logarithmically-homogeneous barrier}.
\newblock \bibinfo{journal}{\emph{Mathematical Programming}}
  \bibinfo{volume}{199}, \bibinfo{number}{1} (\bibinfo{year}{2023}),
  \bibinfo{pages}{123--163}.
\newblock
\href{https://doi.org/10.1007/s10107-022-01820-9}{doi:\nolinkurl{10.1007/s10107-022-01820-9}}


\end{thebibliography}
\end{document}